\documentclass{lmcs}
\pdfoutput=1

\usepackage{silence}
\usepackage{lastpage}
\lmcsdoi{15}{3}{16}
\lmcsheading{}{\pageref{LastPage}}{}{}%
{May~16,~2018}{Aug.~13,~2019}{}


\usepackage{amsmath,amsthm}

\newcommand{\proofideas}{\smallskip\noindent{\emph{The key ideas.}}}

\usepackage{amsthm}
\usepackage{amsfonts,graphicx,url,stfloats,subfig}
\usepackage{amssymb,amsmath,color,verbatim,bm}
\usepackage{stmaryrd,multirow}
\usepackage{environ,xspace}
\usepackage{thmtools,thm-restate,wrapfig}
\usepackage{tikz}
\usetikzlibrary{arrows,automata,shapes}

\usetikzlibrary{calc}

\newcommand{\Paragraph}[1]{\noindent{\textbf{#1}}}

\newcommand{\masterA}{\mathcal{A}_{\textrm{mas}}}
\newcommand{\nestedA}{\mathbb{A}}
\newcommand{\slaveA}{{\mathfrak{B}}}
\newcommand{\nonnestedA}{\mathcal{A}}

\newcommand{\masterStates}{Q_{\textrm{mas}}}

\newcommand{\silent}[1]{\mathsf{sil}({#1})} 

\newcommand{\cost}{{C}}

\newcommand{\masterRun}{\Pi}
\newcommand{\slaveRun}{\pi}
\newcommand{\lang}{\mathcal{L}}

\newcommand{\valueL}[1]{\mathcal{L}_{{#1}}}

\newcommand{\abs}{\mathop{\mathsf{Abs}}}

\newcommand{\EXPTIME}{\textsc{ExpTime}{}}
\newcommand{\PTIME}{\textsc{PTime}{}}

\newcommand{\PSPACE}{\textsc{PSpace}{}}
\newcommand{\EXPSPACE}{\textsc{ExpSpace}{}}
\newcommand{\NLOGSPACE}{\textsc{NLogSpace}{}}

\newcommand{\N}{\mathbb{N}}
\newcommand{\Z}{\mathbb{Z}}
\newcommand{\R}{\mathbb{R}}
\newcommand{\Q}{\mathbb{Q}}

\newcommand{\MCfromNested}{\markov^{\nestedA}}
\newcommand{\tuple}[1]{\langle#1\rangle}

\newcommand{\buchi}{B\"{u}chi}

\newcommand{\fsum}{\textsc{Sum}}
\newcommand{\fBsum}[1]{\textsc{Sum}^{#1}}
\newcommand{\fmax}{\textsc{Max}}
\newcommand{\fmin}{\textsc{Min}}
\newcommand{\flimavg}{\textsc{LimAvg}}
\newcommand{\fliminf}{\textsc{LimInf}}
\newcommand{\flimsup}{\textsc{LimSup}}
\newcommand{\fsup}{\textsc{Sup}}
\newcommand{\finf}{\textsc{Inf}}

\newcommand{\aut}{\mathcal{A}}

\newcommand{\Acc}{\mathsf{Acc}}
\newcommand{\const}{\lambda}
\newcommand{\FinVal}{\mathsf{FinVal}}
\newcommand{\InfVal}{\mathsf{InfVal}}

\newcommand{\uncomp}{Uncomputable}

\newcommand{\probability}{\mathbb{P}}
\newcommand{\SAT}{\mathsf{SAT}}
\newcommand{\expected}{\mathbb{E}}
\newcommand{\distrib}{\mathbb{D}}

\newcommand{\markov}{\mathcal{M}}
\newcommand{\run}{\pi}
\newcommand{\weightedRun}{\pi^W}
\newcommand{\calU}{\mathcal{U}}

\newcommand{\M}{\mathbf{M}}

 \usepackage{hyperref}

\newcommand{\lopen}[1]{(#1]} 

\tikzstyle{state}=[draw,circle,minimum size=0.8cm]

\begin{document}

\title{Quantitative Automata under Probabilistic Semantics}

\titlecomment{This is a combined version of~\cite{conferenceVersion} and~\cite{DBLP:conf/sas/ChatterjeeHO16}.}

\author[K.~Chatterjee]{Krishnendu Chatterjee\rsuper{a}}
\author[T.A.~Henzinger]{Thomas A. Henzinger\rsuper{a}}
\address{\lsuper{a}IST Austria}
\email{krish.chat@gmail.com, tah@ist.ac.at}
\author[J.~Otop]{Jan Otop\rsuper{b}}
\address{\lsuper{b}University of Wroc\l{}aw}
\email{jan.otop@uwr.edu.pl}

\maketitle

\begin{abstract}
Automata with monitor counters, where the transitions do not
depend on counter values, and nested weighted automata are
two expressive automata-theoretic frameworks for quantitative
properties.
For a well-studied and wide class of quantitative functions,
we establish that automata with monitor counters and nested weighted
automata are equivalent.
We study for the first time such quantitative automata under
probabilistic semantics.
We show that several problems that are undecidable for the
classical questions of emptiness and universality become
decidable under the probabilistic semantics.
We present a complete picture of decidability for such automata,
and even an almost-complete picture of computational complexity,
for the probabilistic questions we consider.

 \end{abstract}

\section{Introduction}

\noindent{\em Traditional to quantitative verification.}
While traditional formal verification focused on Boolean properties
of systems, such as ``every request is eventually granted'',
recently significant attention has been shifted to quantitative aspects
such as expressing properties like ``the long-run average success rate
of an operation is at least one half'' or ``the long-run average
(or the maximal, or the accumulated) resource consumption is below a threshold.''
Quantitative properties are essential for performance related properties,
for resource-constrained systems, such as embedded systems.

\smallskip\noindent{\em Overview.}
The first natural way to express quantitative properties is to
consider automata with counters.
However, computational analysis of such models quickly leads to
undecidability, and a classical way to limit expressiveness for
decidability is to consider {\em monitor counters}, i.e.,
the counter values do not influence the control.
The second approach is to consider automata with weights
(or weighted automata).
However, weighted automata have limited expressiveness, and
they have been extended as nested weighted automata~\cite{nested}
(nesting of weighted automata) for expressiveness.
We establish that for a well-studied and wide class of quantitative
functions, automata with monitor counters and nested weighted
automata are equivalent, i.e., they represent a robust class of
quantitative specifications.
We study for the first time such quantitative automata under
probabilistic semantics.
Quite surprisingly we show that several problems that are undecidable
for the classical questions of emptiness and universality become
decidable under the probabilistic semantics.
We present a complete picture of decidability for nested weighted
automata and automata with monitor counters under probabilistic semantics.

\smallskip\noindent{\em Automata with monitor counters.}
Automata with monitor counters are natural extension of weighted automata, where
automata are equipped with integer-valued counters.
At each transition, a counter can be started, terminated, or the value
of the counter can be increased or decreased.
However, the transitions
do not depend on the counter values, and hence they are referred to as
monitor counters.
The values of the counters when they are terminated give rise to
a sequence of \emph{weights}.
A value function aggregates the sequence into a single value.
For example, for words over $\{a,\#\}$, such automata can express
the maximal length of block of $a$'s that appear infinitely often.
Automata with monitor counters are similar in spirit with the class
of \emph{cost register automata}~\cite{DBLP:conf/lics/AlurDDRY13}, and
we consider them over infinite words.

\smallskip\noindent{\em Weighted automata.}
Weighted automata extend finite automata where every transition is
assigned an integer called a weight.
Hence every run gives rise to a sequence of weights, which is
aggregated into a single value by a value function.
For non-deterministic weighted automata, the value of a word
$w$ is the infimum value of all runs over~$w$.
Weighted automata provide a natural and flexible framework for
expressing quantitative\footnote{We use the term ``quantitative'' in a
non-probabilistic sense, which assigns a quantitative value to each
infinite run of a system, representing long-run average or maximal
response time, or power consumption, or the like, rather than taking a
probabilistic average over different runs.}
properties~\cite{Chatterjee08quantitativelanguages}.
First, weighted automata were studied over finite words with weights
from a semiring, and ring multiplication as a value function~\cite{Droste:2009:HWA:1667106},
and later extended to infinite words with limit averaging or supremum as
value functions~\cite{Chatterjee08quantitativelanguages,DBLP:journals/corr/abs-1007-4018,Chatterjee:2009:AWA:1789494.1789497}.
While weighted automata over semirings can express several
quantitative properties~\cite{DBLP:journals/jalc/Mohri02}, they cannot
express long-run average properties that weighted automata with limit
averaging can~\cite{Chatterjee08quantitativelanguages}.
However, even weighted automata with limit averaging cannot express
the following basic quantitative property (the example is from~\cite{nested}).

\begin{exa}\label{ex:intro}
Consider infinite words over $\{r,g,i\}$, where $r$ represents
requests, $g$ represents grants, and $i$ represents idle. A
basic and interesting property is the average number of $i$'s and $r$'s
between a request and the corresponding grant, which represents the
long-run average response time of the system.
\end{exa}

\smallskip\noindent{\em Nested weighted automata.}
To enrich expressiveness, weighted automata were extended
to \emph{nested weighted automata (NWA)}~\cite{nested}.
A nested weighted automaton consists of a master automaton and a set
of slave automata. The master automaton runs over infinite input words.
At every transition the master automaton invokes a slave automaton that runs
over a finite subword of the infinite word, starting at the position where
the slave automaton is invoked.
Each slave automaton terminates after a finite number of steps and returns
a value to the master automaton.
Slave automata are equipped with finite-word value functions to compute the returned values, which are then
aggregated by the master automaton using an infinite-word value function.
For Boolean finite automata, nested automata are equivalent to the non-nested
counterpart, whereas nested weighted automata are strictly more expressive
than non-nested weighted automata~\cite{nested}, for example,
nested weighted automata can express the long-run average response
time property (see~\cite[Example~5]{nested}).
It has been shown in~\cite{nested} that nested weighted automata provide a
specification framework where many basic quantitative properties,
which cannot be expressed by weighted automata, can be expressed easily,
and they provide a natural framework to study quantitative run-time
verification.

\smallskip\noindent\emph{Classical questions.}
Classical questions for automata are \emph{emptiness} and \emph{universality}
that ask for the existence and respectively non-existence of words that are accepted.
Their natural extensions have been studied in the quantitative setting as
well (such as for weighted automata and NWA)~\cite{Chatterjee08quantitativelanguages,nested}.

\smallskip\noindent{\em Motivation for probabilistic questions.}
One of the key reasons for quantitative specifications is to express performance
related properties.
While the classical emptiness and universality questions express the best-case/worst-case scenarios (such as the best-case/worst-case trace of a system for average response time),
they cannot express the average case average response time, where
the average case corresponds to the expected value over all traces.
Performance related properties are of prime interest for probabilistic systems,
and quite surprisingly, quantitative automata have not been studied in a
probabilistic setting, which we consider in this work.

\smallskip\noindent{\em Probabilistic questions.}
Weighted automata and their extensions as nested weighted automata, or automata
with monitor counters represent measurable functions from infinite words to
real numbers.
We consider probability distribution over infinite words, and as a finite
representation for probability spaces we consider the classical model of
finite-state Markov chains.
A stochastic environment is often modeled as a Markov chain~\cite{probabilisticMeasuriung}.
Hence, the theoretical problems we consider correspond to measuring performance (expectation or cumulative distribution) under such stochastic environments, when the specification is a nested weighted automaton.
Moreover, Markov chains are a canonical model for probabilistic systems~\cite{PRISM,BaierBook}.
Given a measurable function (or equivalently a random variable), the classical
quantities w.r.t.\ a probability distribution are: (a)~the expected value; and
(b)~the cumulative distribution below a threshold.
We consider the computation of the above quantities when the function is given
by a nested weighted automaton or an automaton with monitor counters, and the
probability distribution is given by a finite-state Markov chain.
We also consider the approximate variants that ask to approximate the above quantities
within a tolerance term $\epsilon>0$.
Moreover, for the cumulative distribution we consider the special case of
\emph{almost-sure} distribution, which asks whether the probability in the distribution question is exactly~$1$.

\smallskip\noindent{\em Our contributions.}
In this work we consider several classical value functions,
namely, $\fsup$, $\finf$, $\flimsup$, $\fliminf$, $\flimavg$ for infinite words,
and $\fmax$, $\fmin$, $\fsum$, $\fBsum{B}$, $\fsum^+$ (where $\fBsum{B}$ is the sum bounded by $B$,
and $\fsum^+$ is the sum of absolute values)
for finite words.
First, we establish translations (in both directions) between automata
with monitor counters and a subclass of nested weighted automata, called bounded-width nested weighted automata~\cite{nwa-mfcs},
where at any point only a bounded number of slave automata can be active.
However, in general, in nested weighted automata unbounded number of slave
automata can be active.
We describe our main results for nested weighted automata.
\begin{itemize}
\item {\em $\flimsup$ and $\fliminf$ functions.}
We consider deterministic nested weighted automata with $\flimsup$ and $\fliminf$
functions for the master automaton, and show that for all value functions for
finite words that we consider, all probabilistic questions can be answered
in polynomial time.
This is in contrast with the classical questions, where the problems are
$\PSPACE$-complete or undecidable (see Remark~\ref{remark:LimInf-classical-vs-probabilistic} for further details).

\item {\em $\flimavg$ function.}
We consider deterministic nested weighted automata with $\flimavg$
function for the master automaton, and show that for all value functions
for finite words that we consider, all probabilistic questions can be answered
in polynomial time.
Again our results are in contrast to the classical questions (see Remark~\ref{remark:LimAvg-classical-vs-probabilistic}).

\item {\em $\finf$ and $\fsup$ functions.}
We consider deterministic nested weighted automata with $\fsup$ and $\finf$
functions for the master automaton, and show the following:
the approximation problems for all value functions for finite words that we
consider are $\#P$-hard and can be computed in exponential time;
other than the $\fsum$ function, the expected value, the distribution, and the
almost-sure problems are $\PSPACE$-hard and can be solved in $\EXPTIME$;
and for the $\fsum$ function, the above problems are uncomputable.
Again we establish a sharp contrast w.r.t.\ the classical questions
as follows: for the classical questions, the complexity of $\flimsup$ and $\fsup$
functions always coincide, whereas we show a substantial complexity gap for
probabilistic questions (see Remark~\ref{remark:Inf-classical-vs-probabilistic} and Remark~\ref{remark:LimInf-vs-Inf} for
further details).

\item {\em Non-deterministic automata.}
For non-deterministic automata we show two results: first we present an
example to illustrate the conceptual difficulty of evaluating a non-deterministic
(even non-nested) weighted automaton with respect to a Markov chain, and also show that
for nested weighted automata with $\flimsup$ value function for
the master automaton and $\fsum$ value function for slave automata,
all probabilistic questions are undecidable (in contrast to the
deterministic case where we present polynomial-time algorithms).
\end{itemize}

\noindent
Note that from above all decidability results we establish carry over to
automata with monitor counters, and we show that all our undecidability
(or uncomputability) results also hold for automata with monitor counters.
Decidability results for nested weighted automata are more interesting as
compared to automata with monitor counters because in NWA unbounded number of slaves can be active.
Our results are summarized in Theorem~\ref{th:compLimInf} (in Section~\ref{s:liminf}),
Table~\ref{tab:compInf} (in Section~\ref{s:inf}), and Theorem~\ref{th:compLimAvg} (in Section~\ref{s:limavg}).
In summary, we present a complete picture of decidability of the basic
probabilistic questions for nested weighted automata (and automata with
monitor counters).

\smallskip\noindent{\em Technical contributions.}
We call a nested weighted automaton $\nestedA$, an \emph{$(f;g)$-automaton} if its master-automaton
value function is $f$ and the value function of all slave automata is $g$.
We present the key details of our main technical contributions, and for sake of simplicity here explain for the case of
the uniform distribution over infinite words.
Our technical results are more general though (for distributions given by Markov chains).

\begin{itemize}
\item We show that for a deterministic $(\fliminf;\fsum)$-automaton $\nestedA$, whose master automaton is strongly connected as a graph,
almost all words have the same value which, is the infimum over values of any slave automaton from $\nestedA$ over all finite words.

\item We show that the expected value of a deterministic $(\flimavg;\fsum)$-automaton $\nestedA$
coincides with the expected value of the following deterministic (non-nested)
$\flimavg$-automaton $\nonnestedA$.
The automaton $\nonnestedA$ is obtained from $\nestedA$ by replacing in every transition
an invocation of a slave automaton $\slaveA$ by the weight equal to the expected value of $\slaveA$.

\item For a deterministic $(\finf;\fsum)$-automaton $\nestedA$ and $C>0$ we define $\nestedA^C$ as
the deterministic $(\finf;\fsum)$-automaton obtained from $\nestedA$ by stopping every slave automaton if it exceeds $C$ steps.
We show that for every deterministic $(\finf;\fsum)$-automaton $\nestedA$ and $\epsilon >0$, there exists $C$ exponential in $|\nestedA|$
and polynomial in $\epsilon$ such that the expected values of $\nestedA$ and $\nestedA^C$ differ by at most $\epsilon$.
\end{itemize}

\noindent
This paper is an extended and corrected version of~\cite{conferenceVersion,DBLP:conf/sas/ChatterjeeHO16}. We present detailed proofs, which
could not be published in~\cite{conferenceVersion} due to space constrains.
The main corrections over~\cite{conferenceVersion} are: Table~\ref{tab1} and
Theorem~10 (Theorem~\ref{th:undecidable-limsup} in this paper).
These flaws in~\cite{conferenceVersion} are consequences of a false claim about duality
between deterministic $(\finf;g)$-automata (resp., $(\fliminf;g)$-automata) and deterministic
 $(\fsup;-g)$-automata (resp. ($\flimsup;-g)$-automata).
This duality indeed holds for non-deterministic NWA or deterministic NWA that accept all words (or almost all words for
 probabilistic questions).
However, it does not extend to all deterministic NWA (see~\cite{nested} and Remark~\ref{rem:duality}).

Moreover, we discuss extensions of our main results in Section~\ref{s:discussion}, which is a new contribution.
We consider there (1)~the case of NWA that do not accept almost all words, (2)~the probabilistic variant of the quantitative inclusion problem for NWA, and
(3)~the parametric complexity of the probabilistic questions, in which we fix the NWA and ask for the complexity w.r.t.\ the Markov chain.
The parametric complexity corresponds to evaluation of a fixed specification (for example average response time from Example~\ref{ex:intro})
represented by an NWA  on a system represented by a Markov chain.
Finally, we elaborate on translations between NWA and automata with monitor counters discussed in~\cite{conferenceVersion,DBLP:conf/sas/ChatterjeeHO16}.

\smallskip\noindent{\em Related works.}
Quantitative automata and logic have been extensively and intensively
studied in recent years.
The book~\cite{Droste:2009:HWA:1667106} presents an excellent collection of results
of weighted automata on finite words.
Weighted automata on infinite words have been studied in~\cite{Chatterjee08quantitativelanguages,DBLP:journals/corr/abs-1007-4018,DrosteR06}.
The extension to weighted automata with monitor counters over finite words has been considered (under the name of
cost register automata) in~\cite{DBLP:conf/lics/AlurDDRY13}.
A version of nested weighted automata over finite words has been
studied in~\cite{bollig2010pebble}, and nested weighted automata over
infinite words have been studied in~\cite{nested}.
Several quantitative logics have also been studied, such as~\cite{BokerCHK14,BouyerMM14,AlmagorBK14}.
While a substantial work has been done for quantitative automata and logics, quite surprisingly
none of the above works consider the automata (or the logic) under probabilistic semantics that
we consider in this work.
Probabilistic models (such as Markov decision processes) with quantitative properties
(such as limit-average or discounted-sum) have also been extensively studied for
single objectives~\cite{filar,Puterman}, and for multiple objectives and their
combinations~\cite{CMH06,Cha07,CFW13,BBCFK11,CKK15,Forejt,FKN11,CD11,Baier-CSL-LICS-1,Baier-CSL-LICS-2}.
However, these works do not consider properties that are expressible by nested weighted
automata (such as average response time) or automata with monitor counters.

\section{Preliminaries}
\Paragraph{Words}.
We consider a finite \emph{alphabet} of letters $\Sigma$.
A \emph{word} over $\Sigma$ is a (finite or infinite) sequence of letters from $\Sigma$.
We denote the $i$-th letter of a word $w$ by $w[i]$.
The length of a finite word $w$ is denoted by $|w|$; and the length of an infinite word
$w$ is $|w| = \infty$.

\smallskip
\Paragraph{Labeled automata}. For a set $X$, an \emph{$X$-labeled automaton} $\aut$ is a tuple
$\tuple{\Sigma, Q, Q_0, \delta, F, \cost}$, where
(1)~$\Sigma$ is the alphabet,
(2)~$Q$ is a finite set of states,
(3)~$Q_0 \subseteq Q$ is the set of initial states,
(4)~$\delta \subseteq Q \times \Sigma \times Q$ is a transition relation,
(5)~$F$ is the set of accepting states,
and
(6)~$\cost : \delta \mapsto X$ is a labeling function.
A labeled automaton $\tuple{\Sigma, Q, q_0, \delta, F, \cost}$ is
\emph{deterministic} if and only if
$\delta$ is a function from $Q \times \Sigma$ into $Q$
and $Q_0$ is a singleton.
In definitions of deterministic labeled automata we omit curly brackets in the description of $Q_0 = \{ q_0\}$
and write $\tuple{\Sigma, Q, q_0, \delta, F, \cost}$.

\smallskip
\Paragraph{Semantics of (labeled) automata}.
A \emph{run} $\run$ of a (labeled) automaton $\aut$ on a word $w$ is a sequence of states
of $\aut$ of length $|w|+1$
such that $\run[0]$ belongs to the initial states of $\aut$
and for every $0 \leq i \leq |w|-1$ we have $(\pi[i], w[i], \pi[i+1])$  is a transition of $\aut$.
A run $\pi$ on a finite word $w$ is \emph{accepting} if and only if the last state $\pi[|w|]$ of the run
is an accepting state of $\aut$.
A run $\pi$ on an infinite word $w$ is \emph{accepting} if and only if  some accepting state of $\aut$ occurs
infinitely often in $\pi$.
For an automaton $\aut$ and a word $w$, we define $\Acc(w)$ as the set of accepting runs on $w$.
Note that for deterministic automata, every word $w$ has at most one accepting run ($|\Acc(w)| \leq 1$).

\smallskip
\Paragraph{Weighted automata and their semantics}.
A \emph{weighted automaton} is a $\Z$-labeled automaton, where $\Z$ is the set of integers.
The labels are called \emph{weights}.
We assume that weights are given in the unary notation, and, hence,
the values of weights are linearly bounded in the size of weighted automata.

We define the semantics of weighted automata in two steps. First, we define the value of a
run. Second, we define the value of a word based on the values of its runs.
To define values of runs, we will consider  \emph{value functions} $f$ that
assign real numbers to sequences of integers.
Given a non-empty word $w$, every run $\pi$ of $\aut$ on $w$ defines a sequence of weights
of successive transitions of $\aut$, i.e.,
$\cost(\pi)={(\cost(\pi[i-1], w[i], \pi[i]))}_{1\leq i \leq |w|}$;
and the value $f(\pi)$ of the run $\pi$ is defined as $f(\cost(\pi))$.
We denote by $(\cost(\pi))[i]$ the weight of the $i$-th transition,
i.e., $\cost(\pi[i-1], w[i], \pi[i])$.
The value of a non-empty word $w$ assigned by the automaton $\aut$, denoted by  $\valueL{\aut}(w)$,
is the infimum of the set of values of all \emph{accepting} runs;
i.e., $\inf_{\pi \in \Acc(w)} f(\pi)$, and we have the usual semantics that infimum of an
empty set is infinite, i.e., the value of a word that has no accepting runs is infinite.
Every run $\pi$ on an empty word has length $1$ and the sequence $\cost(\pi)$ is empty, hence
we define the value $f(\pi)$ as an external (not a real number) value $\bot$.
Thus, the value of the empty word is either $\bot$, if the empty word is accepted by $\aut$, or $\infty$
otherwise.
To indicate a particular value function $f$ that defines the semantics,
we will call a weighted automaton $\aut$ an $f$-automaton.

\smallskip
\Paragraph{Value functions}.
We will consider the classical functions and their natural variants for
value functions.
For finite runs we consider the following value functions: for runs of length $n+1$ we have
\begin{enumerate}
\item \emph{Min and max}: $\fmin(\pi) = \min_{i=1}^n (\cost(\pi))[i]$ and $\fmax(\pi) = \max_{i=1}^n (\cost(\pi))[i]$,
\item \emph{Sum}: $\fsum(\pi) = \sum_{i=1}^{n} (\cost(\pi))[i]$,
\item \emph{Absolute sum}: $\fsum^+(\pi) = \sum_{i=1}^{n} \abs((\cost(\pi))[i])$
is the sum of the absolute values of the weights ($\abs$ denotes the
absolute value of a number),
and
\item \emph{Bounded sum}: $\fBsum{B}(\pi) = \fsum(\pi)$,
if for all prefixes $\pi'$ of $\pi$ we have $\abs(\fsum(\pi')) \leq B$,
otherwise $\fBsum{B}(\pi)$ is equal to first crossed bound $-B$ or $B$, i.e.,
the bounded sum value function returns the sum if all the
partial absolute sums are below a bound $B$, otherwise it returns the first crossed bound.
Weighted automata with the bounded-sum value function can model bounded quantities such as energy with the lower and the upper bound~\cite{DBLP:journals/acta/BouyerMRLL18}.
\end{enumerate}

\noindent
We denote the above class of value functions for finite words as
\[\FinVal=\{\fmax,\fmin,\fBsum{B},\fsum\}.\]

\noindent
For infinite runs we consider:
\begin{enumerate}
\item \emph{Supremum and Infimum}: $\fsup(\pi) = \sup \{ (\cost(\pi))[i] \mid i > 0 \}$ and
 $\finf(\pi) = \inf \{ (\cost(\pi))[i] \mid i > 0 \}$,
 \item \emph{Limit supremum and Limit infimum}:
 $\flimsup(\pi) = \lim\sup \{ (\cost(\pi))[i] \mid i > 0\}$, and
 $\fliminf(\pi) = \lim\inf \{  (\cost(\pi))[i] \mid i > 0 \}$, and
\item \emph{Limit average}: $\flimavg(\pi) = \limsup\limits_{k \rightarrow \infty} \frac{1}{k} \cdot \sum_{i=1}^{k} (\cost(\pi))[i]$.
\end{enumerate}

\noindent
We denote the above class of infinite-word value functions as
\[\InfVal=\{\fsup,\finf,\flimsup,\fliminf,\flimavg\}.\]

\smallskip
\Paragraph{Silent moves}. Consider a $(\Z \cup \{ \bot\})$-labeled automaton.
We regard such an automaton as an extension
of a weighted automaton in which transitions labeled by $\bot$ are \emph{silent}, i.e., they do not contribute to
the value of a run. Formally, for every function $f \in \InfVal$ we define
$\silent{f}$ as the value function that applies $f$ on sequences after removing $\bot$ symbols.
The significance of silent moves is as follows: they allow to ignore transitions, and thus provide
robustness where properties could be specified based on desired events rather than steps.

\section{Extensions of weighted automata}
\newcommand{\autDiff}{\aut_{\textrm{diff}}}
In this section we consider two extensions of weighted automata,
namely, automata with monitor counters and nested weighted automata.

\subsection{Automata with monitor counters}
Intuitively, automata with monitor counters are an extension of weighted automata
with counters, where the transitions do not depend on values of counters.
We define them formally below.

\Paragraph{Automata with monitor counters.}
\newcommand{\counterA}{\mathcal{A}^{\textrm{m-c}}}
An \emph{automaton with $n$ monitor counters} $\counterA$ is a tuple $\tuple{ \Sigma, Q, Q_0, \delta, F}$  where
\begin{enumerate}
\item $\Sigma$ is the alphabet,
\item $Q$ is a finite set of states, and $Q_0 \subseteq Q$ is the set of initial states,
\item $\delta$ is a finite subset of  $Q \times \Sigma \times Q \times {(\Z \cup \{ s,t \})}^n$ called a transition relation,
(each component refers to one monitor counter, where letters $s,t$ refer to starting  and terminating the  counter,
respectively, and the value from $\Z$ is the value that is added to the counter), and
\item $F$ is the set of accepting states.
\end{enumerate}
Moreover, we assume that for every $(q,a,q',\vec{u}) \in \delta$, at most one component in $\vec{u}$ contains $s$, i.e.,
at most one counter is started at each position.
Intuitively, the automaton $\counterA$ is equipped with $n$ counters.
The transitions of $\counterA$ do not depend on the values of counters (hence, we call them monitor counters); and
every transition is of the form $(q,a,q',\vec{v})$, which means that
if $\counterA$ is in the state $q$ and the current letter is $a$, then
it can move to the state $q'$ and update counters according to $v$.
Each counter is initially inactive.
It is started by the instruction $s$, and
it changes its value at every step by adding the value of the corresponding component of $v$,
until termination $t$.
The value of the counter at the time it terminates is then assigned to the position where it has been started.
An automaton with monitor counters $\counterA$ is \emph{deterministic} if and only if $Q_0$ is a singleton and $\delta$ is a function
from $Q \times \Sigma$ into $Q \times {(\Z \cup \{ s,t \})}^n$.

\Paragraph{Semantics of automata with monitor counters.}
A sequence $\run = \pi[0] \pi[1] \ldots$ of elements from $Q \times {(\Z \times \{\bot\})}^n$ is a \emph{run} of $\counterA$ on a word $w = w[1] w[2] \ldots$
if
\begin{enumerate}
\item $\run[0] = \tuple{q_0, \vec{\bot}}$ and $q_0 \in Q_0$, and
\item for every $i > 0$, if $\run[i-1] = \tuple{q,\vec{u}}$
and $\run[i] = \tuple{q', \vec{u}'}$ then $\counterA$ has a transition
$(q,w[i],q',\vec{v})$ and for every $j \in [1,n]$ we have
\begin{enumerate}
\item if $v[j] = s$, then $u[j] = \bot$ and $u'[j] = 0$,
\item if $v[j] = t$, then $u[j] \in \Z$ and $u'[j] = \bot$, and
\item if $v[j] \in \Z$, then $u'[j] = u[j] + v[j]$.
\end{enumerate}
\end{enumerate}
A run $\run$ is \emph{accepting} if some state from $F$ occurs infinitely often on the first component of $\run$,
some counter is started infinitely often, and every started counter is finally terminated.
An accepting run $\run$ defines a sequence $\weightedRun$
of integers and $\bot$ as follows: let the counter started at position $i$ be $j$, and
let the value of the counter $j$ terminated at the earliest position after $i$ be $x_j$,
then $\weightedRun[i]$ is $x_j$.
Otherwise, if no counter has been started at position $i$, we define $\weightedRun[i] = \bot$,
Observe that for an accepting $\run$, the sequence $\weightedRun$ contains infinitely positions with integer values.
The semantics of automata with monitor counters is given, similarly to weighted automata,
by applying the value function to the sequence $\weightedRun$ with $\bot$ elements removed.

\begin{rem}
Automata with monitor counters are very similar in spirit to \emph{cost register automata} considered
in~\cite{DBLP:conf/lics/AlurDDRY13}. The key difference is that we consider infinite words and value functions
associated with them, whereas previous works consider finite words.
Another key difference is that in this work we will consider probabilistic semantics, and
such semantics has not be considered for cost register automata before.
\end{rem}

\begin{exa}[Blocks difference]%
\label{ex:AMC}
Consider an alphabet $\Sigma = \{a,\#\}$ and the language $\lang$ defined as ${(\#^2 a^* \# a^* \#)}^{\omega}$.
We consider a quantitative property ``the maximal block-length difference between odd and even positions'' on the words from the language $\lang$, i.e.,
the value of word $\#^2 a^{n[1]} \# a^{n[2]} \#^3 \ldots $ is $\sup_{0 \leq i} |n[2\cdot i+1] - n[2\cdot i+2]|$.
This property can be expressed by a $\fsup$-automaton $\autDiff$ with two monitor counters depicted in Figure~\ref{fig:autDiff}.

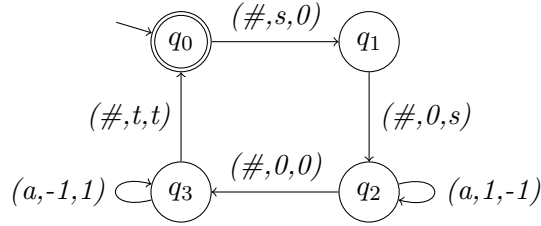
\begin{figure}
\begin{center}
\begin{tikzpicture}[initial text={}]
\newcommand{\x}{2.5}
\node[state,accepting,initial] (q1) at (0,0) {$q_0$};
\node[state] (q2) at (\x,0) {$q_1$};
\node[state] (q5) at (\x,-2) {$q_2$};
\node[state] (q6) at (0,-2) {$q_3$};

\draw[->] (q1) to node[above]  {$(\#,s,0)$} (q2);
\draw[->] (q2) to node[right]  {$(\#,0,s)$} (q5);
\draw[->, loop right] (q5) to node[right]  {$(a,1,-1)$} (q5);
\draw[->] (q5) to node[above]  {$(\#,0,0)$} (q6);
\draw[->, loop left] (q6) to node[left]  {$(a,-1,1)$} (q6);
\draw[->] (q6) to node[left]  {$(\#,t,t)$} (q1);
\end{tikzpicture}
\end{center}
\caption{The automaton $\autDiff$ computing the maximal difference between the lengths of blocks of $a$'s at odd and the following even positions.}%
\label{fig:autDiff}
\end{figure}

The automaton $\autDiff$ has a single initial state $q_0$, which is also the only accepting state.
It processes the word $w$ in subwords $\#^2 a^{k} \# a^{m} \#$ in the following way.
First, it reads $\#^2$ upon which it takes transitions from $q_0$ to $q_1$ and from $q_1$ to $q_2$, where it  starts counters $1$ and $2$.
Next, it moves to the state $q_2$ where it counts letters $a$ incrementing counter $1$ and decrementing counter $2$.
Then, upon reading $\#$, it moves to $q_3$, where it counts letters $a$, but it decrements counter $1$ and increments counter $2$.
After reading $\#^2 a^{k} \# a^{m}$ the value of counter $1$ is $k-m$ and counter $2$ is $m-k$.
In the following transition from $q_3$ to $q_0$, the automaton terminates both counters.
The aggregating function  of $\autDiff$ is $\fsup$, thus the automaton discards the lower value, i.e.,
the value of $\#^2 a^{k} \# a^{m} \#$ is $|k-m|$ and the automaton computes the supremum over values of all blocks.
It follows that the value of $\#^2 a^{n[1]} \# a^{n[2]} \#^3 \ldots $ is $\sup_{0 \leq i} |n[2\cdot i+1] - n[2\cdot i+2]|$.
\end{exa}

\subsection{Nested weighted automata}
In this section we describe nested weighted automata introduced in~\cite{nested},
and closely follow the description of~\cite{nested}.
For more details and illustrations of such automata we refer the reader
to~\cite{nested}.
We start with an informal description.

\smallskip
\noindent{\em Informal description.}
A \emph{nested weighted automaton} consists of a labeled automaton over infinite words,
called the \emph{master automaton}, a value function $f$ for infinite words,
and a set of weighted automata over finite words, called \emph{slave automata}.
A nested weighted automaton can be viewed as follows:
given a word, we consider the run of the master automaton on the word,
but the weight of each transition is determined by dynamically running
slave automata; and then the value of a run is obtained using the
value function $f$.
That is, the master automaton proceeds on an input word as an usual automaton,
except that before it takes a transition, it starts a slave automaton
corresponding to the label of the current transition.
The slave automaton starts at the current position in the word of the master automaton
and works on some finite part of the input word. Once the slave automaton finishes,
it returns its value to the master automaton, which treats the returned
value as the weight of the current transition that is being executed.
The slave automaton might immediately accept and return value $\bot$,
which corresponds to \emph{silent} transitions.
If one of slave automata rejects, the nested weighted automaton rejects.
We define this formally as follows.

\smallskip

\Paragraph{Nested weighted automata}.
A \emph{nested weighted automaton} (NWA) $\nestedA$ is a tuple
\[
    \tuple{\masterA; f; \slaveA_1, \ldots, \slaveA_k}
\]
such that

\begin{enumerate}
\item $\masterA$, called the \emph{master automaton}, is a $\{1, \ldots, k\}$-labeled automaton over infinite words
(the labels are the indexes of automata $\slaveA_1,  \ldots, \slaveA_k$),
\item $f$ is a value function on infinite words, called the \emph{master value function}, and
\item $\slaveA_1, \ldots, \slaveA_k$ are weighted automata over finite words called \emph{slave automata}.
\end{enumerate}
\smallskip

\noindent Intuitively, an NWA can be regarded as an $f$-automaton whose weights are dynamically computed at every step by the corresponding slave automaton.
We define an \emph{$(f;g)$-automaton} as an NWA where the master value function is $f$ and all slave automata are $g$-automata.

\smallskip
\Paragraph{Semantics: runs and values}.
A \emph{run} of an NWA $\nestedA$ on an infinite word $w$ is an infinite sequence
$(\masterRun, \slaveRun_1, \slaveRun_2, \ldots)$ such that
(1)~$\masterRun = \masterRun[0] \masterRun[1] \ldots$ is a run of $\masterA$ on $w = w[1] w[2] \ldots$;
(2)~for every $i>0$ we have $\slaveRun_i$ is a run of the automaton $\slaveA_{\cost(\masterRun[i-1], w[i], \masterRun[i])}$,
referenced by the label $\cost(\masterRun[i-1], w[i], \masterRun[i])$ of the master automaton, on some finite subword $w[i,j]$ of the input word $w$.
The run $(\masterRun, \slaveRun_1, \slaveRun_2, \ldots)$ is \emph{accepting} if all
runs $\masterRun, \slaveRun_1,  \slaveRun_2, \ldots$ are accepting (i.e., $\masterRun$ satisfies its acceptance
condition and each $\slaveRun_1,\slaveRun_2, \ldots$ ends in an accepting state)
and infinitely many runs of slave automata have length greater than $1$ (the master automaton takes infinitely many non-silent transitions).
The value of the run $(\masterRun, \slaveRun_1, \slaveRun_2, \ldots)$ is defined as
$\silent{f}( v(\pi_1) v(\pi_2) \ldots)$, where $v(\pi_i)$ is the value of the run $\pi_i$ in
the corresponding slave automaton.
The value of a word $w$ assigned by the automaton $\nestedA$, denoted by
$\valueL{\nestedA}(w)$, is the infimum of the set of values of all \emph{accepting} runs.
We require accepting runs to contain infinitely many non-silent transitions because
$f$ is a value function over infinite sequences, so we need
the sequence $v(\pi_1) v(\pi_2) \ldots$ with $\bot$ symbols removed to be infinite.

\smallskip
\Paragraph{Deterministic nested weighted automata}. An NWA $\nestedA$ is \emph{deterministic} if (1)~the master automaton
and all slave automata are deterministic, and (2)~in all slave automata
accepting states have no outgoing transitions.
Condition (2) implies that no accepting run of a slave automaton visits an accepting state twice.
Intuitively, slave automata have to accept the first time they encounter an accepting state as
they will not reach an accepting state again.

\smallskip
\Paragraph{Bounded width.} An NWA has \emph{width} $k$ if and only if $k$ is the minimal number such that
in every accepting run at every
position at most $k$ slave automata are active.

\begin{exa}[Average response time with bounded requests]%
\label{ex:NWA}
Consider an alphabet $\Sigma$ consisting of requests $r$, grants $g$ and idle instructions $i$.
The average response time (ART) property asks for the average number of instructions between
any request and the following grant. It has been shown in~\cite{nested} that NWA can express ART\@.
However, the automaton from~\cite{nested} does not have bounded width.
To express the ART property with NWA of bounded width we consider only words such that between any two grants there are at most $k$ requests.

Average response time over words where between any two grants there are at most $k$ requests can be expressed by a $(\flimavg;\fsum)$-automaton
$\nestedA$. Such an automaton $\nestedA = (\masterA; \flimavg; \slaveA_1, \slaveA_2)$ is depicted in Fig.~\ref{fig:ART-k}.
The master automaton of $\nestedA$ accepts only words with  infinite number of requests and grants, where every grant is followed by a request and
there are at most $k$ requests between any two grants.
On letters $i$ and $g$, the master automaton invokes a dummy automaton $\slaveA_1$, which immediately accepts;
the result of invoking such an automaton is equivalent to taking a silent transition as
the automaton $\slaveA_1$ returns $\bot$, the empty value.
On letters $r$, denoting requests, the master automaton invokes $\slaveA_2$, which
counts the number of letters to the first occurrence of letter $g$, i.e.,
the automaton $\slaveA_2$ computes the response time for the request on the position it is invoked.
The automaton $\nestedA$ computes the limit average of all returned values, which is precisely
ART (on the accepted words).
Note that the width of $\nestedA$ is $k$.

\begin{figure}
\centering
\begin{tikzpicture}
\newcommand{\x}{2.0}
\node[state,accepting] (q0) at (-1,0) {$q_0$};
\node[state] (q1) at (\x,0) {$q_1$};
\node[state, inner sep =1pt] (q2) at (2*\x,0) {$q_{k-1}$};
\node[state] (q3) at (3*\x,0) {$q_k$};

\node at (1.5*\x,0) {$\ldots$};

\draw[->, loop above] (q1) to node[above] (e2) {$(i,1)$} (q1);
\draw[->, loop above] (q2) to node[above] (e3) {$(i,1)$} (q2);
\draw[->, loop above] (q3) to node[above] (e4) {$(i,1)$} (q3);

\draw[->,bend left=10] (q1) to node[below] (e5) {$(g,1)$} (q0);
\draw[->,bend left=45] (q2) to node[below] (e6) {$(g,1)$} (q0);
\draw[->,bend left=65] (q3) to node[below] (e7) {$(g,1)$} (q0);

\draw[->] (q0) to node[above] (e8) {$(r,2)$} (q1);
\draw[->] (q2) to node[above] (e9) {$(r,2)$} (q3);

\node at (5.5,-1.5) {$\masterA$};

\begin{scope}[yshift=3cm]
\node[state,accepting] (q00) at (-1,0) {$q_0^1$};

\node (A1) at (-1,-0.8) {$\slaveA_1$};

\node[state] (q10) at (4,0) {$q_0^2$};
\node[state,accepting] (q11) at (6,0) {$q_1^2$};

\node (A2) at (5,-0.8) {$\slaveA_2$};

\draw[->] (q10) to node[above] {$(g,0)$} (q11);
\draw[->, loop above] (q10) to node[above] {$(i,1)$} (q10);
\draw[->, loop left] (q10) to node[left] {$(r,1)$} (q10);
\end{scope}

\draw[->, dashed] (e2) to (A1);
\draw[->, dashed] (e3) to (A1);
\draw[->, dashed] (e4) to (A1);
\draw[->, dashed, bend left = 15] (e5) to (A1);
\draw[->, dashed, bend left = 25] (e6) to (A1);
\draw[->, dashed, bend left = 35] (e7) to (A1);

\draw[->, dotted] (e8) to (A2);
\draw[->, dotted, bend left = 30] (e9) to (A2);

\end{tikzpicture}
\caption{The $(\flimavg;\fsum)$-automaton computing the average response time over words with infinite number of requests and grants such that between any two
grants there are at most $k$ requests.}%
\label{fig:ART-k}
\end{figure}
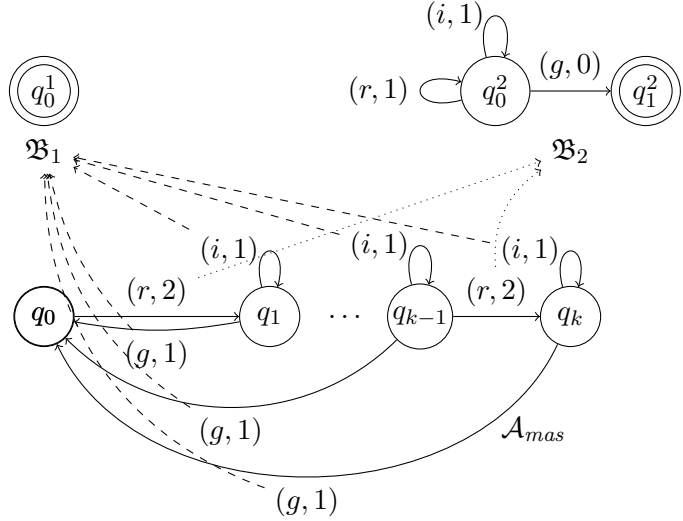
\end{exa}

\subsection{Translation}
We now present translations from NWA to automata with monitor counters and vice-versa. To state correctness of translation, we first define equivalence.

\Paragraph{Equivalence of quantitative automata}. We say that $\aut_1, \aut_2$, each being a weighted automaton, an automaton with monitor counters or an NWA over infinite words from $\Sigma$, are \emph{equivalent}
 if and only if for all words $w \in  \Sigma^{\omega}$ we have
$\aut_1(w) = \aut_2(w)$.

Now, we state the main translation lemma:

\begin{restatable}[Translation Lemma]{lem}{MCvsNested}%
\label{l:mc-vs-nested}
For every value function $f \in \InfVal$ on infinite words we have the following:
\begin{enumerate}
\item Every deterministic $f$-automaton with monitor counters $\counterA$ can be transformed in polynomial time
into an equivalent deterministic $(f;\fsum)$-automaton of bounded width.
\item Every non-deterministic (resp., deterministic) $(f;\fsum)$-automaton of bounded width can be transformed in exponential time
into an equivalent non-deterministic (resp., deterministic) $f$-automaton with monitor counters.
\end{enumerate}
\end{restatable}

\noindent
Before the formal proof, we illustrate below the key ideas of the above translations of
Lemma~\ref{l:mc-vs-nested} to automata from Examples~\ref{ex:AMC}~and~\ref{ex:NWA}.

\begin{exa}[Translation of automata with monitor counters to nested weighted automata]
Consider a deterministic automaton $\aut$ with $k$ monitor counters.
We construct an NWA $\nestedA$ equivalent to $\aut$.
The automaton $\nestedA$ uses $k$ slave automata to track values of $k$ monitor counters in the following way.
The master automaton of $\nestedA$ simulates $\aut$; it invokes slave automata whenever $\aut$ starts monitor counters.
Slave automata simulate $\aut$ as well. Each slave automaton is associated with some counter $i$; it starts in the state (of $\aut$) the counter $i$ is initialized,
simulates the value of counter $i$, and terminates when counter $i$ is terminated.
Figure~\ref{fig:AMCtoNWA} presents the result of transition of  the automaton $\autDiff$ from Example~\ref{ex:AMC}
to a $(\fsup;\fsum)$-automaton of width bounded by $3$.

\begin{figure}
\centering
\begin{tikzpicture}[initial text={}]
\newcommand{\x}{1.0}

\begin{scope}[yshift=4cm,xshift=-3.2cm]
\node[state,accepting] (q11) at (0,0) {$q_0'$};
\node[state,initial] (q12) at (2*\x,0) {$q_1'$};
\node[state] (q15) at (2*\x,-2) {$q_2'$};
\node[state] (q16) at (0,-2) {$q_3'$};

\draw[->] (q12) to node[right]  {$(\#,0)$} (q15);
\draw[->, loop right] (q15) to node[right]  {$(a,1)$} (q15);
\draw[->] (q15) to node[above]  {$(\#,0)$} (q16);
\draw[->, loop left] (q16) to node[left]  {$(a,-1)$} (q16);
\draw[->] (q16) to node[left]  {$(\#,0)$} (q11);

\node (A1) at (1,-2.7) {$\slaveA_1$};

\end{scope}

\begin{scope}[yshift=4cm,xshift=3.5cm]
\node[state]  at (2*\x,0) {$q_1''$};
\node[state,accepting] (q21) at (0,0) {$q_0''$};
\node[state,initial above] (q25) at (2*\x,-2) {$q_2''$};
\node[state] (q26) at (0,-2) {$q_3''$};

\draw[->, loop right] (q25) to node[right]  {$(a,-1)$} (q25);
\draw[->] (q25) to node[above]  {$(\#,0)$} (q26);
\draw[->, loop left] (q26) to node[left]  {$(a,1)$} (q26);
\draw[->] (q26) to node[left]  {$(\#,0)$} (q21);
\node (A2) at (1,-2.7) {$\slaveA_2$};
\end{scope}

\node[state,accepting] (A3state)  at (-3,-0.5) {$q_0$};
\node (A3) at (-3,-1.3) {$\slaveA_3$};

\node[state,accepting,initial] (q1) at (0,0) {$q_0$};
\node[state] (q2) at (2*\x,0) {$q_1$};
\node[state] (q5) at (2*\x,-2) {$q_2$};
\node[state] (q6) at (0,-2) {$q_3$};

\draw[->] (q1) to node[above] (e1) {$(\#,1)$} (q2);
\draw[->] (q2) to node[right] (e2) {$(\#,2)$} (q5);
\draw[->, loop right] (q5) to node[right] (e3) {$(a,3)$} (q5);
\draw[->] (q5) to node[above] (e4) {$(\#,3)$} (q6);
\draw[->, loop left] (q6) to node[left] (e5)  {$(a,3)$} (q6);
\draw[->] (q6) to node[left] (e6) {$(\#,3)$} (q1);

\draw[->,dotted] (e1) to (A1);
\draw[->,dotted] (e2) to (A2);

\draw[->,dotted, bend right=18] (e3) to (A3state);
\draw[->,dotted, bend left=5] (e4) to (A3state);
\draw[->,dotted] (e5) to (A3state);
\draw[->,dotted] (e6) to (A3state);

\end{tikzpicture}
\caption{A nested weighted automaton resulting from translation of the automaton $\autDiff$ from Example~\ref{ex:AMC}.
The master automaton is obtained from $\autDiff$ (see Figure~\ref{fig:autDiff}) by changing the labels of transitions.
All slave automata are defined based on $\autDiff$; each slave automaton corresponds to $\autDiff$ starting in a state $q'$, in which a counter $i$ is initialized, and contains all the transitions that can be taken before
the counter $i$ is terminated.}%
\label{fig:AMCtoNWA}
\end{figure}
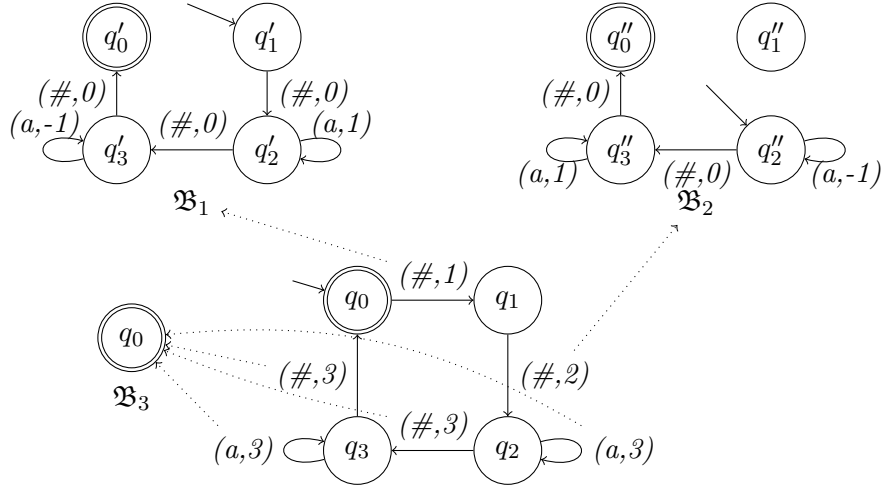
\end{exa}

\begin{exa}[Translation of nested weighted automata of bounded width to automata with monitor counters]
Consider an $(f;\fsum)$-automaton $\nestedA$ of width bounded by $k$.
We construct an automaton with monitor counters $\aut_{\nestedA}$, which
simulates the master automaton and up to $k$ slave automata running in parallel.
To simulate values of slave automata it uses monitor counters, each counter separately for each slave automaton.

Figure~\ref{fig:NWAtoAMC} shows the result of translation of  the automaton $\nestedA$ from Example~\ref{ex:NWA}
to the automaton with monitor counters $\aut_{\nestedA}$. The set of states of $\aut_{\nestedA}$ there is
$\{q_0, \ldots, q_k\} \times {(\{q_0^2, \bot\})}^k$, i.e.,
the states of the master automaton and all non-accepting states of slave automata (in deterministic NWA accepting states are sink states, hence storing them is redundant).
Now, observe that only reachable states of $\aut_{\nestedA}$ are $(q_0, \bot, \ldots, \bot), (q_1, q_0^2, \bot, \ldots, \bot), \ldots, (q_k, q_0^2, \ldots, q_0^2)$, i.e., the reachable part of
$\aut_{\nestedA}$ is isomorphic (in the sense of graphs) to the master automaton of $\nestedA$.

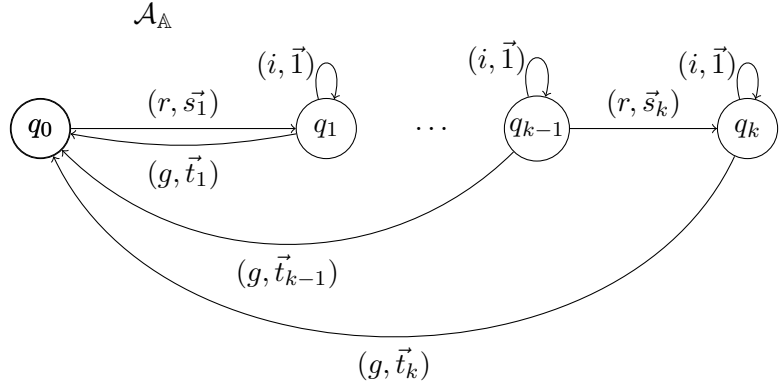
\begin{figure}
\centering
\begin{tikzpicture}
\newcommand{\x}{2.8}
\node[state] (q0) at (-1,0) {$q_0$};
\node[state] (q1) at (\x,0) {$q_1$};
\node[state, inner sep =1pt] (q2) at (2*\x,0) {$q_{k-1}$};
\node[state] (q3) at (3*\x,0) {$q_k$};

\node at (1.5*\x,0) {$\ldots$};

\draw[->, loop above] (q1) to node[above] (e2) {$(i,\vec{1})$} (q1);
\draw[->, loop above] (q2) to node[above] (e3) {$(i,\vec{1})$} (q2);
\draw[->, loop above] (q3) to node[above] (e4) {$(i,\vec{1})$} (q3);

\draw[->,bend left=10] (q1) to node[below] (e5) {$(g,\vec{t}_1)$} (q0);
\draw[->,bend left=45] (q2) to node[below] (e6) {$(g,\vec{t}_{k-1})$} (q0);
\draw[->,bend left=65] (q3) to node[below] (e7) {$(g,\vec{t}_k)$} (q0);

\draw[->,bend left=10] (q0) to node[above] (e8) {$(r,\vec{s}_1)$} (q1);
\draw[->] (q2) to node[above] (e9) {$(r,\vec{s}_k)$} (q3);
\node at (0.5,1.5) {$\aut_{\nestedA}$};
\end{tikzpicture}
\caption{The (reduced) result of translation of the automaton $\nestedA$ from  Example~\ref{ex:NWA} to an automaton with monitor counters. All vectors have dimension $k$.
Vector $\vec{1}$ denotes the vector with all components equal $0$.
Vector $\vec{s}_i$ denotes the whose $i$-th component is $s$ and other components are $1$.
Vector $\vec{t}_i$ denotes the vector whose components $1, \ldots, i$ are $t$ and the remaining components are $0$.
}%
\label{fig:NWAtoAMC}
\end{figure}
\end{exa}

\begin{proof}[Proof of Lemma~\ref{l:mc-vs-nested}]
\Paragraph{(Translation of automata with monitor counters to NWA)}:
Consider a deterministic $f$-automaton $\counterA$ with $k$ monitor counters and the set of states $Q^{m-c}$.
We define an $(f;\fsum)$-automaton $\nestedA$, which consists of a master automaton $\masterA$ and
slave automata $\{ \slaveA_{i,q} \mid i \in \{1,\ldots, k\}, q \in Q^{m-c} \} \cup \{\slaveA_{\bot} \}$.
The slave automaton $\slaveA_{\bot}$ is a dummy automaton, i.e., it has only a single state which is both
the initial and the accepting state. Invoking such an automaton is equivalent to taking a silent transition (with no weight).
Next, the master automaton $\masterA$ and slave automata $\{ \slaveA_{i,q} \mid i \in \{1,\ldots, k\}, q \in Q^{m-c} \}$
are variants of $\counterA$, i.e., they share the underlying transition structure.
The automaton $\masterA$ simulates $\counterA$, i.e., it has the same states and the transitions among these states as $\counterA$.
However, whenever $\counterA$ activates counter $i$, the master automaton invokes the slave automaton $\slaveA_{i,q}$, where $q$ is their current state (both $\masterA$ and the simulated $\counterA$).
The accepting condition of $\masterA$ is the same as of $\counterA$.
We can construct $\masterA$ in polynomial time in $|\counterA|$.
For every $i \in \{1,\ldots, k\}$, the slave automaton $\slaveA_{i,q}$
keeps track of counter $i$, i.e., it
simulates $\counterA$ and applies instructions of $\counterA$ for counter $i$
to its value. That is, whenever $\counterA$ changes the value of counter $i$ by $m$, the automaton $\slaveA_{i,q}$
takes a transition of the weight $m$. Finally, $\slaveA_{i,q}$ terminates precisely when $\counterA$ terminates counter $i$.
The automaton $\slaveA_{i,q}$ can be constructed in polynomial time in $|\counterA|$.
There are at most $k \cdot |\counterA|$ such slave automata and each of them has the size bounded by $|\counterA|$.
Therefore, $|\nestedA|$ is polynomial in $|\counterA|$ and can be constructed in polynomial time in $|\counterA|$.

The semantics of automata with monitor counters implies that $\nestedA$ accepts if and only if $\counterA$ accepts and, for every word,
the sequences of weights produced by the runs of $\nestedA$ and $\counterA$ on that word coincide. Therefore,
the values of $\nestedA$ and $\counterA$ coincide on every word.

\Paragraph{(Translation of NWA of bounded width to automata with monitor counters)}: We show that non-deterministic (resp., deterministic)
$f$-automata with monitor counters  subsume
non-deterministic (resp., deterministic) $(f;\fsum)$-automata of bounded width.
Consider a non-deterministic $(f;\fsum)$-automaton $\nestedA$ with width bounded by $k$.
We define an $f$-automaton $\counterA$ with $k$ monitor counters that works as follows.
Let $Q_{mas}$ be the set of states of the master automaton of $\nestedA$ and $Q_s$ be the union of the sets of states of the slave automata of $\nestedA$.
The set of states of $\counterA$ is $Q_{mas} \times (Q_{s} \cup \{ \bot \}) \times \cdots \times (Q_{s} \cup \{ \bot \}) =
Q_{mas} \times {(Q_{s} \cup \{ \bot \})}^k$.
The automaton $\counterA$ simulates runs of the master automaton and slave automata by keeping track of the state of the master automaton and
states of up to $k$ active slave automata. If there are less than $k$ active slave automata, $\counterA$ uses $\bot$ to mark slots that can be used in the future to simulate slave automata.
Moreover, it uses counters to simulate the values of slave automata, i.e.,
whenever a slave automaton is activated, $\counterA$ simulates the execution of this automaton and
assigns some counter $i$ to that automaton.
Next, when the simulated slave automaton takes a transition of the weight $m$ the automaton $\counterA$
changes the value of counter $i$ by $m$.
Finally, $\counterA$ terminates counter $i$ when the corresponding slave automaton terminates.
The size of $|\counterA|$ is bounded by $|\nestedA|^k$ and it can be constructed in time $O(|\nestedA|^k)$.

Since $\nestedA$ has width bounded by $k$, the simulating automaton $\counterA$ never runs out of counters to simulate slave automata.
Moreover, as it simulates runs of the master automaton and slave automata of $\nestedA$, there is a one-to-one
correspondence between runs of $\counterA$ and runs of $\nestedA$ and accepting runs of $\nestedA$ correspond to accepting runs of $\counterA$.
Finally, the sequence of weights for the master automaton determined by a given run of $\nestedA$ coincides with the sequence of
weights of $\counterA$ on the corresponding run. Therefore, the values of $\nestedA$ and $\counterA$ coincide on every word.
Thus, non-deterministic $f$-automata with monitor counters  subsume
non-deterministic $(f;\fsum)$-automata of bounded width.

Now, assume that $\nestedA$ is deterministic. Therefore, the master automaton and all slave automata are deterministic and accepting states of slave automata have no outgoing transitions.
We claim that $\counterA$ is deterministic as well.
Consider a state $(s_1, q_1, \ldots, q_k)$ from $Q_{mas} \times {(Q_{s} \cup \{ \bot \})}^k$ and every letter $a \in \Sigma$.
The successor over $a$ of $s_1$ is uniquely determined as the master automaton is deterministic.
For all $q_i$, which are not accepting, the successor states of deterministic slave automata are uniquely determined.
If some state $q_i$ is accepting, then the slave automaton has no outgoing transition and the successor state is $\bot$.
Finally, for $q_i$ equal $\bot$, the one with the least index becomes the initial state of the newly invoked slave automaton and
the other states remain $\bot$.
Therefore, the automaton $\counterA$ is deterministic.
\end{proof}

A direct consequence of Lemma~\ref{l:mc-vs-nested}  is the following theorem:

\begin{thm}%
\label{th:equivalence}
For every  $f \in \InfVal$, deterministic
bounded-width $(f,\fsum)$-automata and deterministic $f$-automata with monitor counters are expressively equivalent.
\end{thm}

\begin{rem}[Discussion]
Theorem~\ref{th:equivalence}
 states that deterministic automata with monitor counters have
the same expressive power as deterministic NWA of bounded width. However, the latter may be exponentially more succinct.
In consequence, lower bounds on deterministic automata with monitor counters imply lower bounds on NWA of bounded width.
Conversely, deterministic NWA can be considered as automata with infinite number of monitor counters, therefore
upper bounds on deterministic NWA imply upper bounds on deterministic automata with monitor counters
\end{rem}

\section{Problems}
\subsection{Classical questions}

The classical questions in automata theory are \emph{emptiness} and \emph{universality} (of a language).
These problems have their counterparts in the quantitative setting of weighted automata and their extensions.
The (quantitative) emptiness and universality problems are defined in the same way for weighted automata, NWA and automata with monitor counters, i.e.,
in the following definition the automaton $\aut$ can be a weighted automaton, an NWA or an automaton with monitor counters.
\begin{itemize}
\item \textbf{Emptiness}: Given an automaton $\aut$ and $\lambda \in \Q$, decide whether there is a word $w$ with
$\valueL{\aut}(w) \leq \lambda$?
\item \textbf{Universality}: Given an automaton $\aut$ and  $\lambda \in \Q$, decide whether for all words $w$ we have
$\valueL{\aut}(w) \leq \lambda$?
\end{itemize}
The universality question asks for \emph{non-existence} of a word $w$ such that $\valueL{\aut}(w) > \lambda$.

\subsection{Probabilistic questions}
The classical questions ask for the (non-)existence of words for input 
automata, whereas in the probabilistic setting, input automata are analyzed w.r.t.
a probability distribution.
We consider probability distributions over infinite words $\Sigma^{\omega}$, and as a finite representation
consider the classical model of Markov chains.

\smallskip
\Paragraph{Labeled Markov chains}.
A \emph{(labeled) Markov chain} is a tuple $\tuple{\Sigma,S,s_0,E}$,
where $\Sigma$ is the alphabet of letters,
$S$ is a finite set of states, $s_0$ is an initial state,
$E \colon S \times \Sigma \times S \mapsto [0,1]$ is the edge probability function, which
for every $s \in S$ satisfies that $\sum_{a \in \Sigma, s' \in S} E(s,a,s') = 1$.

\smallskip
\Paragraph{Distributions given by Markov chains}.
Consider a Markov chain $\markov$. For every finite word $u$, the probability of $u$, denoted $\probability_{\markov}(u)$,
w.r.t.\ the Markov chain $\markov$ is the sum of probabilities of paths labeled by $u$,
where the probability of a path is the product of probabilities of its edges.
For basic open sets $u\cdot \Sigma^\omega = \{ uw : w \in \Sigma^{\omega} \}$,
we have $\probability_{\markov}(u\cdot \Sigma^\omega)=\probability_{\markov}(u)$, and then the
probability measure over infinite words defined by $\markov$ is the unique
extension of the above measure (by Carath\'{e}odory's extension theorem~\cite{feller}).
We will denote the unique probability measure defined by $\markov$ as $\probability_{\markov}$, and
the associated expectation measure as $\expected_{\markov}$.

We define \emph{the uniform probability measure} $\calU$ such that for every $u \in \Sigma^*$ we have
$\probability_{\calU}(u \cdot \Sigma^{\omega}) = |\Sigma|^{-|u|}$. It can be defined by a single-state Markov chain, in which
all transitions are self-loops labled with the same probability $\frac{1}{|\Sigma|}$.

\smallskip
\Paragraph{Automata as random variables}.
Note that deterministic weighted automata,  NWA or automata with monitor counters
all define functions $h \colon \Sigma^\omega \mapsto \R$, which are measurable with respect to probability measures given by Markov chains, and hence
these functions can be interpreted as random variables.
Therefore, given an automaton $\aut$ and a Markov chain $\markov$, we consider the following fundamental quantities:
\begin{enumerate}
\item \textbf{Expected value}: $\expected_{\markov}(\aut)$ is the expected value of
the random variable defined by the automaton $\aut$ w.r.t.\ the probability measure defined by the Markov chain $\markov$.
\item \textbf{(Cumulative) distribution}: $\distrib_{\markov, \aut}(\const) = \probability_{\markov}(\{w \mid \valueL{\aut}(w) \leq \const \})$ is the
cumulative distribution function of
the random variable defined by the automaton $\aut$ w.r.t.\ the probability measure defined by the Markov chain $\markov$.
\end{enumerate}

\smallskip
\Paragraph{Computational questions.}
Given an automaton $\aut$ and a Markov chain $\markov$, we consider the following basic
computational questions:
\begin{enumerate}[label={(Q\arabic*)}]
\item The \emph{expected question} asks to compute  $\expected_{\markov}(\aut)$.
\item The \emph{distribution question} asks, given a threshold $\const \in \Q$, to compute $\distrib_{\markov, \aut}(\const)$.
\end{enumerate}
\smallskip
Questions (Q1) and (Q2) have their approximate variants, which, given an additional input $\epsilon > 0$, ask to compute values that are $\epsilon$-close to
$\expected_{\markov}(\aut)$ or $\distrib_{\markov, \aut}(\const)$, i.e., given $\epsilon > 0$:
\begin{enumerate}[resume,label={(Q\arabic*)}]
\item The \emph{approximate expected question} asks to compute a value $\eta$ such that $|\eta - \expected_{\markov}(\aut)| \leq \epsilon$, and
\item The \emph{approximate distribution question} asks to compute a value $\eta$ such that $|\eta - \distrib_{\markov, \aut}(\const)| \leq \epsilon$.
\end{enumerate}
\smallskip
Additionally, a special important case for the distribution question is
\begin{enumerate}[resume,label={(Q\arabic*)}]
\item The \emph{almost-sure distribution question} asks whether for a given $\const \in \Q$ the probability $\distrib_{\markov, \aut}(\const)$ is exactly $1$.
\end{enumerate}

\noindent
We refer to questions (Q1)--(Q5) as \emph{probabilistic questions}. Note that an upper bound on the complexity of the expected and distribution questions imply 
the same upper bound on all probabilistic questions as approximate and almost-sure variants are special cases.

\begin{exa}[Expected average response time]
Consider an NWA $\nestedA$ from Example~\ref{ex:NWA}. Recall that it computes ART on words it accepts (bounded number of requests between any two grants).
Next, consider a Markov chain $\markov$ which gives a distribution on words over $\{r,g,i\}$.
In such a case, the value $\expected_{\markov}(\nestedA)$ is the expected ART\@.
\end{exa}

\section{Results on classical questions}
	\smallskip\noindent{\bf Existing results.}
The complexity of the classical decision problems for NWA
has been established in~\cite{nested} which is presented in Table~\ref{tab1}.
\begin{table}[t]
\centering
\def\tabcolsep{5pt}
\begin{tabular}{|c|c|c|c|c|} 
\hline 
\multicolumn{2}{|c|}{}& $\finf$& $\fsup$ & \multirow{2}{*}{$\flimavg$} \\
\multicolumn{2}{|c|}{}& $\fliminf$ & $\flimsup$  & \\
\hline 
$\fmin, \fmax$ &Empt.&
\multicolumn{3}{|c|}{ {$\PSPACE$-comp.}}  \\
\cline{2-2}
\cline{4-4}
$\fBsum{B}$ &Univ.&  &\PTIME&  \\
\hline 
\multirow{2}{*}{$\fsum^+$} & Empt. &
\multicolumn{2}{|c|}{ {$\PSPACE$-comp.}} &
\multirow{2}{*}{$\EXPSPACE$ } \\
\cline{2-2}
\cline{4-4}
&Univ. & &\PTIME&    \\
\hline 
\multirow{2}{*}{$\fsum$} & Empt. & \multirow{2}{*}{$\PSPACE$-comp.}  & Undecidable & Open \\
\cline{2-2}
\cline{4-4}
& Univ. &  & $\PTIME$  &  \\
\hline 
\end{tabular}
\caption{Decidability and complexity of emptiness and universality for deterministic $(f;g)$-automata.
Functions $f$ are listed in the first row and functions $g$ are in the first column.
}%
\label{tab1}
\end{table}

\smallskip\noindent{\bf New results}.
Due to Lemma~\ref{l:mc-vs-nested},
decidability of deterministic $(f;\fsum)$-automata implies decidability of
deterministic automata with monitor counters with the value function $f$.
However, the undecidability result of NWA does not imply undecidability for
automata with monitor counters as the NWA in the reduction may have unbounded width.
We present the undecidability result for NWA of bounded width, which implies undecidability of the emptiness problem for automata with monitor counters.
Thus, our following result completes the decidability picture also for automata with monitor counters
(i.e., the decidability results coincide with the $\fsum$ row of Table~\ref{tab1}).

\begin{thm}%
\label{th:undecidable-limsup}
The emptiness problem is undecidable for deterministic $\fsup$-automata (resp., $\flimsup$-automata) with $8$ monitor counters.
\end{thm}
\begin{proof}
We show undecidability of the emptiness problem for  deterministic $(\flimsup,\fsum)$-automata of width $8$.
The proof for deterministic $(\fsup,\fsum)$-automata is virtually the same.
Then, the theorem follows from the translation lemma (Lemma~\ref{l:mc-vs-nested}).

We show a reduction from the halting problem for deterministic two-counter machines, which is undecidable~\cite{minsky1961recursive}.
Let $\M$ be a deterministic two-counter machine and let $Q$ be the set of states of $\M$.
We define a deterministic $(\flimsup,\fsum)$-automaton $\nestedA$ of width $8$ such that
$\nestedA$ has a run of the value not exceeding $0$ if and only if $\M$ has an accepting computation.

Consider the alphabet $\Sigma = Q \cup \{ 1,2,\#,\$\}$.
We encode computations of $\M$ as a sequence of configurations separated by $\#$.
A single configuration of $\M$, where the machine is in the state $q$, the first counter has the value
$x$ and the second $y$ is encoded by the word $q 1^{x} 2^{y}$. Finally,
computations of $\M$ are separated by $\$$. We define the automaton $\aut$ that for a word 
$w \in \Sigma^*$ returns the value $0$ if (some infinite suffix of) $w$ encodes a sequence valid accepting computations of $\M$.
Otherwise, $\nestedA$ returns the value at least $1$.

The automaton $\nestedA$ works as follows. On a single computation, i.e., between symbols $\$$,  $\nestedA$ 
checks consistency of the transitions by checking two conditions: (C1)~Boolean consistency, and (C2)~counter consistency.
The condition (C1) states that encoded subsequence configurations, which are represented by subwords  $q 1^x 2^y \# q' 1^{x'} 2^{y'}$,
 are consistent with the transition function of $\M$ modulo counter values, i.e., under counter abstraction to values $0$ and ``strictly positive''.
Observe that a finite automaton can check that. The conditions that need to be checked are as follows:
(C1-1)~Boolean parts of transitions are consistent; the automaton checks only emptiness/nonemptiness of counters and 
based on that verifies whether a subword $q 1^x 2^y \# q'$ is valid w.r.t.\ transitions of $\M$.
For example, consider transition $(q,\bot, +, q', +1, -1)$ of $\M$ stating that ``if $\M$ is in state $q$, the first counter
is $0$ and the second counter is positive, then change the state to $q'$ increment the first counter and
decrement the second one''. This transition corresponds to the regular expression
$q 2^+ \# q'$.
(C1-2)~The initial and final configurations in each computation (between $\$$ symbols) 
are respectively $q_I 1^0 2^0$ and $q_f 1^0 2^0$.
(C1-3)~The word encodes infinitely many computations, i.e., the word contains infinitely many $\$$ symbols. 
The last conditions rejects words encoding non-terminating computations.

To check the condition (C2), $\nestedA$ uses slave automata.
It uses $4$ slave automata to check transitions between even and odd positions and the other $4$ slave automata to check validity of the remaining
transitions. Then, between even and odd positions it uses $2$ slave automata for each counter of $\M$.
These slave automata encode the absolute values between the intended values of counters (i.e., assuming that counter values are consistent with the instructions)
and the actual values.
For example, for a subword $q 1^x 2^y \# q' 1^{x'} 2^{y'}$, the automaton $\nestedA$ checks whether the value
of counter $1$ is consistent with transition $(q,\bot, +, q', +1, -1)$ in the following way.
The first slave automaton ignores letters $2$ and initially decrements its value at every letter $1$ until it reads letter $\#$ (where its value is $-x$).
Next, it switches its mode and increments its value at letters $1$ while ignoring letters $2$.
In that way its value upon reading $q 1^x 2^y \# q' 1^{x'} 2^{y'}$ equals $-x + x'$.
Finally, it increments its value by $1$.
Thus, the value of the slave automaton is $-x +x' +1$.
The second slave automaton works in a similar way, but it decrements whenever the first counter increments and vice versa.
 At the end, the value of the second slave automaton is
$x - x' -1$. Observe that the maximum of these value is $|x - (x'+1)|$, which is $0$ if and only if the value of counter $1$ is consistent with the transition
 $(q,\bot, +, q', +1, -1)$. It follows that the supremum over the values returned by all slave automata is $0$ only if all counter values are consistent with the transitions.
Therefore, the value $\flimsup$ of the whole word  is $0$ if and only if starting at some point all computations are valid and accepting.
The latter is possible only if $\M$ has at least one such a computation.
Otherwise, the value of $\flimsup$ is at least $1$.

Observe that this construction works for $\fsup$ as well.
\end{proof}

\section{Basic results on probabilistic questions}
In this section we discuss basic properties of the probabilistic questions and present some basic facts about Markov
chains. Next, in the following Section~\ref{s:limit} and Section~\ref{s:nolimit} we study the probabilistic questions for NWA\@.
We consider there separately NWA with $\fliminf, \flimsup, \flimavg$ value functions for the master automaton (Section~\ref{s:limit})
and NWA with $\finf, \fsup$ value functions for the master automaton (Section~\ref{s:nolimit}).

We begin with the discussion on the acceptance by an NWA\@.

\subsection{Property about almost-sure acceptance}

Observe that if the probability of the set of words rejected by an automaton $\aut$
is strictly greater than $0$, then the expected value of such an automaton
is infinite or undefined.
In the next lemma we show that given a deterministic NWA $\nestedA$ and
a Markov chain $\markov$ we can decide in polynomial time
whether the NWA is \emph{almost-surely accepting}, i.e., the set of words whose runs are accepting has probability~1.
In Section~\ref{s:limit}  we consider all the computational problems for NWA which are almost-surely accepting.
This assumption does not influence the complexity of computational questions
related to the expected value, but has an influence on
the complexity of distribution questions, which we discuss in Section~\ref{s:no-almost-surely-accepting}.

\begin{restatable}{prop}{AcceptAlmostAllInP}%
\label{prop:almostAll}
Given a deterministic NWA $\nestedA$ and a Markov chain $\markov$, we can decide in
polynomial time whether $\probability_{\markov}(\{w \mid \Acc(w) \neq \emptyset\})=1$?
\end{restatable}
\begin{proof}
First, the master automaton has to accept almost all words. We can check this in polynomial time by considering the master automaton as a \buchi{} automaton and applying the classical methods~\cite{BaierBook}.

For all pairs $(q,s)$, where $q$ is the initial state of some slave automaton $\slaveA_i$ and $s$ is a state of the Markov chain $\markov$,
we check that either $\slaveA_i$ is never invoked while $\markov$ is in the state $s$ or $\slaveA_i$ almost-surely accepts (w.r.t.\ the distribution given by $\markov$
started in $s$).
Observe that $\slaveA_i$ almost-surely accepts if for every finite word $u$ generated by $\markov$ starting in the state $s$ there exists its finite extension $u'$ (generated by $\markov$), which is accepted by $\slaveA_i$
(i.e., $\slaveA_i$ terminates in an accepting state and returns a finite value).
One can easily check that this condition is necessary and sufficient, and it can be checked in polynomial time~\cite{BaierBook}.
\end{proof}

\begin{rem}[Almost-sure acceptance]%
\label{rem:almost-sure}
The answer to the expected value problem does not change even without
the assumption.
We show next that without the almost-sure acceptance condition,
the distribution questions become similar to $\finf$ and $\fsup$ value functions.
Hence, in Section~\ref{s:limit} we consider the almost-sure acceptance property,
and presented the conceptually interesting results.
Moreover, classically weighted automata have been considered without any
acceptance conditions (i.e., all words are accepted), and then
the almost-sure acceptance is trivially ensured.
\end{rem}

\subsection{Duality property between infimum and supremum}
In Section~\ref{s:limit} and Section~\ref{s:nolimit}, when we consider the expected value and the distribution,
in most cases we consider only $\finf$ and $\fliminf$ value functions,
and by duality, we obtain results
for $\fsup$ and $\flimsup$ value functions, respectively.
The only exception are $(\finf, \fsum^+)$-automata and $(\fsup, \fsum^+)$-automata,
which have to be considered separately.
For every value function $g \in \FinVal \setminus \{\fsum^+\}$ we define $-g$ as  follows:
$-{\fmin} = \fmax, -{\fmax} = \fmin$ and
$-{g} = g$ for $g \in \{\fBsum{B}, \fsum\}$.

\begin{restatable}{lem}{Duality}%
\label{l:sup-to-inf}
For every $g \in \FinVal \setminus \{\fsum^+\}$, every deterministic $(\fsup; g)$-automaton (resp. $(\flimsup; g)$-automaton) $\nestedA_1$
accepting almost-all words can be transformed to a deterministic $(\finf; -{g})$-automaton (resp. $(\fliminf; -{g})$-automaton) $\nestedA_2$ of the same size such that for almost all words $w$ we have
$\valueL{\nestedA_1}(w) = -\valueL{\nestedA_2}(w)$.
\end{restatable}
\begin{proof}
The automaton $\nestedA_2$ is obtained from $\nestedA_1$ by multiplying all the weights by $-1$.
\end{proof}

\begin{rem}[Limited duality]%
\label{rem:duality}
As we work under the almost-sure acceptance assumption, the above duality result implies that all complexity results for NWA with
$\finf$ (resp,. $\fliminf$) master value function transfer to NWA with $\fsup$ (resp., $\flimsup$) master value function.
However, this duality does not extend to the classical problems of emptiness and universality.
Indeed, if $\nestedA_1$ does not have an accepting run over $w$ then
$\valueL{\nestedA_1}(w) = \valueL{\nestedA_2}(w) = \infty$ (where $\nestedA_1, \nestedA_2$ are from Lemma~\ref{l:sup-to-inf}).
This cannot be fixed with a different construction as the master automaton in NWA has \buchi{} acceptance conditions and deterministic \buchi{} automata are not closed under complementation.
This leads to different complexity results for these automata.
In particular, the emptiness problem for deterministic $(\fsum,\fsum)$-automata is undecidable, while  the universality problem for deterministic  $(\finf,\fsum)$-automata is $\PSPACE$-complete (Table~\ref{tab1}).
\end{rem}

\subsection{Basic facts about Markov chains}
\Paragraph{Labeled Markov chains with weights.}
A labeled Markov chain with weights is a (labeled) Markov chain $\markov$ with a function $r$, which
associates rationals with edges of $\markov$.
Formally, a \emph{(labeled) Markov chain with weights} is a tuple $\tuple{\Sigma,S,s_0,E,r}$,
where $\tuple{\Sigma,S,s_0,E}$ is a labeled Markov chain and
$r \colon S \times \Sigma \times S \mapsto \Q$.

\smallskip
\Paragraph{Graph properties on Markov chains}.
Standard graph notions have their counterparts on Markov chains by considering edges
with strictly positive probability as present and edges with probability $0$  as absent.
For example, we consider the following graph notions:
\begin{itemize}
\item \textbf{(reachability)}: A state $s$ is \emph{reachable} from $s'$ in a Markov chain if there exists a sequence of edges with positive probability
starting in $s' $ and ending in $s$.
\item \textbf{(SCCs)}: A subset of states $Q$ of a Markov chain is a \emph{strongly connected component} (SCC) if and only if
from any state of $Q$ all states in $Q$ are reachable.
\item
\textbf{(bottom SCCs)}: An SCC $Q$ is a \emph{bottom} SCC if and only if
there are no edges leaving $Q$.
\end{itemize}

\smallskip
\Paragraph{The product of an automaton and a Markov chain.}
Let $\aut = \tuple{ \Sigma, Q, q_0, \delta, F, \cost}$ be a deterministic weighted automaton
and let $\markov = \tuple{\Sigma,S,s_0,E, r}$ be a Markov chain.
We define the product of $\aut$ and $\markov$, denoted by $\aut \times \markov$, as a Markov chain
$\tuple{ \Sigma, Q \times S, \tuple{q_0, s_0}, E', r'}$, where
(1)~$E'(\tuple{q_1, s_1}, a, \tuple{q_2, s_2}) = E(s_1, a, s_2)$ if $(q_1, a, q_2) \in \delta$ and
$E'(\tuple{q_1, s_1}, a, \tuple{q_2, s_2}) = 0$ otherwise, and
(2)~$r'(\tuple{q_1, s_1}, a, \tuple{q_2, s_2}) = \cost(q_1, a, q_2) + r(s_1, a, s_2)$.

The expected value and distribution questions can be answered in polynomial time for deterministic weighted automata
with value functions from $\InfVal$~\cite{ChatterjeeDH09LimInf}.

\begin{fact}%
\label{t:weighted-inf-expected}%
\label{t:weighted-limavg-expected}
Let $f \in \InfVal$.
Given a  Markov chain $\markov$, a deterministic $f$-automaton $\aut$ and a value $\const$,
the values $\expected_{\markov}(\aut)$ and $\distrib_{\markov, \aut}(\const)$ can be computed
in polynomial time.
\end{fact}

\section{Results on limit value functions}%
\label{s:limit}
In this section we study NWA with $\fliminf, \flimsup$ and $\flimavg$ value functions for the master automaton.
All these value functions are prefix independent and hence in a (deterministic) strongly-connected almost-surely accepting NWA,
returning a value $\lambda$ is a tail event, which has probability either $0$ or $1$.
It follows that almost all words have the same value.
We use this property to establish polynomial-time algorithms for all probabilistic questions.

Throughout this section we assume that all NWA that are almost-surely accepting, i.e., for almost all words $w$, the run on $w$ is accepting.
In the classical setting of weighted automata, which have no
accepting condition, the almost-sure acceptance is trivially satisfied.
This is a conceptually interesting case as we are mainly interested in the quantitative aspect of (nested) weighted automata.
Moreover, we can check whether a given deterministic NWA is almost-surely accepting in polynomial time (Proposition~\ref{prop:almostAll}).
If it is not, the expected value is either infinite or undefined and hence the complexity of the expected question does not change.
However, the complexity of the distribution question changes and we discuss it in Section~\ref{s:no-almost-surely-accepting}.

\subsection{LimInf and LimSup value functions}%
\label{s:liminf}

In this section we study NWA with $\fliminf$ and $\flimsup$ value functions for the master automaton.
We start with a result for the special case when the master automaton is
strongly connected w.r.t.\ the Markov chain.

\smallskip
\Paragraph{An automaton strongly connected on a Markov chain.}
We say that a deterministic automaton $\aut$ is \emph{strongly connected on} a Markov chain $\markov$ if and only if
the states reachable (with positive probability) in $\aut \times \markov$ from the initial state
form an SCC\@.

\proofideas{}
The value functions $\fliminf$ and $\flimsup$ return values that occur infinitely often. Therefore, in a strongly connected Markov chain,
for every finite word $u$, the set of infinite words that contain $u$ infinitely many times has probability $0$ or $1$.
We extend this property to establish some sort of 0-1 law for NWA with $\fliminf$ or $\flimsup$ master value function (Lemma~\ref{l:in-scc-all-equal}), which states that 
 if the product of the Markov chain and the master automaton of $\nestedA$ form an SCC, then almost all words have the same value
which is the infimum returned by slave automata of $\nestedA$.

In the following result we do not assume that a given NWA accepts almost surely.

\begin{restatable}{lem}{stronglyConnectedComponenets}%
\label{l:in-scc-all-equal}
Let $g \in \FinVal$, $\markov$ be a Markov chain, and
$\nestedA$ be a deterministic $(\finf;g)$-automaton (resp., $(\fliminf; g)$-automaton).
Assume that the master automaton of $\nestedA$ is strongly connected on $\markov$.
Then, the following conditions hold:
\begin{enumerate}
\item either runs on almost all words are accepting or runs on almost all words are rejecting,
\item there exists a unique value $\lambda$ such that $\probability_{\markov}(\{ w \mid \valueL{\nestedA}(w) = \lambda \}) = 1$,
\item $|\lambda| \leq |\nestedA|\cdot |\markov|$ or $\lambda$ is extreme, i.e.,  $\lambda \in \{-\infty,\infty\}$ for $g \in \{\fsum,\fsum^+\}$ or $\lambda  \in \{ -B, B,\infty\}$ for $g = \fBsum{B}$, and
\item given $\markov$ and $\nestedA$, the value $\lambda$ can be computed in polynomial time in $|\markov| + |\nestedA|$.
\end{enumerate}
\end{restatable}
\begin{proof}
We first show duality (1) and then consider the case of almost-surely accepting NWA.\smallskip

\noindent\emph{Duality}.
Since $\nestedA$ is deterministic, all runs of $\nestedA$ on the distribution given by $\markov$ correspond to
the paths in $\masterA \times \markov$, where $\masterA$ is the master automaton of $\nestedA$.
Since $\masterA \times \markov$ is strongly connected, any finite path occurs infinitely often with the probability either $0$ or $1$~\cite{BaierBook}.
Therefore, either almost all runs satisfy the \buchi{} condition of $\masterA$ or almost all runs violate it.
We can check in polynomial time which of these two cases holds.

Now, either in almost all words all slave automata accept or
there is a state $(q,s)$ of $\masterA \times \markov$, where the master automaton invokes some slave automaton $\slaveA_i$,
which does not accept with positive probability. In the latter case, almost all words are rejected by $\nestedA$.
Observe that there exists a finite word $u$ generated by $\markov$ in state $s$ such that no finite extension of $u$ (generated by $\markov$ in state $s$) is accepted by $\slaveA_i$.
Again, since $\masterA \times \markov$ is strongly connected, the sequence of states $(q,s)$
followed by states forming the word $u$ occurs infinitely often in almost all words.
It follows that on almost all words $\nestedA$ does not have an accepting run.
We can check existence of $(q,s)$, which is visited infinitely often in $\masterA \times \markov$ and some invoked slave automaton rejects with positive probability, in polynomial time in $|\nestedA| + |\markov|$.

Recall that for a word $w$, the run of $\nestedA$ on $w$ is accepting if and only if
 the run of $\masterA$ on $w$ is accepting (it satisfies its \buchi{} condition), and
the runs of all invoked slave automata are accepting. In summary, one of the following holds:
(i) on almost all words $w$, the run of $\nestedA$ on $w$ is accepting, or
(ii) on almost all words $w$, the run of $\nestedA$ on $w$ is rejecting, and hence $\valueL{\nestedA}(w) = \infty$.
We can decide in polynomial time which of these cases holds. Moreover, if (ii) holds, then $\probability_{\markov}(\{ w \mid \valueL{\nestedA}(w) = \infty \}) = 1$, i.e., $\lambda = \infty$.
We next assume that (i) holds.
\smallskip

\noindent\emph{Almost-surely accepting NWA}.
Assume that $\nestedA$ is almost-surely accepting.
 Consider a state $(q,s)$ of $\masterA \times \markov$, where the master automaton invokes some slave automaton $\slaveA_i$, and the slave automaton $\slaveA_i$ attains its minimal value on the following letters.
(If $\slaveA_i$ does not attain its minimal value, we consider a sequence that tends to $-\infty$.)
Therefore, almost all runs contain the considered sequence infinitely often.
It follows that the value of almost all runs is the minimum over reachable states $(q,s)$ from $\masterA \times \markov$ and transitions $(s,a,s')$ of $\markov$
of the minimal value the slave automaton invoked in $(q,a,q')$ can achieve on all words generated by $\markov$ starting with the transition $(s,a,s')$.
This value can be computed in polynomial time in $|\markov| + |\nestedA|$.

Observe that either some invoked slave automaton can reach a cycle with the sum of weights being negative and iterate over it (i.e., a cycle in $\slaveA_i \times \markov$),
and hence the minimum is $-\infty$ (the minimum is bounded by $-B$ for  $g = \fBsum{B}$).
Otherwise, the minimum is attained over some word which does not form a cycle in $\slaveA_i \times \markov$, i.e., the length of this word is bounded by the number of states of $\slaveA_i$ times the size of $\markov$.
Since we consider weights to be given in the unary notation,
the sum of weights over such a word is bounded by $|\slaveA_i|$ times $|\markov|$.
Thus, $|\lambda| \leq |\nestedA| \cdot |\markov|$ or $\lambda = -\infty$ (resp., $-B$ for  $g = \fBsum{B}$).
If  $g = \fBsum{B}$ and $B < |\nestedA| \cdot |\markov|$, then $\lambda = B$.
\end{proof}

Lemma~\ref{l:in-scc-all-equal} implies the following main lemma of this section.

\proofideas{}
Consider a $(\fliminf; g)$-automaton (resp., $(\flimsup; g)$-automaton)
$\nestedA$ that accepts almost all words.
The value $\valueL{\nestedA}(w)$ depends only on the infinite behavior of the (unique) run of $\nestedA$ on $w$, which ends up in some bottom SCC (for almost all words $w$).
In a bottom SSC, almost all words have the same value, which can be computed in polynomial time (Lemma~\ref{l:in-scc-all-equal}).
Thus, to compute $\expected_{\markov}(\nestedA)$, we compute probabilities of reaching each of the
bottom SCCs and values of $\nestedA$ in these SSCs.
In a similar way, we can compute $\distrib_{\markov,\nestedA}(\const)$.

\begin{restatable}{lem}{liminfIsPolynomial}%
\label{th:liminfIsPoly}
Let $g \in \FinVal$.
For a deterministic almost-surely accepting  $(\fliminf;g)$-automata (resp., $(\flimsup; g)$-automata) $\nestedA$ and a Markov
chain $\markov$, given a threshold $\const$, both $\expected_{\markov}(\nestedA)$
and $\distrib_{\markov,\nestedA}(\const)$ can be computed in polynomial time.
\end{restatable}
\begin{proof}
First, we discuss how to compute the expected and the distribution questions of a deterministic $(\fliminf; \fsum)$-automaton $\nestedA$.

The value of $(\fliminf; \fsum)$-automaton $\nestedA$ on a word depends on weights that appear infinitely often.
Since $\nestedA$ reaches some bottom SCC with probability $1$,
we can neglect values of slave automata returned before the master automaton $\masterA$ (of $\nestedA$)
reaches a bottom SCC of $\masterA \times \markov$.
Thus, the expected value of $(\fliminf; \fsum)$-automaton $\nestedA$ w.r.t.\ a Markov chain $\markov$ can be computed in the following way.
Let $S_1, \ldots, S_l$ be all bottom SCCs of $\masterA \times \markov$.
We compute probabilities $p_1, \ldots, p_l$ of reaching the components $S_1, \ldots, S_l$ respectively.
These probabilities can be computed in polynomial time~\cite{BaierBook}.
Next, for every component $S_i$ we compute in polynomial time the unique value $m_i$, which $\nestedA$ returns on almost every word
whose run ends up in $S_i$ (Lemma~\ref{l:in-scc-all-equal}).
The expected value $\expected_{\markov,\nestedA}$ is equal to $p_1 \cdot m_1 + \cdots + p_l \cdot m_l$.
Observe that, given a value $\const$, the distribution $\distrib_{\markov, \nestedA}(\const)$ is equal to
the sum the probabilities $p_i$ over such $i$ that $m_i \leq \const$.
Hence, the expected and the distribution questions can be computed in polynomial time.

Due to Lemma~\ref{l:sup-to-inf},  the case of $\flimsup$ reduces to the case of $\fliminf$.
All value functions from $\FinVal$ are special cases of $\fsum$.
This concludes the proof.
\end{proof}

Lemma~\ref{th:liminfIsPoly} states the the expected question and the distribution question can be computed in polynomial time.
The remaining probabilistic questions are their  special cases and hence they can be computed in polynomial time as well.
The following theorem summarizes results of this section.

\begin{thm}%
\label{th:compLimInf}
Let $g \in \FinVal$.
All probabilistic questions for deterministic almost-surely accepting $(\fliminf; g)$-automata (resp., $(\flimsup; g)$-automata)
can be solved in polynomial time.
\end{thm}

\begin{rem}[Contrast with classical questions]%
\label{remark:LimInf-classical-vs-probabilistic}
Consider the results on classical questions shown in Table~\ref{tab1} and the results for
probabilistic questions we establish in Theorem~\ref{th:compLimInf}.
While for classical questions the problems are $\PSPACE$-complete or
undecidable, we establish polynomial-time algorithms for all probabilistic questions.
\end{rem}

\subsection{The expected question for the LimAvg value function}%
\label{s:limavg}

In this section we study NWA with the $\flimavg$ value function for the master automaton.
We essentially show that to compute the expected value of a given $(\flimavg;g)$-automaton, it suffices to
substitute in each transition invoking a slave automaton $\slaveA_i$ by the expected value of $\slaveA_i$.

We assume that considered $(\flimavg;g)$-automata are deterministic and accept almost all words.
We discuss the case of almost-surely accepting NWA\@.
This assumption does not change the complexity of the expected question (Remark~\ref{rem:almost-sure}).

\begin{lem}%
\label{l:limavg-poly}
Let $g \in \FinVal$.
Given a Markov chain $\markov$ and a deterministic almost-surely accepting $(\flimavg;g)$-automaton $\nestedA$, the value $\expected_{\markov}(\nestedA)$ can be computed in polynomial time.
\end{lem}

\noindent\emph{Overview}.
We present the most interesting case when $g=\fsum$.
Let $\nestedA$ be a $(\flimavg;\fsum)$-automaton and let $\markov$ be a Markov chain.
We define a weighted Markov chain $\MCfromNested$ as the product $\masterA \times \markov$,
where $\masterA$ is the master automaton of $\nestedA$.
The weights of  $\MCfromNested$ are  the expected values of invoked slave automata, i.e.,
the weight of the transition $\tuple{(q,s),a,(q',s')}$ is the expected value of $\slaveA_i$,
the slave automaton started by $\masterA$ in the state $q$ upon reading $a$, w.r.t.\ the distribution given by $\markov$ starting in $s$.

In the remaining part of this section we show
that the expected value of $\nestedA$ w.r.t. $\markov$ and the expected value of $\MCfromNested$ coincide (Lemma~\ref{l:limavgReducesToMC}).
The Markov chain $\MCfromNested$ can be computed in polynomial time and
has polynomial size in $|\nestedA| + |\markov|$. Thus, we can compute the expected values of $\MCfromNested$, and
in turn $\expected_{\markov}(\nestedA)$, in  polynomial time in $|\nestedA| + |\markov|$.

\begin{lem}%
\label{l:limavgReducesToMC}
Let $\nestedA$ be a deterministic almost-surely accepting  $(\flimavg;\fsum)$-automaton.
The values $\expected_{\markov}(\nestedA)$  and $\expected(\MCfromNested)$ coincide.
\end{lem}

In the following we prove Lemma~\ref{l:limavgReducesToMC}.
First, we show that lemma for $(\flimavg;\fsum)$-automata in which
duration of runs of slave automata is bounded by some $N \in \N$.
Next, we show how to solve the general case of all $(\flimavg;\fsum)$-automata by the reduction to
this special case.

Note that $\MCfromNested$ can have silent moves labeled by $\bot$. Indeed, an automaton that starts in the accepting state always returns value $\bot$, which is its expected value.
Before we continue, we discuss computing the expected values of Markov chains with silent moves.

\smallskip
\Paragraph{Expected limit averages  of Markov chains with silent moves}.
\newcommand{\silentMarkov}{\markov_{\textrm{sil}}}
\newcommand{\Path}{\rho}
Let $\silentMarkov$ be a Markov chain labeled by $\Q \cup \{ \bot \}$, where $\bot$
corresponds to a silent transition. We consider the limit average value function with silent moves $\silent{\flimavg}$,
which applied to a sequence $a_1 a_2 \ldots$ of elements of $\Q \cup \{\bot\}$ removes all $\bot$ symbols are
applies the standard $\flimavg$ function (defined on the sequences of rational numbers in the similar way as for integers) to the sequence consisting of the remaining elements.
The expected value of the limit average of a path in $\silentMarkov$ can be computed by a slight modification of the standard method
for Markov chains without silent transitions~\cite{filar}.

Without loss of generality we can assume that $\silentMarkov$ is strongly connected. It if is not, we can  compute bottom strongly connected components
$B_1, \ldots, B_k$ of $\silentMarkov$, then compute probabilities $p_1, \ldots, p_k$ of reaching these
components and $\expected(\silentMarkov) = \sum_{i=1}^k p_i \expected(\silentMarkov[B_i])$, where $\silentMarkov[B_i]$ is $\silentMarkov$ with the initial state being some state from $B_i$.

Assume that  $\silentMarkov$ is strongly connected and contains non-silent transitions.
We associate with each transition $e = (s,a,s')$ of $\silentMarkov$ a real-valued variable $x[e]$, which is the frequency of transition $e$.
Formally, given an infinite path $\Path$ in $\silentMarkov$ we define $|\Path[1,n]|_e$ as the number of transitions $e$ among first $n$ transitions of $\Path$.
Let $e_1, \ldots, e_k$ be all non-silent transitions in $\silentMarkov$.
We state a system of equations and inequalities such that for almost all infinite paths $\Path$
in $\silentMarkov$ and all $i \in \{1, \ldots, k\}$ we have
\begin{equation}
\lim_{n \to \infty} \frac{|\Path[1,n]|_{e_i} }{|\Path[1,n]|_{e_1} + \cdots + |\Path[1,n]|_{e_k}} = x[e_i].
\label{eq:correctness-silent}
\end{equation}
These equations and inequalities are as follows:
\begin{enumerate}[label={(E\arabic*)}]
\item\label{E1} for every transition $e = (s,a,s')$ we put
\[
x[(s,a,s')] = E(s,a',s') \cdot \sum_{s'' \in S_{\silentMarkov}, a' \in \Sigma} x[(s'',a',s)],
\]
the frequency of $(s,a,s')$ is the probability of taking $(s,a,s')$ from $s$ multiplied by the sum of frequencies of all transitions leading to $s$,
\item\label{E2} $x[e_1] + \cdots + x[e_k] = 1$, where
the sum of frequencies of all non-silent transitions is $1$,
\item\label{E3} $0 \leq x[e]$ for every transition $e$.
\end{enumerate}

\noindent
Following the argument for Markov chains without silent moves~\cite{filar,BaierBook}, we can show that the above system of equations has the unique solution and it satisfies~\eqref{eq:correctness-silent}.
Then, the expected limit average of $\silentMarkov$ is given as
$c(e_1) \cdot x[e_1] + \cdots + c(e_k) \cdot x[e_k]$, where $c(e_i)$ is the cost of transition $e_i$.

\subsubsection{The expected value in the bounded-duration case}

First, we show that  Lemma~\ref{l:limavgReducesToMC} holds if we assume that for some $N>0$ all slave automata take at most $N$ transitions.

\begin{lem}%
\label{l:limavgReducesToMC-bounded-width}
Let $\nestedA$ be an almost-surely accepting deterministic $(\flimavg;\fsum)$-automaton
in which duration of runs of slave automata is bounded by $N$
and let $\MCfromNested$ be the weighted Markov chain corresponding to $\nestedA$.
The values $\expected_{\markov}(\nestedA)$  and $\expected(\MCfromNested)$ coincide.
\end{lem}

Before we proceed with the proof of Lemma~\ref{l:limavgReducesToMC-bounded-width}, we present an example.

\begin{exa}%
\label{ex:N-ART}
Recall the average response time property (ART) presented in Example~\ref{ex:intro}.
We consider a variant of ART called $N$-bounded ART\@. We define the \emph{$N$-bounded response time} of a request as the minimum of $N$ and the number of steps to the following grant.  The $N$-bounded ART is the limit average of $N$-bounded response times over all requests.
The $N$-bounded ART property can be computed by a $(\flimavg;\fsum^+)$-automaton $\nestedA_N$ that at each request invokes a slave automaton that takes at most $N$ steps and computes the $N$-bounded response time.
The NWA $\nestedA_N$ invokes a dummy slave automaton on the remaining transitions, which corresponds to taking a silent transition. For simplicity, we restrict the events to requests and grants only (no idle events).

We consider the uniform probability measure over ${\{r,g\}}^{\omega}$ (without events $i$), which can be given by a single-state Markov chain $\markov_U$. First, the expected $N$-bounded response time equals
\[ \sum_{i=1}^{N-1} {\left(\frac{1}{2}\right)}^{i} \cdot i + N \sum_{i=N}^{\infty} {\left(\frac{1}{2}\right)}^i = 	2 - (N-1)\cdot {\left(\frac{1}{2}\right)}^{N-1} + N \cdot  {\left(\frac{1}{2}\right)}^{N-1} = 2 -  {\left(\frac{1}{2}\right)}^{N-1}.\]
The Markov chain $\markov^{\nestedA_N}$ has a single state with a  self-loop labeled by a grant of the empty weight $\bot$ and a self-loop labeled by a request of weight $ 2 - {(\frac{1}{2})}^{N-1}$. Therefore, $\expected(\markov^{\nestedA_N}) = 2 - {(\frac{1}{2})}^{N-1}$.

Now, to compute $\expected_{\markov_U}(\nestedA_N)$ we construct a $\silent{\flimavg}$-automata $\aut_N$ that works as follows.
In each block $r^* g$, for the first $N$ requests, the automaton $\aut_N$ assigns weights $1, 2, \ldots, N$, and then
for the following requests it assigns weight $N$.
 It takes silent transitions over grants $g$. The automaton $\aut_N$ is depicted in Figure~\ref{fig:autN}.

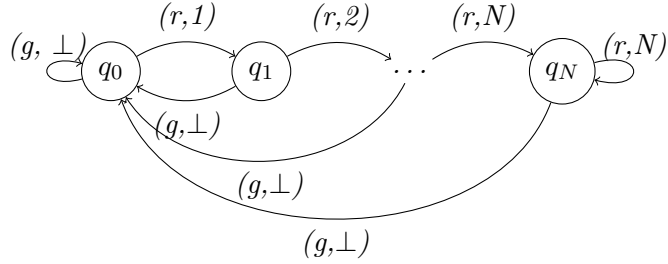
\begin{figure}
\begin{center}
\begin{tikzpicture}
\node[draw, circle] (Q0)    at (0,0) {$q_0$};
\node[draw, circle] (Q1)    at (2,0) {$q_1$};
\node[draw, circle] (QN)   at (6,0) {$q_N$};

\node[] (dots) at (4,0) {\ldots};

\draw[bend left,->] (Q0) to node[above] {$(r,1)$} (Q1);
\draw[bend left,->] (Q1) to node[above] {$(r,2)$} (dots);
\draw[bend left,->] (dots) to node[above] {$(r,N)$} (QN);

\draw[loop left, ->] (Q0) to node[left] {$(g, \bot)$} (Q0);
\draw[loop right, ->] (QN) to node[right] {$(r,N)$} (QN);

\draw[bend left,->] (Q1) to node[below ] {$(g,\bot)$} (Q0);
\draw[bend left= 80,->] (dots) to node[below] {$(g,\bot)$} (Q0);
\draw[bend left = 100,->] (QN) to node[below] {$(g,\bot)$} (Q0);

\end{tikzpicture}
\end{center}
\caption{The automaton $\aut_N$}%
\label{fig:autN}
\end{figure}

Observe that in each block $r^k g$ the $N$-bounded response times are
\[
    \min(N,k), \min(N,k-1), \ldots, 2, 1,
\]
while the weights returned by the automaton $\aut$ are
\[
    1,2, \ldots, \min(N,k-1), \min(N,k).
\]
On all words $w$ with infinitely many grants  the values $\lang_{\nestedA_N}(w)$ and $\lang_{\aut_N}(w)$ are equal.
Therefore, $\expected_{\markov_U}(\nestedA_N) = \expected_{\markov_U}(\aut_N)$.

We compute the $\expected_{\markov_U}(\aut_N)$ in the standard way. Let $x_i$ be the density of visiting state $q_i$ in $\aut_i$.
Clearly, for $i=0, \ldots, N-1$, we have $x_i = {(\frac{1}{2})}^{i}$ and $x_N = {(\frac{1}{2})}^{N-1}$.
Since transitions labeled with $g$ are silent,
$\expected_{\markov_U}(\aut_N) = (\sum_{i=0}^{N-1} x_i \cdot (i+1)) + N \cdot x_N =  2 -  {(\frac{1}{2})}^{N-1}$.
Thus,
the values $\expected_{\markov}(\nestedA_N)$  and $\expected(\markov^{\nestedA_N})$ coincide.
\end{exa}

\smallskip
\Paragraph{The plan of the proof}.
We define a $\silent{\flimavg}$-automaton $\nonnestedA$ that simulates runs of
$\nestedA$; the value on $\nonnestedA$ on every word coincides with $\nestedA$.
Then, we transform the Markov chain $\nonnestedA \times \markov$ into
 a Markov chain $\markov_E$ by adjusting its weights only.
We change all weights to the empty weight $\bot$ except for the transitions
corresponding to the invocation of slave automata, where the weight is the expected value of the invoked slave automaton w.r.t.
the distribution given by $\markov$ in the current state.
In the proof we argue that the expected values of limit average of $\nonnestedA \times \markov$ and $\markov_E$ coincide.
We show that by looking at the linear equations corresponding to computing the expected limit average of each of the Markov chains.
Basically, the frequency of each transition is the same in both Markov chains and changing the value of the slave automaton
from its actual value to the expected value does not affect the solution to the set of equations.
Next, we observe that runs of slave automata past the first transition do not matter. Indeed, all runs of slave automata are accepting and all weights past
the first transition are $0$. Thus, we can reduce $\markov_E$ to a Markov chain $\markov_R$ by projecting out information about the runs of  slave automata past the first transition.
Finally, we observe that the Markov chain $\markov_R$ is in fact $\MCfromNested$. Hence, we have shown that
\[
\expected_{\markov}(\nestedA) = \expected_{\markov}(\nonnestedA) = \expected (\markov_E) = \expected(\markov_R) = \expected(\MCfromNested)
\]
\smallskip

\begin{proof}
Every slave automaton of $\nestedA$ takes at most $N$ steps. Therefore,
$\nestedA$ has width bounded by $N$.
Moreover, without loss of generality, we assume that each slave automaton
takes transitions of weight $0$ except for the last transition, which may have a non-zero weight,
and all slave automata are either trivial, i.e., they start in the accepting state and take no transitions,
or they take precisely $N$ transitions.
Basically, slave automata may keep track of the accumulated values
and the number of steps in their states.

\Paragraph{The automaton $\nonnestedA$}. Let
$Q_{mas}$ be the set of states of the master automaton of $\nestedA$ and let
$Q_s$ be the union of the set of states of the slave automata of $\nestedA$.
We define $\nonnestedA$ as a $\silent{\flimavg}$ automaton over the set of states $Q_{mas} \times {(Q_s \cup \{\bot\})}^N$.
The component $Q_{mas}$ is used to keep track of the run of the master automaton while
the component  ${(Q_s \cup \{\bot\})}^N$ is used to keep track of up to $N$ slave automata running concurrently.
The symbol $\bot$ corresponds to an empty slot that can be used to simulate another slave automaton.
Since $\nestedA$ has width bounded by $N$, the automaton $\nonnestedA$ can simulate the Boolean part of the run of $\nestedA$.
The weight of a transition of $\nonnestedA$ is either $\bot$ if no automaton terminates or it is
the value of a terminating slave automaton (non-trivial slave automata take precisely $N$ steps, so at most one can terminate at each position).
Transitions at which no slave automaton terminates are silent transitions.
The automata $\nestedA$ and $\nonnestedA$ encounter the same weights but differ in their aggregation.
The value of a slave automaton is associated to the position at which it is invoked, while in $\nonnestedA$ it is associated with
the position at which the slave automaton terminates.
However, these positions differ by $N$, therefore
the limit averages of both sequences coincide.
Hence, for every word $w$, the values $\valueL{\nestedA}(w)$ and $\valueL{\nonnestedA}(w)$ coincide.
It follows that $\expected_{\markov}(\nestedA) = \expected_{\markov}(\nonnestedA)$.

\Paragraph{The Markov chain $\markov_E$}. We define $\markov_E$ as $\nonnestedA \times \markov$ with altered weights defined as follows.
All transitions which correspond to the invocation of a slave automaton $\slaveA_i$ with the state of the Markov chain $\markov$ being $s$
have weight equal to the expected value of $\slaveA_i$ w.r.t.\ the distribution given by $\markov$ starting in the state $s$.
Other transitions are silent.

\Paragraph{Expected values of $\masterA \times \markov$ and $\markov_E$ coincide}.
Assume that $\masterA \times \markov$ and $\markov_E$ are strongly connected. If they are not,
we can apply the following reasoning for all bottom strongly connected components of both Markov chains as they have the same underlying structure.

Recall that the expected limit average of a Markov chain with silent moves is given by
$c(e_1) \cdot x[e_1] + \cdots + c(e_k) \cdot x[e_k]$ where
variables $x[e]$, over all transitions $e$, form a solution to the system of equations and inequalities~\ref{E1},~\ref{E2} and~\ref{E3}, and
$e_1, \ldots, e_k$ are all non-silent transitions.
Now, observe that the equations~\ref{E1} and inequalities~\ref{E3} are the same for both Markov chains
 $\nonnestedA \times \markov$ and  $\markov_E$ as they have the same structure with the same probabilities.
The equation~\ref{E2} is, in general, different for $\nonnestedA \times \markov$ and for $\markov_E$.
However, non-silent transitions of $\nonnestedA \times \markov$, denoted by $e_1, \ldots, e_k$,
are all states at which at least one slave automaton terminates, while
non-silent transitions of $\markov_E$, denoted by $e_1', \ldots, e_l'$ are all states at which some (non-trivial) slave automaton is invoked.
Observe that every terminating slave automaton has been invoked, and, in $\nonnestedA$,
every invoked slave automaton terminates. Therefore, the sum of frequencies of invocations
and terminations of slave automata are equal, i.e., equations~\ref{E1} imply
\[
x[e_1] + \cdots + x[e_k] = x[e_1']+ \cdots +x[e_l'].
\]
 It follows that the unique solution to equations and inequalities~\ref{E1},~\ref{E2} and~\ref{E3}
corresponding to $\nonnestedA \times \markov$ and to $\markov_E$ are the same.
It remains to show that
\[ c(e_1) \cdot x[e_1] + \cdots + c(e_k) \cdot x[e_k] =
 c'(e_1') \cdot x[e_1'] + \cdots + c'(e_l') \cdot x[e_l'],
\]
 where $c$ (resp. $c'$) are weights in  $\nonnestedA \times \markov$ (resp., $\markov_E$).

Since $c'(e')$ is the expected value of the slave automaton started at $e'$,
the expected value $c'(e')$ is given by $c'(e') = \sum_{e'' \in T} p(e',e'') \cdot c(e'')$, where
$T$ is the set of transitions that correspond the the final transitions of the slave automaton started at the transition $e'$, and
$p(e',e'')$ is the probability of reaching the transition $e''$ from $e'$ omitting the set $T$.
Indeed, each (non-trivial) slave automaton takes precisely $N$ transitions, hence
at each position at most one non-trivial slave automaton terminates and
$c(e'')$ is the value of the slave automaton terminating at $e''$.
Therefore, $c'(e') = \sum_{e'' \in T} p(e',e'') \cdot c(e'')$.

Now, we take  $c'(e_1') \cdot x[e_1'] + \cdots + c'(e_l') \cdot x[e_l']$ and substitute each $c(e_i')$ by
the corresponding $c'(e_i') = \sum_{e'' \in T_i} p(e',e'') \cdot  c(e'')$.
Then, we now group in all the terms by $e''$ and we get
\[
 c'(e_1') \cdot x[e_1'] + \cdots + c'(e_l') \cdot x[e_l'] = \sum_{i=1}^k c(e_i) \cdot \big(x[e_1'] \cdot p(e_1',e_i) + \cdots + x[e_l'] \cdot  p(e_l',e_i)\big)
\]
Observe that the frequency of taking the transition $e_i$ at which some slave automaton $\slaveA$ terminates is equal to the
sum of frequencies on transitions at which this slave automaton $\slaveA$ has been invoked, in which each frequency is multiplied by the probability of reaching
 the terminating transition $e_1$ from a given invoking transition.
Therefore, we have
\[
x[e_1'] \cdot  p(e_1',e_i) + \cdots + x[e_l']\cdot   p(e_l',e_i) = x[e_i].
\]
It follows that
\[ c(e_1) \cdot x[e_1] + \cdots + c(e_k) \cdot x[e_k] =
 c'(e_1') \cdot x[e_1'] + \cdots + c'(e_l') \cdot x[e_l']
\] and
$ \expected_{\markov}(\nonnestedA) = \expected (\markov_E)$.

\Paragraph{The Markov chain $\markov_R$}.
We construct $\markov_R$ from $\markov_E$ by projecting out the component ${(Q_s \cup \{\bot\})}^N$.
We claim that this step preserves the expected value.
First, observe that the distribution is given by an unaffected component $\markov$ and the weights depend only on
the state of the Markov chain $\markov$ and the state of the master automaton $\masterA$. Thus, projecting out
the component ${(Q_s \cup \{\bot\})}^N$ does not affect the expected value, i.e.,
$ \expected_{\markov}(\markov_E) = \expected (\markov_R)$. Now, observe that the set of states of $\markov_R$ is
$Q_{mas} \times Q_{\markov}$.
Observe that the probability and the weights of the transitions of $\markov_R$ match the conditions of the definition
of $\MCfromNested$. Therefore, $\markov_R = \MCfromNested$.
\end{proof}

\subsubsection{Reduction to the bounded-duration case}
Let $\nestedA$ be a $(\flimavg;\fsum)$-automaton.
For every $N$, we define $\nestedA^N$ as $\nestedA$ with the bound $N$ imposed on slaves, i.e.,
each slave automaton terminates either by reaching an accepting state or
when it takes the $N$-th step. Let $\MCfromNested_N$ be the Markov chain that corresponds to $\nestedA^N$.
Observe that as $N$ tends to infinity, weights in $\MCfromNested_N$ converge
to the weights in $\MCfromNested$. It remains to be shown that, as $N$ tends to
infinity, the expected values of $\nestedA^N$ converge to the expected value of $\nestedA$.
We show in the following Lemma~\ref{l:convergence} that random variables generated by
$\nestedA^N$ converge in probability to the random variable generated by $\nestedA$, i.e.,
 for every $\epsilon > 0$ we have
\[
\lim_{N \rightarrow \infty} \probability_{\markov}(\{ w \mid |\valueL{\nestedA}(w) - \valueL{\nestedA^N}(w)| \geq \epsilon \}) = 0
\]
Convergence in probability implies convergence of the expected values.
It follows that the expected values of
$\nestedA$ and $\MCfromNested$ coincide.

\begin{lem}%
\label{l:convergence}
The random variables defined by ${\{ \lang_{\nestedA^N} \}}_{N\geq 0}$ converge in probability to
the random variable defined by $\nestedA$.
\end{lem}

\begin{exa}
Recall Example~\ref{ex:N-ART}.
Lemma~\ref{l:convergence} implies that with $N$ tending to infinity,
the limit of the expected $N$-bounded ARTs converges to the expected ART\@.
For $N>0$, the expected $N$-bounded ART is
$\expected_{\markov_U}(\aut_N) =  2 -  {(\frac{1}{2})}^{N-1}$.  Therefore, the expected ART is $\lim_{N \to \infty} \expected_{\markov_U}(\aut_N) = 2$.
\end{exa}

\begin{proof}
\newcommand{\excessA}[1]{\ensuremath{\nestedA^{\geq #1}}}
\newcommand{\partialExcessA}[1]{\nestedA[#1]}
\newcommand{\flimavgsup}{\textsc{LimAvgSup}}
We define an $(\flimavgsup; \fsum)$-automaton $\excessA{N}$ as the automaton obtained from $\nestedA$ in the following way.
First, each slave automaton take transitions of weight $0$ for the first (up to) $N$ steps, past which it takes transitions of weight $1$
until it terminates.
Second, the value function of the master automaton is $\flimavgsup$ defined on  $a_1, a_2, \ldots$ as
$\flimavgsup(a_1 \ldots ) = \limsup_n \frac{1}{n} \sum_{i=1}^{n} a_i$.
Intuitively, the automaton $\excessA{N}$ computes the limit average (supremum) of the steps slave automata take above the threshold $N$.
Let $C$ be the maximal absolute weight in slave automata of $\nestedA$. Then, for every word $w$ we have
\[
\valueL{\nestedA^N}(w) - C \cdot \valueL{\excessA{N}}(w) \leq \valueL{\nestedA}(w) \leq \valueL{\nestedA^N}(w) + C \cdot \valueL{\excessA{N}}(w).
\]
It follows that
\[
\probability_{\markov}(\{ w \mid |\valueL{\nestedA}(w) - \valueL{\nestedA^N}(w)| \geq \epsilon \}) =
\probability_{\markov}(\{ w \mid |\valueL{\excessA{N}}(w)| \geq \frac{\epsilon}{C} \})
\]
We show that with $N$ increasing to infinity,
probabilities $\probability_{\markov}(\{ w \mid |\valueL{\excessA{N}}(w)| \geq \frac{\epsilon}{C}\} )$
converge to $0$.
From that we conclude that
random variables $\lang_{\nestedA^N}$ converge in probability to $\lang_{\nestedA}$ as $N$ tends to infinity.

Observe that for every word $w$ and every $N$ we have  $0 \leq \valueL{\excessA{N}}(w)$ and $\valueL{\excessA{N}}(w) \geq \valueL{\excessA{N+1}}(w)$.
Therefore, we only need to show that
 for every $\epsilon > 0$ there for $N$ large enough $\expected_{\markov}(\excessA{N}) \leq \epsilon$. Then, by Markov inequality,
$\probability_{\markov}(\{ w \mid |\valueL{\excessA{N}}(w)| \geq \sqrt{\epsilon}) < \sqrt{\epsilon}$.

To estimate the value of $\expected_{\markov}(\excessA{N})$ we consider $\silent{\flimavgsup}$-automata $\partialExcessA{K,i}$ defined as follows.
The automaton $\partialExcessA{K,i}$ simulates the master automaton $\nestedA$ and slaves that are invoked at positions $\{ K\cdot l + i \mid l \in \N \}$.
For every $l>0$, the transition at the position $K \cdot (l+1) + i $ has the weight $1$ if the slave invoked at the position $K \cdot l + i$ works for at least $K$ steps.
Otherwise, this transition has weight $0$. On the remaining positions, transitions have weight $0$.
Observe that due to distributivity of the limit supremum, the limit average supremum of the number of slave automata that take at least $K$ steps
at a given word $w$ is bounded by $\sum_{i=0}^{K-1} \valueL{\partialExcessA{K,i}}(w)$. It follows that for every word $w$ we have
$\valueL{\excessA{N}}(w) \leq \sum_{K \geq N} \sum_{i=0}^{K-1} \valueL{\partialExcessA{K,i}}(w)$. Therefore,
\[
\text{(*)}\quad \expected_{\markov}(\excessA{N}) \leq \sum_{K \geq N} \sum_{i=0}^{K-1} \expected_{\markov}(\partialExcessA{K,i}).
\]

Now, we estimate $\expected_{\markov}(\partialExcessA{K,i})$.
Let $n$ be the maximal size of a slave automaton in $\nestedA$ and let $k$ be the number of slave automata.
We assume, without loss of generality, that every state of slave automata is reached along some run on words generated by $\markov$.
Now, observe that from every state of slave automata some accepting state is reachable.
Otherwise, there would be a set of strictly positive probability at which $\nestedA$ does not accept.
Moreover, as it is reachable, it is reachable within $n$ steps.
Therefore,
the probability $p$ such that any slave automaton in any state
terminates after next $n$ steps is greater than $0$.
It follows that $\expected_{\markov}(\partialExcessA{K,i}) \leq \frac{1}{K} p^{\lfloor \frac{K}{n} \rfloor}$.
We apply this estimate to (*) and obtain $\expected_{\markov}(\excessA{N}) \leq \sum_{K \geq N}  p^{\lfloor \frac{K}{n} \rfloor} \leq n \cdot \frac{p^{\lfloor \frac{N}{n} \rfloor }}{1-p}$.
Therefore, $\expected_{\markov}(\excessA{N})$ converges to $0$ as $N$ increases to infinity.
Finally, this implies that
$\nestedA^N$ converges in probability to $\nestedA$ as $N$ tends to infinity.
\end{proof}

\subsection{The distribution question for the LimAvg value function}

\begin{restatable}{lem}{LimAvgDistribution}%
\label{l:limavg-dist-poly}
Let $g \in \FinVal$.
Given a Markov chain $\markov$, a deterministic almost-surely accepting  $(\flimavg; g)$-automaton $\nestedA$ and a value $\const$,
the value $\distrib_{\markov,\nestedA}(\const)$ can be computed in polynomial time.
\end{restatable}

\proofideas{}
We show that the distribution is discrete. More precisely,
let $\aut$ be the product of the Markov chain $\markov$ and the master automaton of $\nestedA$.
We show that almost all words, whose run end up in the same bottom SCC of $\aut$,
have the same value, which is equal to the expected value over words that end up in that SCC\@.
Thus, to answer the distribution question, we have to
compute for every bottom SCC $C$ of $\aut$, the expected value over words that end up
in $C$ and the probability of reaching $C$.
Both values can be computed in polynomial time (see Lemma~\ref{l:limavg-poly}).

\begin{proof}
Let $\nestedA$ be a deterministic $(\flimavg;\fsum)$-automaton with the master automaton $\masterA$ and let $\markov$ be a Markov chain.
Moreover, let  $\MCfromNested$ be the Markov chain obtained from $\markov$ and $\nestedA$.
We show that the distribution $\distrib_{\markov,\nestedA}$ and the distribution defined by ${\MCfromNested}$ coincide.

\Paragraph{The single-SCC case}. Assume that $\masterA \times \markov$ is an SCC\@.
Observe that the event ``the value of $\nestedA$ equals  $\const$''
is a tail event w.r.t.\ the Markov chain $\markov$, i.e., it does not depend on finite prefixes.
Therefore, its probability is either $0$ or $1$~\cite{feller}. It follows that
the value of almost all words is equal to the expected value of $\nestedA$.
Now, $\MCfromNested$ is structurally the same as $\masterA \times \markov$, hence it is also an SCC\@.
Therefore, also in  $\MCfromNested$ almost all words have the same value, which is equal to $\expected(\MCfromNested)$.
As $\expected_{\markov}(\nestedA) = \expected(\MCfromNested)$ (Lemma~\ref{l:limavgReducesToMC}) we have
$\distrib_{\markov,\nestedA}$ and the distribution defined by ${\MCfromNested}$ coincide.
\smallskip

\Paragraph{The general case}. Consider the case where $\masterA \times \markov$ consists multiple bottom SCCs $S_1,  \ldots, S_k$.
Using conditional probability, we can repeat the single-SCC-case argument to show that in all bottom SCCs $S_1,  \ldots, S_k$
the values of $\nestedA$ are the same and equal to the expected values in these SCCs.
Similarly, in each bottom SCC of ${\MCfromNested}$, all words have the same value, which is equal to the expected value in that SCC\@.
Since $\masterA \times \markov$ is structurally the same as $\MCfromNested$, each SCC $S_1,  \ldots, S_k$  corresponds to an SCC in
${\MCfromNested}$. Lemma~\ref{l:limavgReducesToMC} states that $\expected_{\markov}(\nestedA) = \expected(\MCfromNested)$.
By applying Lemma~\ref{l:limavgReducesToMC}  to $\markov$ and $\nestedA$ with different initial states
of $\markov$ and $\masterA$ (in each $S_1, \ldots, S_k$), we infer that in every SCC $S_1,  \ldots, S_k$
the expected values of $\nestedA$ and $\MCfromNested$ coincide. Therefore, the distribution
 $\distrib_{\markov,\nestedA}$ and the distribution defined by ${\MCfromNested}$ coincide.
\end{proof}

\begin{thm}%
\label{th:compLimAvg}
Let $g \in \FinVal$.
All probabilistic questions for deterministic almost-surely accepting  $(\flimavg; g)$-automata
can be solved in polynomial time.
\end{thm}

\begin{rem}[Contrast with classical questions]%
\label{remark:LimAvg-classical-vs-probabilistic}
Our results summarized in Theorem~\ref{th:compLimAvg} contrast the
results on classical questions shown in Table~\ref{tab1}.
While classical questions are $\PSPACE$-complete, in $\EXPSPACE$ or open,
we establish polynomial-time algorithms for all probabilistic questions.
\end{rem}

\section{Results on non-limit value functions}%
\label{s:nolimit}%
\label{s:inf}

In this section we study NWA with $\finf, \fsup$ value functions for the master automaton.
In contrast to $\fliminf$ and $\flimsup$ value functions, for which all probabilistic questions can
be answered in polynomial time (Theorem~\ref{th:compLimInf}), we show \#P-hardness, $\PSPACE$-hardness and uncomputability results
for $\finf, \fsup$ value functions as well as exponential-time upper bounds in some cases.

In contrast to $\fliminf$ and $\flimsup$ value functions the almost-sure acceptance condition
does not change the complexity of the probabilistic questions.
We show the lower bounds under the almost-sure acceptance condition. However, for the upper bounds
we do not assume almost-sure acceptance of NWA\@. We first present several hardness results
for $\finf$ and $\fsup$ value functions.

\subsection{Lower bounds for all value functions for slave automata}

We present the key ideas of the hardness results.

\smallskip\noindent{\em $\PSPACE$-hardness}.
NWA can invoke multiple slave automata working independently over the same finite subword, which we use to express the problem:
is the intersection of given regular languages empty, which is $\PSPACE$-complete.
We transform given DFA $\aut_1, \ldots, \aut_k$ into slave automata that return $1$ if the original automaton accepts and $0$ otherwise.
The resulting NWA $\nestedA$ is deterministic, accepts almost all words and $\distrib_{\calU,\nestedA}(0)=1$ if the intersection is empty.
Note however, that words of the minimal length in the intersection can have exponential length and their probability can be doubly-exponentially small.
Therefore, even if $\distrib_{\calU,\nestedA}(0) \neq 1$, the difference between $\distrib_{\calU,\nestedA}(0)$ and $1$ can be small and hence
$\PSPACE$-hardness does not apply to the approximation problems (which we establish below).

\smallskip\noindent{\em $\#P$-hardness}.
We show $\#P$-hardness of the approximate probabilistic questions by reduction from
$\#\SAT$, which is $\#P$-complete~\cite{valiant1979complexity,papadimitriou2003computational}.
The $\#\SAT$ problem asks for the number of variable assignments satisfying a given CNF formula $\varphi$.
In the proof, the input word gives an assignment, and each slave automaton checks the satisfaction of one clause and returns $1$ if it is satisfied and $0$ otherwise.
Therefore, all slave automata return $1$ if and only if all clauses are satisfied, and hence
we can compute the number of satisfying assignments of $\varphi$ from $\expected_{\calU}(\nestedA)$ and $\distrib_{\calU, \nestedA}(0)$,
where $\calU$ is the uniform distribution over infinite words.
The lower bounds hold even under the additional almost-sure acceptance condition.

\begin{restatable}[Hardness results]{lem}{InfimumIsHard}%
\label{l:hardness-for-det-inf}
Let $g \in \FinVal$ be a value function, and $\calU$ denote the uniform distribution over the infinite
words.
\begin{enumerate}
\item The following problems are $\PSPACE$-hard:
Given a deterministic almost-surely accepting $(\finf;g)$-automaton (resp., $(\fsup; g)$-automaton) $\nestedA$,
decide whether $\expected_{\calU}(\nestedA)=0$; and decide whether $\distrib_{\calU,\nestedA}(0)=1$?
\item The following problems are \#P-hard:  Given $\epsilon > 0$ and a deterministic almost-surely accepting $(\finf;g)$-automaton (resp., $(\fsup; g)$-automaton)
$\nestedA$, compute $\expected_{\calU}(\nestedA)$ up to precision $\epsilon$; and compute $\distrib_{\calU, \nestedA}(0)$ up to precision $\epsilon$.
\end{enumerate}
\end{restatable}

\begin{proof}
We present the following argument for $g  = \fmin$. The same proof works for $g = \fmax$.
Lemma~\ref{l:sup-to-inf} implies that problems in (1) and (2) for nested weighted automaton with $\fsup$ value function reduce to the corresponding problems for  nested weighted automaton with $\finf$ value function.
Since $\fmin$ can be regarded as a special case of $\fBsum{B}, \fsum^+$ or $\fsum$, the result holds for these functions as well.
Hence, we consider only the case of $(\finf, \fmin)$-automata.

\Paragraph{PSpace-hardness}: We show $\PSPACE$-hardness by reduction from the emptiness problem for the intersection of regular languages.
Let $\lang_1, \ldots, \lang_n \subseteq {\{a,b\}}^*$ be regular languages recognized by
deterministic finite automata $\aut_1, \ldots, \aut_n$.
We define a deterministic $(\finf;\fmin)$-automaton $\nestedA$ that at the first $n$ steps starts slave automata
$\slaveA_1, \ldots, \slaveA_n$ and then it invokes only a dummy slave automaton that returns $1$ after a single step.
For every $i$, the slave automaton $\slaveA_i$ first reads $n-i$ letters which it ignores, then
it simulates $\aut_i$ until the first $\#$ when it terminates. It returns $1$ if the simulated automaton $\aut_i$ accepts
and $0$ otherwise. More precisely, $\slaveA_i$ works on subwords $uv \#$, where $u \in {\{a,b,\#\}}^{n-i}, v \in {\{a,b \}}^*$ and
returns $1$ if $v \in \lang_i$ and $0$ otherwise.
Observe that on a word $w = u v \# w'$ where $u \in {\{a,b,\# \}}^{n}, v \in {\{a,b\}}^*$ and $w' \in {\{a,b,\# \}}^{\omega}$,
the automaton $\nestedA$ returns $1$ if and only if all automata $\aut_1, \ldots, \aut_n$ accept $v$.
Otherwise, $\nestedA$ assigns value $0$ to $w$.
In consequence, the following conditions are equivalent:
(1)~the intersection  $\lang_1 \cap \ldots \cap \lang_n$ is empty,
(2)~the expected value $\expected_{\calU}(\nestedA)$ is $0$, and
(3)~the distribution $\distrib_{\calU, \nestedA}(0) = 1$.
Note that the almost-sure distribution question in $\PSPACE$-hard as well.

Observe that if the intersection $\lang_1 \cap \ldots \cap \lang_n$ is non-empty it might be the case that the word of the minimal
length in the intersection consists of a single word of  exponential length. In such a case, the values
$\expected_{\calU}(\nestedA)$ and $|1-\distrib_{\calU, \nestedA}(0)|$ are non-zero, but doubly-exponentially small.
Therefore, we cannot use this reduction to show hardness of the approximate versions of the probabilistic problems.

\Paragraph{$\#P$-hardness}:
We show \#P-hardness by reduction from the \#SAT problem, which, given a propositional formula $\varphi$ in conjunctive normal form
asks for the number of valuations that satisfy $\varphi$.
Let $n$ be the number of variables of $\varphi$ and let $C_1, \ldots, C_m$ be the clauses of $\varphi$.
For every $i \in [1,m]$, we define a slave automaton $\slaveA_i$ (associated with $C_i$) that ignores the first $m-i$ letters, next
considers the following $n$ letters $0,1$ as the valuation of the successive variables and checks
whether this valuation satisfies the clause $C_i$. If it does, the slave automaton returns $1$, otherwise it returns $0$.
The master automaton first invokes slave automata $\slaveA_1, \ldots, \slaveA_m$ and then it invokes
a dummy slave automaton that returns $1$ after a single step.
Observe that for $w = uvw'$, where $u \in {\{0,1\}}^m, v \in {\{0,1\}}^n$ and $w' \in {\{0,1\}}^{\omega}$,
the automaton $\nestedA$ returns $1$ on $w$ if and only if the valuation given by $v$ satisfies
all clauses $C_1, \ldots, C_m$, i.e., it satisfies $\varphi$.
Otherwise, $\nestedA$ returns $0$ on $w$.
Therefore, the values $\expected_{\calU}(\nestedA)$ and $1 - \distrib_{\calU, \nestedA}(0)$
 are equal, and multiplied by $2^{n}$ give the number of valuations satisfying $\varphi$.
In follows that all approximate probabilistic questions are $\#P$-hard.
\end{proof}

\subsection{Upper bounds for value functions \texorpdfstring{$g\in  \FinVal \setminus \{ \text{Sum}^+,\text{Sum} \}$}{g in FinVal\textbackslash\{ Sum+, Sum \}}}
We now present upper bounds for value functions
for $(\finf;g)$-automata and $(\fsup;g)$-automata, where $g \in  \FinVal \setminus \{ \fsum^+,\fsum \}$.

\Paragraph{Overview}.
We begin with the discussion on the bounded-sum value function.
We show the translation lemma (Lemma~\ref{l:bsum-to-inf}), which states that deterministic $(\finf; \fBsum{B})$-automata
can be translated to deterministic $\finf$-automata with exponential blow-up. Moreover, this blow-up can be avoided
by considering NWA of bounded width and the bound in the sum $B$ given in unary.
Since the probabilistic questions can be solved for
$\finf$-automata in polynomial time, Lemma~\ref{l:bsum-to-inf} implies that all probabilistic
questions can be solved in exponential time for deterministic $(\finf; \fBsum{B})$-automata.
Next, we use Lemma~\ref{l:bsum-to-inf} to show the upper bounds on the probabilistic questions for all slave value functions $g\in  \FinVal \setminus \{ \fsum^+,\fsum \}$.

\proofideas{}
It has been shown in~\cite{nested} that for
$f \in \InfVal$ and $g \in \FinVal \setminus \{ \fsum^+, \fsum \}$,
$(f;g)$-automata can be transformed to exponential-size $f$-automata.
In the original transformation the threshold $B$ is fixed.
We slightly modify the construction from~\cite{nested}, to show the case of variable threshold for  $(\finf; \fBsum{B})$-automata.

\begin{restatable}{lem}{BoundedSumReducesToInf}%
\label{l:bsum-to-inf}
(1)~Given $B >0$ in the binary notation and a deterministic $(\finf; \fBsum{B})$-automaton $\nestedA$, one can construct
in exponential time an exponential-size deterministic $\finf$-automaton $\nonnestedA$ such that for every word $w$
we have $\valueL{\nestedA}(w) = \valueL{\nonnestedA}(w)$.
(2)~Let $k > 0$.
Given $B >0$ in the unary notation and a deterministic $(\finf; \fBsum{B})$-automaton $\nestedA$ of width bounded by $k$,
one can construct in polynomial time a polynomial-size deterministic $\finf$-automaton $\nonnestedA$ such that for every word $w$
we have $\valueL{\nestedA}(w) = \valueL{\nonnestedA}(w)$.
\end{restatable}
\begin{proof}
\Paragraph{(1)}: Let $Q_m$ be the set of states of the master automaton and
$Q_s$ be the union of the sets of states of slave automata of $\nestedA$.
We define an $\finf$-automaton $\nonnestedA$ over the set of states $Q_m \times {(Q_s \times [-B, B] \cup \{ \bot\})}^{|Q_s|}$.
Intuitively, $\nonnestedA$ simulates runs of $\nestedA$ by simulating (a)~the run of the master automaton using the component $Q_m$
and (b)~selected runs of up to $|Q_s|$ slave automata using the component ${(Q_s \times [-B, B])}^{|Q_s|}$.
Slave automata are simulated along with their values, which are stored in the state, i.e., the state $(q, l)$ encodes
that a given slave automaton is in the state $q$ and its current value is $l$.
Then, the value of a given transition of $\nonnestedA$ is the minimum over the values of simulated slave automata
that terminate at the current step.
Finally, the symbol $\bot$ denotes ``free'' components in the product
${(Q_s \times [-B, B] \cup \{ \bot\})}^{|Q_s|}$, which can be used to simulate newly invoked slave automata.
We need to convince ourselves that we need to simulate at most $|Q_s|$ slave automata.
Therefore, every time a new slave automaton is invoked, we have a free component to simulate it.

Observe that if at some position two slave automata $\slaveA_1, \slaveA_2$ are in the same state $q$ and they have
computed partial sums of weights $l_1 \leq l_2$, then we can discard the simulation of the automaton $\slaveA_2$, which computed the value $l_2$.
Indeed, since slave automata are deterministic and recognize prefix-free languages,  the remaining runs of both slave automata
$\slaveA_1, \slaveA_2$ are the same, i.e., they either both reject or both return values, respectively,  $l_1 + v$ and $l_2 +v$ for some common $v$.
Thus, the run of $\slaveA_2$ does not change the value of the infimum and we can stop simulating it, i.e., we can substitute $(q, l_2)$ by $\bot$.
Therefore, at every position at most $|Q_s|$ components are necessary.
It follows from the construction that the values of $\nestedA$ and $\nonnestedA$ coincide on every word.

\Paragraph{(2)}: If $B$ is given in the unary notation and the width is bounded by $k$, we can basically repeat the construction as above for
the automaton with the set of states  $Q_m \times {(Q_s \times [-B, B] \cup \{ \bot\})}^{k}$, which is polynomial in $\nestedA$.
Thus, the resulting automaton has polynomial size in $\nestedA$ and can be constructed in polynomial time (in $\nestedA$).
\end{proof}

\smallskip\noindent{\em Key ideas}.
We consider almost surely accepting automata and hence by Lemma~\ref{l:sup-to-inf}, the results of Lemma~\ref{l:bsum-to-inf} apply to
$(\fsup; \fBsum{B})$-automata.
The bounded-sum value function $\fBsum{B}$ subsumes all value functions from $g\in  \FinVal \setminus \{ \fsum^+,\fsum \}$, and hence
Lemma~\ref{l:bsum-to-inf} implies that a deterministic $(\finf;g)$-automaton (resp., $(\fsup;g)$-automaton) can be transformed to an equivalent exponential-size
$\finf$-automaton (resp.,  $\fsup$-automaton).
Therefore, using Fact~\ref{t:weighted-inf-expected}, both $\expected_{\markov}(\nestedA)$ and $\distrib_{\markov,\nestedA}(\const)$
can be computed in exponential time.

\begin{restatable}{lem}{InfExpectedSolution}%
\label{l:infSolutions}
Let $g \in \FinVal \setminus \{ \fsum^+, \fsum \}$ be a value function.
(1)~Given a Markov chain $\markov$, a deterministic $(\finf;g)$-automaton (resp., $(\fsup; g)$-automaton) $\nestedA$, and a
threshold $\const$ in binary, both $\expected_{\markov}(\nestedA)$ and $\distrib_{\markov,\nestedA}(\const)$
can be computed in exponential time.
(2)~If $\nestedA$ has bounded width, then the above quantities can be computed in polynomial time.
\end{restatable}

\begin{rem}
Observe the following:
\begin{enumerate}
\item We show in Lemma~\ref{l:infSolutions} a polynomial-time upper bound for NWA with bounded width, which gives a
polynomial-time upper bound for automata with monitor counters.
\item For $g = \fBsum{B}$, the value $B$ can be given in binary in input, and the complexity  in (1) from Lemma~\ref{l:infSolutions}
does not change.
\end{enumerate}
\end{rem}

\noindent
We first prove Lemma~\ref{l:infSolutions} for $g \in \FinVal \setminus \{ \fsum,\fsum^+\}$.
Next, in Lemma~\ref{l:bound-from-below} we show that Lemma~\ref{l:infSolutions} holds for deterministic $(\finf;\fsum^+)$-automata.
The statement of Lemma~\ref{l:bound-from-below} is more general though.

\begin{proof}
Observe that deterministic $\fmin$-automata and $\fmax$-automata can be
transformed in polynomial time to equivalent deterministic $\fBsum{B}$-automata.
Basically, a deterministic $\fBsum{B}$-automaton simulating a $\fmin$-automaton (resp., $\fmax$-automaton)
uses its bounded sum to track the currently minimal (resp., maximal) weight taken by the automaton.
Therefore, we focus on $g = \fBsum{B}$.
Consider a deterministic $(\finf;\fBsum{B})$-automaton $\nestedA$.
By Lemma~\ref{l:bsum-to-inf}, $\nestedA$ can be transformed in exponential time into an equivalent exponential-size deterministic $\finf$-automaton
$\nonnestedA$ and hence
$\expected_{\markov}(\nestedA) = \expected_{\markov}(\nonnestedA)$ and
$\distrib_{\markov,\nestedA}(\const) = \distrib_{\markov,\nonnestedA}(\const)$ (for all $\const$).
The values  $\expected_{\markov}(\nonnestedA), \distrib_{\markov,\nonnestedA}(\const)$ can be computed in polynomial time
in $\nonnestedA$ (Fact~\ref{t:weighted-inf-expected}), which amounts to exponential time in $\nestedA$.
Observe, however, that for $\nestedA$ of bounded width the automaton $\nonnestedA$ has polynomial size (assuming that the bound on the width is constant),
and the values  $\expected_{\markov}(\nonnestedA), \distrib_{\markov,\nonnestedA}(\const)$  can be computed in polynomial time in $\nestedA$.
\end{proof}

Now, we turn to deterministic $(\finf;\fsum)$-automata.

\subsection{The \texorpdfstring{$\text{Sum}^+$}{Sum+} and Sum value functions for slave automata}
We now establish the result when $g \in \{ \fsum,\fsum^+ \}$.
First we establish decidability of the approximation problems,
and then undecidability of the exact questions.
Finally, we show that for deterministic $(\finf; \fsum^+)$-automata all probabilistic questions are decidable.

\proofideas{}
The main difference between $\finf$ and $\fliminf$ value functions is that the latter discards all values encountered before the master automaton
reaches a bottom SCC where the infimum of values returned by slave automata coincides with the limit infimum and hence it can be computed in polynomial time (Lemma~\ref{l:in-scc-all-equal}).
We show that we can bound the values returned by slave automata and the expected values and the distributions do not change much.
More precisely, given $\nestedA$ and $\epsilon$, we show that for some $B$, exponential in $|\nestedA|$ and polynomial in the binary representation of $\epsilon$,
the probability that any slave automaton collects a (partial) sum outside the interval $[-B,B]$ is smaller than $\epsilon$.
Therefore, to approximate $\expected_{\markov}(\nestedA)$ and $\distrib_{\markov, \nestedA}(\const)$ up to precision $\epsilon$, we can regard a given
$(\finf; \fsum)$-automaton (resp., $(\fsup; \fsum)$-automaton) as
$(\finf; \fBsum{B})$-automaton (resp., $(\fsup; \fBsum{B})$-automaton) and use Lemma~\ref{l:infSolutions}.

\begin{restatable}{lem}{InfSumSolution}%
\label{l:infSumSolution}
Let $g \in \{\fsum^+, \fsum\}$.
Given $\epsilon > 0$, a Markov chain $\markov$, a deterministic $(\finf;g)$-automaton (resp., $(\fsup; g)$-automaton) $\nestedA$,
a threshold $\const$, both
$\expected_{\markov}(\nestedA)$ and $\distrib_{\markov, \nestedA}(\const)$ can be computed up to precision $\epsilon$
in exponential time in $\nestedA$, polynomial time in $\markov$ and  the binary representation of $\epsilon$.
\end{restatable}

\begin{proof}
Consider a deterministic $(\finf;\fsum)$-automaton $\nestedA$. Let $\nestedA^{\lim}$ be $\nestedA$ considered as  a $(\fliminf;\fsum)$-automaton.
First, we assume that $\nestedA^{\lim}$ has finite expected value (in particular it accepts almost all words).
We can check whether this assumption holds in polynomial time by computing $\expected_{\markov}(\nestedA^{\lim})$ (Lemma~\ref{th:liminfIsPoly}).
Then, we show the following {\bf claim}: for every $\epsilon >0$, there exists $B>0$, exponential in $|\nestedA|$ and linear in $|\log(\epsilon)|$ such that for $\nestedA_B$ defined as $\nestedA$ considered
as an $(\finf; \fBsum{B})$-automaton we have $|\expected_{\markov}(\nestedA) - \expected_{\markov}(\nestedA_B)| \leq \epsilon$.

\smallskip
\Paragraph{The claim implies the lemma}. Observe that due to Lemma~\ref{l:infSolutions} with the following remark on $B$,
the expected value $\expected_{\markov}(\nestedA_B)$ can be computed in polynomial time in $\nestedA_B$, hence exponential time in $\nestedA$ and polynomial time in $\markov$.
Therefore, we can approximate $\expected_{\markov}(\nestedA)$ up to $\epsilon$ in exponential time in $\nestedA$ and polynomial time in $\markov$.
Due to Markov inequality, for every $\const$ we have $\distrib_{\markov,\nestedA}(\const+\epsilon) - \distrib_{\markov,\nestedA_B}(\const-\epsilon) < \epsilon$.
However, the values of $\nestedA$ are integers, therefore for $\epsilon < 0.5$ we get
$|\distrib_{\markov,\nestedA}(\const) - \distrib_{\markov,\nestedA_B}(\const)| < \epsilon$.
Therefore, again by Lemma~\ref{l:infSolutions},
we can approximate $\distrib_{\markov,\nestedA}(\const)$ in exponential time in $\nestedA$, polynomial time in $\markov$ and $|\log(\epsilon)|$.

\smallskip
\Paragraph{The proof of the claim}.
First, we observe that every run ends up in some some SCC of $\masterA \times \markov$, and, hence,
Lemma~\ref{l:in-scc-all-equal} implies that values of all words are bounded from above by $|\nestedA| \cdot |\markov|$.
Next, the values of all slave automata invoked in bottom SCCs of $\masterA \times \markov$
are bounded from below. Otherwise, the expected value of $\nestedA$ as a $(\fliminf;\fsum)$-automaton is $-\infty$.
Assume that the values of all slave automata invoked in bottom SCCs of $\masterA \times \markov$ are bounded from below, which implies that
they are bounded by $-|\nestedA|$.
Then, we need to estimate the influence on the expected value of the slave automata invoked before the master automaton reaches
a bottom SCC of $\masterA \times \markov$.

Let $Y_B$ be a random variable on finite words such that $Y_B(u)$ is the maximum of $0$ and the number of steps of any slave automaton on $\nestedA$ takes on $u$
minus $B$. Let $E_2$ be the expected number of steps of the master automaton before it reaches a bottom SCC\@.
It follows that for $B > |\nestedA|$ and $C$ being the maximal absolute weight in slave automata of $\nestedA$, we have
$|\expected_{\markov}(\nestedA) - \expected_{\markov}(\nestedA_B)| < C \expected(Y_B) \cdot E_2$.

Let $p$ be the minimal positive probability that occurs in $\markov$ and let $n = |\nestedA|$.
We show that for $B > \frac{n}{p^n} |\log \frac{n^2}{p^n} \epsilon|$,
 we have $\expected(Y_B) \cdot E_2 < \epsilon$. We first estimate $E_2$.
Observe that starting from every state, there exists at least one word of length at most $|\masterA|$ upon which
the master automaton reaches a bottom SCC of $\masterA \times \markov$.
Therefore, the master automaton reaches a bottom SCC  in $|\masterA|$ steps with probability at least $p^{|\masterA|}$, and, hence,
the number of steps before $\masterA$ reaches a bottom SCC is estimated from above by
$|\masterA|$ multiplied by the geometric distribution with the parameter $p^{|\masterA|}$.
Hence, $E_2$ is bounded by $\frac{n}{p^{n}}$.

Now, we estimate $\expected(Y_B)$.
Observe that for every reachable state $q$ of any slave automaton $\slaveA$, there exists a word
of the length at most $|\slaveA|$ such that $\slaveA$, starting in $q$ terminates upon reading that word.
Therefore, the probability $q_l(\slaveA)$ that $\slaveA$ works at least $l$ steps is bounded by ${(1-p^{|\slaveA|})}^{\frac{l}{|\slaveA|}}$.
Now, $\expected(Y_B) $ is bounded by the maximum over slave automata $\slaveA$ of $\sum_{l \geq B} q_l(\slaveA)$.
We have $\sum_{l \geq B} q_l(\slaveA) \leq  \frac{n}{p^n} \cdot {(1-p^n)}^{\frac{B}{n}}$.
Hence, $\expected(Y_B) \leq \frac{n}{p^n} \cdot {(1-p^n)}^{\frac{B}{n}}$ and
$\expected(Y_B) \cdot E_2 \leq \frac{n^2}{p^n} \cdot {(1-p^n)}^{\frac{B}{n}}$.
Observe that for $B > \frac{n}{p^n} s$, where $s  = |\log \frac{n^2}{p^n} \epsilon|$,
we have $\frac{n^2}{p^n} \cdot {(1-p^n)}^{\frac{B}{n}} \leq  \frac{n^2}{p^n} \cdot {(\frac{1}{2})}^s$ and
 $\expected(Y_B) \cdot E_2 \leq \epsilon$.
Observe that $\frac{n}{p^n} \cdot |\log \frac{n^2}{p^n} \epsilon|$ is exponential in $|\nestedA|$ and
linear in $|\log(\epsilon)|$.
\smallskip

\Paragraph{Lifting the assumption}.
Now, we discuss how to remove the assumption that $\expected_{\markov}(\nestedA^{\lim})$ is finite.
First, we check in polynomial time whether $\nestedA$ accepts almost all words (Proposition~\ref{prop:almostAll}).
For the expected question, observe that $\expected_{\markov}(\nestedA) \leq \expected_{\markov}(\nestedA^{\lim})$, hence
if the latter is $-\infty$, we can return the answer $\expected_{\markov}(\nestedA) = -\infty$ if $\nestedA$ accepts almost all words, and otherwise
the expected value is undefined.
For the distribution question, consider threshold $\const$.
Observe that for every $w$, we have $\valueL{\nestedA}(w) \leq \valueL{\nestedA_B}(w)$.
Moreover, $\valueL{\nestedA}(w) < \const$ while $\valueL{\nestedA_B}(w) \geq \const $ holds only if there is a slave automaton
run before $\masterA$ reaches a bottom SCC which runs more than $B$ steps. Therefore,
the probability $\probability_\markov(\{w \mid \valueL{\nestedA}(w) \leq \const \wedge  \valueL{\nestedA_B}(w) \geq \const \})$ is bounded from above by
$\expected({Y_B}) \cdot E_2$. Thus, by the previous estimate on $\expected({Y_B}) \cdot E_2$,  for
 $B > \max(\const+1, \frac{n}{p^n} \log \frac{n^2}{p^n} \epsilon|)$  we have
$\probability_\markov(\{w \mid \valueL{\nestedA}(w) \leq \const \wedge  \valueL{\nestedA_B}(w) \geq \const \}) < \epsilon$
and $|\distrib_{\markov, \nestedA}(\const) - \distrib_{\markov,\nestedA_B}(\const)| < \epsilon$.
Again, $\distrib_{\markov,\nestedA_B}(\const)$ can be computed in exponential time in $|\nestedA|$ and polynomial in $\markov$ and $|\log(\epsilon)|$.
\end{proof}

We now show that the exact values in probabilistic questions are uncomputable for deterministic almost-surely accepting $(\finf; \fsum)$-automata (resp., $(\fsup; \fsum)$-automata).

\proofideas{}
To show uncomputability of the distribution question, we show undecidability of the almost-sure distribution problem.
Given a two-counter machine $\M$, we modify the construction from the proof of Theorem~\ref{th:undecidable-limsup} and construct
 a deterministic $(\finf; \fsum)$-automaton $\nestedA_{\M}$ such that the following conditions are equivalent:
(a)~$\M$ has an accepting computation,
(b)~there exists a finite word $u$ such that for all $w = uw'$ we have $\valueL{\nestedA_{\M}}(w) \geq 0$, and
(c)~$\probability_{\markov}( \{ w \mid  \valueL{\nestedA_{\M}}(w) \geq 0 \}) > 0$.
Condition (c) is equivalent to $\distrib_{\calU, \nestedA}(-1) < 1$.
We show  uncomputability of the expected question via reduction from the almost-sure distribution problem
to deciding equality of expected values of given to automata.
Given an $(\finf; \fsum)$-automaton $\nestedA_{\M}$, we construct $\nestedA_{\M}'$ such that
for all words $w$ we have $\valueL{\nestedA_{\M}'}(w) = \min(-1, \valueL{\nestedA_{\M}}(w))$.
Observe that the following conditions are equivalent:
(i)~the expected values of $\nestedA_{\M}'$ and $\nestedA_{\M}$ are equal,
(ii)~$\nestedA_{\M}'$ and $\nestedA_{\M}$ are equal on almost every word, and
(iii)~$\distrib_{\calU, \nestedA_{\M}}(-1) = 1$.

\begin{restatable}{lem}{InfSumUndecidable}%
\label{l:inf-prob-undec}
Let $\calU$ be the uniform distribution over the infinite words.
The following problems are undecidable:
(1)~Given a deterministic almost-surely accepting $(\finf; \fsum)$-automaton (resp., $(\fsup; \fsum)$-automaton) $\nestedA$ of width $8$, decide whether
$\distrib_{\calU, \nestedA}(-1) = 1$.
(2)~Given two deterministic almost-surely accepting $(\finf; \fsum)$-automata (resp., $(\fsup; \fsum)$-automata) $\nestedA_1, \nestedA_2$ of width bounded by $8$, decide whether
$\expected_{\calU}(\nestedA_1) = \expected_{\calU}(\nestedA_2)$.
\end{restatable}
\begin{proof}
\Paragraph{(1)}:
In the following, we discuss how to adapt the proof
of Theorem~\ref{th:undecidable-limsup} to prove this lemma.

Given a deterministic two-counter machine $\M$, we construct a deterministic $(\finf; \fsum)$-automaton $\nestedA_M$
such that for a word $w$ of the form $\$u\$ w'$ it returns $0$ if $u$ is a valid accepting computation of $\M$
and a negative value otherwise.
We use $\Sigma =  Q \cup \{ 1,2,\#,\$\}$ for convenience; one can encode letters from $\Sigma$ using two-letter alphabet $\{0,1\}$.
On words that are not of the form $\$u\$ w'$, the automaton $\nestedA_M$ returns values less or equal to $-1$.
Basically, the automaton $\nestedA_M$ simulates on $u$ the execution of $\nestedA$ (as defined in the proof of Theorem~\ref{th:undecidable-limsup}) with the opposite values of slave automata, i.e.,
all weights of slave automata are multiplied by $-1$.
Recall, that the supremum of the values returned by slave automata on a subword $\$u\$$
is $0$ if and only if $u$ encodes a valid and accepting computation of $\M$.
Otherwise, the supremum is at least $1$.
Thus, in our case, the infimum over the values of slave automata is $0$ if and only if $u$
encodes a valid and accepting computation of $\M$. Otherwise,
the value of $\finf$ is at most $-1$.
Therefore,  $\distrib_{\calU,\nestedA}(-1) = 1$ if and only if $\M$ does not have an accepting computation.

\Paragraph{(2)}: We show that knowing how to compute the expected value of deterministic
$(\finf; \fsum)$-automata, we can decide equality in the distribution question.
Let $\nestedA$ be an automaton and we ask whether $\distrib_{\calU,\nestedA}(-1) = 1$.
We construct another automaton $\nestedA'$ that simulates $\nestedA$, but at the first transition it invokes a slave automaton that returns the value $-1$.
The values of automata $\nestedA$ and $\nestedA'$ differ precisely on words which have values (assigned by $\nestedA$) greater than $-1$.
Thus, their expected values
$\expected_{\calU}(\nestedA)$ and $\expected_{\calU}(\nestedA')$ differ if and only if
$\distrib_{\calU,\nestedA}(-1)$ is different than $1$. Due to undecidability of the latter problem, there exists no Turing machine that computes the expected value of $(\finf;\fsum)$-automata
over the uniform distribution.
\end{proof}

Finally, we have the following result for the absolute sum value function, which guarantees that the return values are at least $0$. We present a slightly more general result.
Recall that we assume that weights in slave automata are given in unary.

\begin{lem}%
\label{l:bound-from-below}
Given a Markov chain $\markov$, a value $\const \in \Q$ and a deterministic $(\finf;\fsum)$-automaton
such that the value of every slave automaton is bounded from below,
the values $\expected_{\markov}(\nestedA)$ and $\distrib_{\markov,\nestedA}(\const)$ can be computed in exponential time in $|\nestedA|$ and polynomial time in $|\markov|$.
\end{lem}

\begin{proof}
Consider a deterministic $(\finf;\fsum)$-automaton $\nestedA$ such that the value of every slave automaton is bounded from below.
Let $B = |\nestedA|\cdot |\markov|$ and let $\nestedA'$ be $\nestedA$ considered as a deterministic $(\finf; \fBsum{B})$-automaton.
We show that on almost all words $w$ we have $\valueL{\nestedA}(w) = \valueL{\nestedA'}(w)$.
Then,
$\expected_{\markov}(\nestedA) = \expected_{\markov}(\nestedA')$
and $\distrib_{\markov,\nestedA}(\const) = \distrib_{\markov,\nestedA'}(\const)$ and the values
$\expected_{\markov}(\nestedA')$ and $\distrib_{\markov,\nestedA'}(\const)$ can be computed in exponential time by
Lemma~\ref{l:infSolutions} taking into account the remark about $B$ being input.

Since the value of every slave automaton $\slaveA_i$ is bounded from below,
the (reachable part of) the weighted Markov chain $\slaveA_i \times \markov$ considered as a weighted graph does not have negative cycles.
Therefore, the minimal value $\slaveA_i$ can achieve is greater than $-|\slaveA_i|\cdot|\markov| > - |\nestedA|\cdot|\markov|$.
Moreover, every accepting run of $\nestedA$ ends up in some SCC of $\masterA \times \markov$, where
almost all words have the same value~(Lemma~\ref{l:in-scc-all-equal}), which
is either $\infty$ (if almost all words are rejected) or bounded from above by $|\nestedA|\cdot|\markov|$. This value can be computed in polynomial time.
Therefore, the value of almost all words belong to the interval $[- |\nestedA|\cdot|\markov|,  |\nestedA|\cdot|\markov|]$ or it is $\infty$ if the run on $\nestedA$ on this word is rejecting.
Finally, the sets of words with accepting runs in $\nestedA$ and $\nestedA'$ coincide.
\end{proof}

The above lemma implies that the probabilistic questions for deterministic $(\finf;\fsum^+)$-automata can be answered in exponential time in $|\nestedA|$ and polynomial time in $|\markov|$.
Note that $(\finf;\fsum^+)$-automata and $(\fsup;\fsum^+)$-automata are not dual.
Indeed, in Lemma~\ref{l:sup-to-inf} we multiply weights by $-1$, which turns $\fsum^+$-automata into $\fsum$-automata with negative weights.
Thus, we consider separately the distribution question for $(\fsup;\fsum^+)$-automata.
We show that the distribution question for deterministic $(\fsup;\fsum^+)$-automata is
decidable in $\EXPTIME$.

\begin{lem}%
\label{l:distributionSumPlusDecidable}
The distribution question for deterministic $(\fsup;\fsum^+)$-automata can be computed in exponential time in $|\nestedA|$ and polynomial time in $|\markov|$.
\end{lem}
\begin{proof}
Let $\nestedA$ be a deterministic $(\fsup;\fsum^+)$-automaton, $\markov$ be a Markov chain and let $\const$ be a threshold in the distribution question.
Consider $\nestedA'$ defined as $\nestedA$ considered as a $(\fsup;\fBsum{B})$-automaton with $B = \const+1$.
Observe that for every word $w$ we have $\valueL{\nestedA}(w) \leq \const$ if and only if
$\valueL{\nestedA'}(w) \leq \const$. Therefore,
$\distrib_{\markov,\nestedA}(\const) = \distrib_{\markov,\nestedA'}(\const)$.
The latter value can be computed in exponential time in $\nestedA$ and polynomial time in $|\markov|$ (Lemma~\ref{l:infSolutions}).
\end{proof}

\begin{thm}%
\label{th:compInf}
Let $g \in \FinVal$.
The complexity results for the probabilistic questions for $(\finf; g)$-automata
and $(\fsup,g)$-automata are summarized in Table~\ref{tab:compInf}, with the exception
of the expected question of $(\fsup; \fsum^+)$-automata.
\end{thm}

\begin{table}[ht]
\centering
\begin{tabular}{|c|c|c|c|} 
\hline 
& $\fmin$, $\fmax$, & \multirow{2}{*}{$\fsum$} \\
& $\fBsum{B},\fsum^+$ &  \\
\hline 
Expected value & \multirow{3}{*}{$\EXPTIME$~(Lemma~\ref{l:infSolutions},~\ref{l:infSumSolution},~\ref{l:distributionSumPlusDecidable})} & \multirow{3}{*}{\uncomp}\\
\cline{1-1}
{Distribution} & \multirow{3}{*}{$\PSPACE$-hard~(Lemma~\ref{l:hardness-for-det-inf})} & \multirow{3}{*}{(Lemma~\ref{l:inf-prob-undec})}\\
\cline{1-1}
Almost sure & &  \\
distribution &  & \\
\hline 
Approximate: & \multicolumn{2}{|c|}{\multirow{2}{*}{$\EXPTIME$~(Lemma~\ref{l:infSolutions},~\ref{l:infSumSolution})}} \\
(a)~expected value & \multicolumn{2}{|c|}{\multirow{2}{*}{\#P-hard~(Lemma~\ref{l:hardness-for-det-inf})}} \\
(b)~distribution & \multicolumn{2}{|c|}{} \\
\hline 
\end{tabular}
\caption{The complexity results for various problems for deterministic NWA with $\finf$ and $\fsup$ value functions,
with exception of the expected question of $(\fsup,\fsum^+)$-automata which is open.
Columns represent slave-automata value functions, rows represent probabilistic questions.
}%
\label{tab:compInf}
\end{table}

\smallskip\noindent{\em Open question}.
The decidability of the expected question of $(\fsup; \fsum^+)$-automata is open.
This open problem is related to the language inclusion problem of deterministic
$(\fsup; \fsum^+)$-automata which is also an open problem.

\begin{rem}[Contrast with classical questions]%
\label{remark:Inf-classical-vs-probabilistic}
Consider Table~\ref{tab1} for the classical questions and
our results established in Table~\ref{tab:compInf} for probabilistic questions.
There are some contrasting results, such as, while for
$(\fsup,\fsum)$-automata the emptiness problem is undecidable,
the approximation problems are decidable.
\end{rem}

\begin{rem}[Contrast of $\fliminf$ vs $\finf$]%
\label{remark:LimInf-vs-Inf}
We remark on the contrast of the $\fliminf$ vs $\finf$ value functions.
For the classical questions of emptiness and universality, the complexity and
decidability always coincide for $\fliminf$ and $\finf$ value functions for NWA
(see Table~\ref{tab1}).
Surprisingly we establish that for probabilistic questions there is a substantial
complexity gap: while the $\fliminf$ problems can be solved in polynomial time, the
$\finf$ problems are undecidable, $\PSPACE$-hard, and even $\#P$-hard for approximation.
\end{rem}

\section{Results on non-deterministic automata}%
\label{s:nondeterminism}
In this section, we briefly discuss non-deterministic NWA evaluated on Markov chains.
First, we discuss the definition of random variables defined by non-deterministic NWA\@. Next, we present two negative results.

\smallskip\noindent{\em Non-deterministic NWA as random variables}.
A non-deterministic NWA $\nestedA$ defines a function $h \colon \Sigma^{\omega} \to \R$ as in the deterministic case via
$h(w) = \valueL{\nestedA}(w)$.
Recall that $\valueL{\nestedA}(w) = \inf_{\pi \in \Acc(w)} f(\pi)$, where $\Acc(w)$ is the set of accepting runs of $\nestedA$ on $w$.
We show that $h$ is measurable w.r.t.\ any probability measure given by a Markov chain.
It suffices to show that for every interval $\lopen{-\infty,x}$, its preimage $h^{-1}[\lopen{-\infty,x}]$ is measurable.
Consider a set $A_x \subseteq \Sigma^{\omega} \times \masterStates^{\omega} \times {(\Z \cup \{\bot\})}^{\omega}$ of words with runs of $\nestedA$ of
the value at most $x$. More precisely, let $f$ be the the master value function of $\nestedA$.
     Then, $(w,\masterRun,\alpha) \in A_x$ if and only if
(1)~$\masterRun$ is a run of the master automaton of $\nestedA$ on $w$,
(2)~for every $i \in \N$, the slave automaton invoked by the master automaton at position $i$ has a finite run on $w[i,\infty]$ of the value
$\alpha[i]$,
(3)~$f(\alpha) \leq x$.
Observe that $A_x$ is a Borel set in the product topology. The preimage $h^{-1}[\lopen{-\infty,x}]$ is
the projection of the Borel set $A_x$, and hence it is an \emph{analytic set}, which is measurable~\cite{kechris}.

\smallskip\noindent{\em Conceptual difficulty}.
The evaluation of a non-deterministic (even non-nested) weighted automaton over a Markov
chain is conceptually different as compared to the standard model of
Markov decision processes (MDPs).
Indeed, in an MDP, probabilistic transitions are interleaved with non-deterministic transitions,
whereas in the case of an automaton, it runs over a word that has been already generated by
the Markov chain.
In MDPs, the strategy to resolve non-determinism can only rely on the past, whereas
in the automaton model the whole future is available (i.e., there is a crucial distinction
between online vs offline processing of the word).
Below we present an example to illustrate
this conceptual problem.

\begin{exa}
Consider a non-deterministic $\flimavg$-automaton $\aut$, depicted in Figure~\ref{fig:nondet-vs-MDP}.
Intuitively, the automaton processes a given word in blocks of letters $a,b$ separated by letters $\#$.
At the beginning of every block it decides whether the value of this block is
the number of $a$ letters $n_a$ minus the number of $b$ letters $n_b$ divided by $n_a + n_b$
(i.e., $\frac{n_a - n_b}{n_a + n_b}$) or
the opposite (i.e., $\frac{n_b - n_a}{n_a + n_b}$).
\begin{figure}
\centering
\begin{tikzpicture}

\node[state,accepting] (Q0) at (0,0) {$q_0$};
\node[state,accepting] (Q1) at (3,0) {$q_1$};

\draw[->,loop left] (Q0) to node[left]{$(\#,0)$} (Q0);
\draw[->,loop above] (Q0) to node[above]{$(a,1)$} (Q0);
\draw[->,loop below] (Q0) to node[below]{$(b,-1)$} (Q0);

\draw[->,loop right] (Q1) to node[right]{$(\#,0)$} (Q1);
\draw[->,loop above] (Q1) to node[above]{$(a,-1)$} (Q1);
\draw[->,loop below] (Q1) to node[below]{$(b,1)$} (Q1);

\draw[->,bend left] (Q0) to node[above]{$(\#,0)$} (Q1);
\draw[->,bend left] (Q1) to node[below]{$(\#,0)$} (Q0);
\end{tikzpicture}
\caption{An example of non-deterministic automaton, in which non-deterministic choices has to ``depend on the future'' in order to obtain the infimum.}%
\label{fig:nondet-vs-MDP}
\vspace{-1em}
\end{figure}
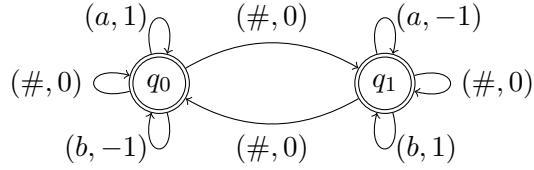
Let $\calU$ be the uniform distribution on infinite words over $\Sigma$.
Suppose that the expected value of $\aut$ w.r.t. $\calU$ is evaluated as in MDPs case, i.e.,
non-deterministic choices depend only on the read part of the word.
Then, since the distribution is uniform, any strategy results in the same expected value, which is equal to $0$.
Now, consider $\expected_{\calU}(\aut)$. The value of every block is at most $0$ as the automaton works over fully
generated words and at the beginning of each block can guess whether the number of $a$'s or $b$'s is greater.
Also, the blocks $a\#, b\#$ with the average  $-\frac{1}{2}$ appear with probability $\frac{2}{9}$, hence
$\expected_{\calU}(\aut) < -\frac{1}{9}$.
Thus, the result of evaluating a non-deterministic weighted automaton over a Markov chain is different than evaluating
it as an MDP\@.
\end{exa}

\smallskip\noindent{\em Non-deterministic $\flimavg$-automata under probabilistic semantics}.
Non-deterministic $\flimavg$-automata evaluated over Markov chains have been studied in~\cite{concur18}.
It has been established that the expected value or the distribution of such automata (over the uniform distribution) can be irrational and these values are not computable.
This is in stark contrast to deterministic $\flimavg$-automata, where the answers are rational and can be computed in polynomial time~\cite{BaierBook}.

\smallskip\noindent{\em Computational difficulty.}
In contrast to our polynomial-time algorithms for the probabilistic questions for
deterministic $(\flimsup,\fsum)$-automata, we establish the following
undecidability result for the non-deterministic automata with width~1.
Lemma~\ref{th:nondeterminism-is-hard} implies Theorem~\ref{c:nondet-undecidable}.

\begin{restatable}{lem}{NondeterminismLemma}%
\label{th:nondeterminism-is-hard}
The following problem is undecidable: given a non-deterministic $(\flimsup; \fsum)$-automaton $\nestedA_M$ of width $1$,
decide whether $\probability(\{ w \mid \valueL{\aut_M}(w) = 0 \}) = 1$ or $\probability(\{ w \mid \valueL{\aut_M}(w) = -1 \}) = 1$ w.r.t.
the uniform distribution on infinite words.
\end{restatable}
\begin{proof}
In the following, we discuss how to adapt the proof
of Theorem~\ref{th:undecidable-limsup} to prove this lemma.
First, observe that we can encode any alphabet $\Sigma$ using a two-letter alphabet $\{0,1\}$, therefore we will present our argument
for multiple-letters alphabet as it is more convenient.
Given a two-counter machine $\M$ we construct
 a non-deterministic $(\flimsup; \fsum)$-automaton $\nestedA_M$ such that $\valueL{\aut_M}(w) = 0$ if and only if $w$
contains infinitely many subsequences that correspond to valid accepting computations of $\M$.
As in the proof of Theorem~\ref{th:undecidable-limsup}, for every subsequence $\$ u \$$, where $u$ does not contain $\$$, we check whether 
$u$ is an encoding of a  valid accepting computation of $\M$.
To do that, we check conditions (C1) and (C2) as in the proof of Theorem~\ref{th:undecidable-limsup}.
At the letter $\$$, the master automaton non-deterministically decides whether $u$ violates (C1) or (C2) and 
either starts a slave automaton checking (C1) or (C2).
The slave automaton checking (C1) works as in the proof of Theorem~\ref{th:undecidable-limsup}.
It returns $-1$ if (C1) is violated and $0$ otherwise.
The slave automaton checking (C2) non-deterministically picks the position of invocation of a slave automaton
from the proof of Theorem~\ref{th:undecidable-limsup} that returns a negative value.
Finally, at the letter $\$$ following $u$, the master automaton starts the slave automaton that returns the value $-1$.
It follows that the supremum of all values of slave automata started at $u\$$ is either $-1$ or $0$.
By the construction, there is a (sub) run on $u\$$ such that the supremum of the values of all slave automata is
$-1$ if and only if $u$ does not encode valid accepting computation of $\M$.
Otherwise, this supremum is $0$.
Therefore, the value of the word $w$ is $0$ if and only if $w$ contains
infinitely many subsequences that correspond to valid accepting computations of $\M$.
Now, if $\M$ has at least one valid accepting computation $u$, then almost all words contain infinitely many
occurrences of $u$ and almost all words have value $0$. Otherwise, all words have value $-1$.
\end{proof}

Lemma~\ref{th:nondeterminism-is-hard} implies that there is no terminating Turing machine that computes any of probabilistic questions.

\begin{restatable}{thm}{NondeterministicUncomputable}%
\label{c:nondet-undecidable}
All probabilistic questions (Q1-Q5) are undecidable for non-deterministic
$(\flimsup,\fsum)$-automata of width~1.
\end{restatable}

\begin{table}
\centering
\begin{tabular}{|c|c|c|} 
\hline 
& Deterministic & Non-deterministic \\
\hline 
Emptiness & \multicolumn{2}{|c|}{Undecidable~(from~\cite{nested})} \\
\hline 
Probabilistic & \multirow{2}{*}{Polynomial time~(Theorem~\ref{th:compLimInf})} & \multirow{2}{*}{\uncomp~(Theorem~\ref{c:nondet-undecidable})} \\
questions & & \\
\hline 
\end{tabular}
\caption{Decidability and complexity status of the classical and probabilistic questions for $(\flimsup;\fsum)$-automata.
The negative results hold also for NWA of bounded width and automata with monitor counters.}%
\label{tab2}
\vspace{-1em}
\end{table}

\section{Discussion}%
\label{s:discussion}
In this section we discuss extensions of our main results.
We begin with the distribution question for NWA without the almost-sure acceptance condition.
Next, we show how to apply the obtained results to compute the probabilistic variant of the quantitative inclusion problem~\cite{Chatterjee08quantitativelanguages,nested}.
Finally, we discuss the parametric complexity of the questions considered in the paper.

\subsection{Non almost-sure acceptance}%
\label{s:no-almost-surely-accepting}
In Section~\ref{s:limit} we consider NWA that accept almost all words, i.e., they assign finite values to almost all words.
Dropping almost-sure acceptance condition does not change the complexity of the expected question, but it influences the complexity of the distribution question, which we discuss in this section.

We begin with the lower bounds for the distribution question for NWA without almost-sure acceptance condition.

\proofideas{} We show that removing the restriction that almost all words are accepted
changes the complexity of distribution questions.
The intuition behind this is that the condition ``all slave automata accept''
allows us to simulate (in a restricted way) the $\fsup$ master value function.

\begin{lem}
For all $f \in \InfVal$ and $g \in \FinVal$, we have
\begin{enumerate}
\item The distribution question for deterministic $(f;g)$-automata is $\PSPACE$-hard.
\item The approximate distribution question for deterministic $(f;g)$-automata is $\#P$-hard.
\end{enumerate}
\end{lem}
\begin{proof}
We say that an NWA is \emph{simple} if it is deterministic, accepts almost all words and its slave automata have only weights $0$ and $1$.
First, we show that the almost-sure distribution and the approximate distribution questions for simple $(\fsup;\fmax)$-automata are respectively $\PSPACE$-hard and
\#P-hard. Second, we reduce the almost-sure distribution and the approximate distribution problems for simple $(\fsup;\fmax)$-automata to
the corresponding problems for deterministic $(f;g)$-automata (which are not
almost-surely accepting). These two steps show (1) and (2).

We observe that $(\finf;\fmin)$-automata considered
in the proof of Lemma~\ref{l:hardness-for-det-inf} are simple; they are deterministic, accept almost all words and all slave automata have only weights $0,1$.
Given an $(\finf;\fmin)$-automaton $\nestedA$, let $\nestedA'$ be a $(\fsup;\fmax)$-automaton
obtained from $\nestedA$ by swapping weights $0$ and $1$ in all slave automata. Note that
for all words $w$, $\nestedA(w)$ returns $0$ (resp. $1$)  if and only if $\nestedA'$ returns $1$ (resp. $0$).
Therefore, $\distrib_{\calU,\nestedA}(0) = 1 - \distrib_{\calU,\nestedA'}(0)$.
It follows that the almost-sure distribution and the approximate distribution questions for simple
$(\fsup;\fmax)$-automata are respectively  $\PSPACE$-hard and \#P-hard.

Now, we show reduction from simple $(\fsup;\fmax)$-automata  to deterministic $(f;g)$-automata (which are not almost-surely accepting).
For a deterministic $\fmax$-automaton $\aut$ with only weights $0$ and $1$, we define $\aut^g$ as
a deterministic $g$-automaton obtained from $\aut$ by deletion of transitions of weight $1$.
Observe that $\aut^g$ returns $0$ whenever $\aut$ returns $0$, and it rejects whenever $\aut$ rejects or returns $1$.
This construction works for all $g$, which return $0$ on any sequence of $0$'s.
All value functions $g$ from $\FinVal$ have this property.
Now, consider any $f \in \InfVal$.
Given a simple $(\fsup;\fmax)$-automaton $\nestedA$,
 we apply to all slave automata $\slaveA$ of $\nestedA$ the transformation $\slaveA \to \slaveA^g$.
Let $\nestedA^f$ be the resulting $(f;g)$-automaton.
This NWA is deterministic and its slave automata return only value $0$ or reject (return $\infty$).
Therefore, for every $f \in \InfVal$ and every word $w$, we have
(a)~$\valueL{\nestedA^f}(w) = 0$ if and only if $\valueL{\nestedA}(w) = 0$, and
(b)~$\valueL{\nestedA^f}(w) = \infty$ if and only if $\valueL{\nestedA}(w) \in \{1,\infty\}$.
It follows that $\distrib_{\calU,\nestedA^f}(0)=\distrib_{\calU,\nestedA}(0)$.
Therefore, the almost-sure distribution (resp., the approximate distribution) problem for simple $(\fsup;\fmax)$-automata
reduces to the almost-sure distribution (resp., the approximate distribution) problem for deterministic $(f;g)$-automata.
\end{proof}

We now show the upper bound for the distribution question for NWA considered in Section~\ref{s:limit}.

\begin{restatable}{lem}{DistribNotUnicversal}%
\label{l:notUniversalDistrib}
Let $f \in \{\fliminf, \flimsup, \flimavg\}$ and let $g \in \FinVal$ be a value function.
Given a Markov chain $\markov$, a deterministic $(f;g)$-automaton $\nestedA$, and a
threshold $\const$ in binary, the value $\distrib_{\markov,\nestedA}(\const)$
can be computed in exponential time in $\nestedA$ and polynomial time in $\markov$.
Moreover, if $\nestedA$ has bounded width, then the above quantities can be computed in polynomial time.
\end{restatable}
\begin{proof}
Consider $\masterA \times \markov$, where $\masterA$ is the master automaton of $\nestedA$.
For all bottom SCCs of  $\masterA \times \markov$ either almost all words have an accepting run or all words are rejected (Lemma~\ref{l:in-scc-all-equal} --- acceptance is independent of the value function and hence it works for $\flimavg$ as well).
Moreover, in a bottom SCC in which almost all words are accepted, almost all words have the same value (Lemma~\ref{l:in-scc-all-equal} and Lemma~\ref{l:limavg-dist-poly}).
We compute these values in all almost-accepting bottom SCCs.
Next, we need to compute the probability of the set of words, which are accepted by $\nestedA$ and reach any almost-accepting bottom SCC, in which almost all words have the value not exceeding $\lambda$.

To do that, we consider $\nestedA$ as an $(\finf;\fBsum{B})$-automaton, transform it into an $\finf$-automaton $\nonnestedA$ (Lemma~\ref{l:bsum-to-inf}).
Then, we compute the probability of reaching the the corresponding SCCs in $\nonnestedA \times \markov$ with the standard reachability analysis~\cite{BaierBook}.
Recall that each state of $\nonnestedA$ contains a state of $\masterA$ as a component and hence we can identify states of
$\nonnestedA \times \markov$ that corresponds to the selected SCC of  $\masterA \times \markov$.
The size of $\nonnestedA \times \markov$ is exponential in $\nestedA$ and polynomial in $\markov$.
Moreover, if $\nestedA$ has bounded width, then
$\nonnestedA \times \markov$  is polynomial in $|\nestedA|$.
\end{proof}

\subsection{Quantitative inclusion}
In this section, we discuss \emph{probabilistic inclusion question}, which is the probabilistic variant of the quantitative inclusion problem for weighted automata.
We show how to adapt the results from Section~\ref{s:limit} and Section~\ref{s:nolimit} to establish decidability and complexity of the probabilistic inclusion question.
\smallskip

The following definition is common for weighted automata, automata with monitor counters and NWA, and hence we refer to them collectively as automata.
The \emph{probabilistic inclusion question} asks, given two automata $\aut_1, \aut_2$ and a Markov chain $\markov$, to compute the probability of the set
$\{ w \mid \valueL{\aut_1}(w) \leq \valueL{\aut_2}(w) \}$ w.r.t.\ the probability measure given by $\markov$.
\smallskip

\proofideas{}
Let $\markov$ be a Markov chain, $f \in \{\fliminf, \flimsup, \flimavg\}$ and let $\nestedA_1, \nestedA_2$ be deterministic $(f;\fsum)$-automata accepting almost all words.
For all such NWA, almost all words reaching the same bottom SSC (of the product $\masterA \times \markov$) have the same value (Section~\ref{s:limit}), and hence
to decide probabilistic inclusion we examine all pairs of bottom SSCs $B_1, B_2$ from $\nestedA_1$ and $\nestedA_2$ respectively, in which the value of $\nestedA_1$ does not exceed the value of $\nestedA_2$
and compute the probability of words that lead to $B_1$ in $\nestedA_1$ and $B_2$ in $\nestedA_2$. The sum of such probabilities is the answer to the quantitative inclusion problem.
These probabilities can be computed in polynomial time~\cite{BaierBook}.

For $f \in \{\finf,\fsup\}$ and $g \in \{\fmin, \fmax, \fBsum{B} \}$, deterministic $(f;g)$-automata $\nestedA_1, \nestedA_2$ are equivalent to deterministic $f$-automata $\nonnestedA_1, \nonnestedA_2$ respectively.
The automata $\nonnestedA_1, \nonnestedA_2$ have exponential size in $\nestedA_1$ and $\nestedA_2$ respectively.
The probabilistic inclusion question for deterministic $\finf$-automata (resp., $\fsup$-automata) can be computed in polynomial time, which gives us the exponential-time upper bound in $\nestedA_1$ and $\nestedA_2$.
Finally, the probabilistic inclusion question subsumes the distribution question, which gives us lower bounds.

\begin{thm}
The following conditions hold:
\begin{enumerate}
\item For all $f \in \{\fliminf, \flimsup, \flimavg\}$ and $g \in \FinVal$, the probabilistic inclusion question for deterministic almost-surely accepting $(f;g)$-automata can be solved in polynomial time.
\item For all $f \in \{\finf, \fsup\}$ and $g \in \{ \fmin, \fmax, \fBsum{B} \}$, the probabilistic inclusion question for deterministic $(f;g)$-automata is $\PSPACE$-hard and can be solved in
exponential time.
\item For all $f \in \{\finf, \fsup,\}$, the probabilistic inclusion question for deterministic $(f;\fsum)$-automata is uncomputable.
\end{enumerate}
\end{thm}

\begin{proof}
\Paragraph{$\fliminf$ and $\flimsup$ and $\flimavg$ value functions}.
We discuss the case of $g = \fsum$ as it subsumes other value functions from $\FinVal$.
Let $\markov$ be a Markov chain, $f \in \{\fliminf, \flimsup, \flimavg\}$ and let $\nestedA_1, \nestedA_2$ be deterministic $(f;\fsum)$-automata accepting almost all words.
Let $\masterA^1, \masterA^2$ be the master automata of respectively $\nestedA_1, \nestedA_2$.
We construct the product Markov chain $(\masterA^1 \times \masterA^2) \times \markov$ and compute all its bottom SCCs $S_1, \ldots, S_k$.
For almost all words $w$, both runs of respectively $\nestedA_1, \nestedA_2$ on $w$ finally reach one of these SCCs.
Note that  each $S_i$ projected to  $\masterA^1 \times \markov$ or $\masterA^2 \times \markov$ is still an SCC and hence almost all words $w$, whose run end up in $S_i$
have the same value in $\nestedA_1$  (resp., in $\nestedA_2$).
We compute these value for all SCCs  $S_1, \ldots, S_k$ and select these components, in which the value in $\nestedA_1$ does not exceed the value in $\nestedA_2$.
This can be done in polynomial time for $f \in \{\fliminf, \flimsup\}$ (Lemma~\ref{l:in-scc-all-equal}) and for $f = \flimavg$ (Lemma~\ref{l:limavg-dist-poly}).
Finally, we compute the probability of reaching the selected SCCs in  $(\masterA^1 \times \masterA^2)\times \markov$, which can be done in polynomial time~\cite{BaierBook}.
\smallskip

\Paragraph{$\finf$ and $\fsup$ value functions}.
We consider the case of $f = \finf$ as the case of $f = \fsup$ is symmetric.
First, consider the case of  $g \in \{\fmin, \fmax, \fBsum{B} \}$ and consider two deterministic $(\finf;g)$-automata $\nestedA_1, \nestedA_2$.
These NWA can be transformed to equivalent deterministic $\finf$-automata $\nonnestedA_1, \nonnestedA_2$ respectively, which have exponential size in $\nestedA_1$ and $\nestedA_2$ respectively (Lemma~\ref{l:bsum-to-inf}).
Now, let $x_1, \ldots x_n$ be weights of $\nonnestedA_1$. For each of these weights $x_i$, we compute the probability $p_i$ of the set of words such that the value of $\nonnestedA_1$ is $x_i$ and
the value of $\nonnestedA_2$ is at least $x_i$. To compute $p_i$, we construct the product Markov chain $\markov \times \masterA^1 \times \masterA^2$, remove from it
all transitions, which correspond to transition of $\nonnestedA_1$ or $\nonnestedA_2$ of weight less than $x_i$, and compute the probability of the set of paths which take at least once a transition of
$\nonnestedA_1$ of weight $x_i$. Finally, the answer to the probabilistic inclusion question for $\nonnestedA_1$ and $\nonnestedA_2$, and hence  $\nestedA_1$ and $\nestedA_2$, is the sum of $p_i$.

Observe that for $\nestedA_2$ that returns $\lambda$ for every word, the probabilistic inclusion problem becomes the distribution question for $\nestedA_1$.
\end{proof}

\subsection{Parametric complexity}
The problems we consider correspond to measuring performance (expectation or cumulative distribution) under stochastic environments, when
the specification is an NWA and the system is modeled by a Markov chain.
In this section we discuss the parametric complexity of the probabilistic problems, where the specification, represented by an NWA, is fixed.
We discuss the parametric complexity for different value functions for the master automaton:

\begin{itemize}
\item For $\fliminf, \flimsup, \flimavg$ value functions for the master automaton the expected question  solvable in polynomial time in $\nestedA$ and $\markov$ (Theorem~\ref{th:compLimInf} and Theorem~\ref{th:compLimAvg}).
The distribution question is solvable in polynomial time (in $\nestedA$ and $\markov$) as well,
but only for NWA that are almost-surely accepting (Theorem~\ref{th:compLimInf} and Theorem~\ref{th:compLimAvg}).
For NWA that are not almost-surely accepting Lemma~\ref{l:notUniversalDistrib} states that the distribution question can be solved in polynomial time in $|\markov|$.

\item For $f \in \{\finf,\fsup\}$ and $g \in \{\fmin, \fmax, \fBsum{B}\}$, a deterministic $(f;g)$-automaton $\nestedA$ can be transformed to an equivalent  deterministic $f$-automaton $\nonnestedA$ (Lemma~\ref{l:bsum-to-inf}).
The size of  $\nonnestedA$ is exponential in $\nestedA$. However, if $\nestedA$ is fixed, the complexity of all probabilistic questions is the same as for weighted automata; they can be solved in polynomial time
(Lemma~\ref{t:weighted-inf-expected}).

\item For $f \in \{\finf,\fsup\}$ and $g \in \{\fsum, \fsum^+\}$, the approximate expected question and the approximate distribution questions for $(f;g)$-automata
can be solved in polynomial time in $\markov$ and the length of the binary representation of $\epsilon$ (Lemma~\ref{l:infSumSolution}).

\item For $f \in \{\finf,\fsup\}$, the distribution question for $(f;\fsum^+)$-automata can be solved in polynomial time in $|\markov|$ (Lemma~\ref{l:bound-from-below} and Lemma~\ref{l:distributionSumPlusDecidable}).

\item Finally, for deterministic $(\finf;\fsum)$-automata and $(\fsup;\fsum)$-automata, surprisingly the expected question, the distribution question and the almost-sure distribution question remain uncomputable. We sketch this in the following Lemma~\ref{l:parametric-uncomputable}.
\end{itemize}

\noindent
In summary, in all computable cases fixing the NWA makes the complexity of all probabilistic problems drop to polynomial time.
Interestingly, the uncomputable cases remain uncomputable.
This is similar to the parametric complexity analysis from~\cite{nested}, where fixing the size of slave automata reduces the complexity of the classic decision questions for NWA, but the undecidable cases remain undecidable.

\proofideas{}
We modify the construction from the proof of Lemma~\ref{l:inf-prob-undec}, where given a Minsky machine $\M$,
we construct an $(\finf;\fsum)$-automaton $\nestedA$ that returns $0$ if the input word $w$ is of the form $u \$ w'$
and $u$ encodes an accepting computation of $\M$. Otherwise it returns negative values.
The constructed NWA checks two types of conditions:
(C1)~Boolean conditions stating that the sequence of configurations is consistent with instructions of $\M$, and
(C2)~quantitative conditions, which imply that if the counters values are inconsistent with increments and decrements, the NWA returns negative values.
We observe that (C2) are independent of $\M$, while conditions (C1) are Boolean and can be checked with a finite automaton, or they can be enforced by the Markov chain.
Thus, given a Minsky machine $\M$, we construct a Markov chain $\markov_{\M}$ checking (C1), while the NWA $\nestedA'$ that checks (C2) is independent of $\M$, and hence can be fixed.

\begin{lem}%
\label{l:parametric-uncomputable}
The following conditions hold:
\begin{enumerate}
\item There exists a deterministic $(\finf;\fsum)$-automaton $\nestedA$ such that the problem: given a Markov chain $\markov$, decide whether $\distrib_{\nestedA, \markov}(-1) = 1$ is undecidable.
\item There exist  deterministic  $(\finf;\fsum)$-automata $\nestedA_1, \nestedA_2$ such that the problem: given a Markov chain $\markov$, decide whether $\expected_{\markov}(\nestedA_1) = \expected_{\markov}(\nestedA_2)$ is undecidable.
\end{enumerate}
\end{lem}
\begin{proof}
Let $\Sigma = \{ 0,1,2,\#,\$\}$.
We construct an NWA working over words of the form $u \$ w'$, where $u \in {(0^* 1^* 2^* \#)}^*$.
Each block $0^n 1^k 2^j$ represents a configuration of some Minsky machine such that the machine is in the state $q_n$ the value of the first counter is $k$ and the value of the second counter is $j$.
We define $(\finf;\fsum)$-automaton $\nestedA$ that for each counter invokes two slave automata, which respectively compute the difference between two consecutive values of the same counter plus $1$ and its inverse.
Moreover, $\nestedA$ invokes one slave automaton returning value $0$ to ensure that the value of each word is at most $0$.
Thus, $\nestedA$ returns $0$ on an input word $u \$ w'$ if the counter values in all blocks differ by at most $1$, i.e., for any infix $\# 0^{n_1} 1^{k_1} 2^{l_1} \# 0^{n_2} 1^{k_2} 2^{l_2} \#$ we have $|k_1 - k_2 | \leq 1$ and $|l_1 - l_2| \leq 1$.
On all other words it returns values less or equal to $-1$.

Now, given a Minsky machine $\M$ we define a Markov chain $\markov$ that generates sequences of the form $\# 0^{n[0]} 1^{k[0]} 2^{j[0]} \# 0^{n[1]} 1^{k[1]} 2^{j[1]} \ldots \$ \Sigma^{\omega}$ consistent with the instructions from $\M$.
More precisely, using Boolean conditions we specify that
(i)~states encoded by $0^{n}$ change according to the instructions of $\M$ (we can verify zero and non-zero tests),
(ii)~the first configuration and the configuration before $\$$ are respectively initial and the final configuration of $\M$, and
(iii)~values of counters modulo $3$ change according to the instructions of $\M$.
Observe that if there is a word $u \$ w'$ generated by $\markov$ has value $0$ assigned by $\nestedA$, then the values of counters in $u$ change by at most $1$ and the change modulo $3$ is verified in condition (iii).
Both conditions imply that the counters change according to instructions of $\M$.
Therefore, the word $u$ encodes an accepting computation of $\M$.
It follows that for the NWA $\nestedA$, the problem given a Markov chain $\markov$ decide whether $\distrib_{\nestedA, \markov}(-1) = 1$ is undecidable.
In consequence, the expected and the distribution questions are undecidable (see Lemma~\ref{l:inf-prob-undec}).

For the expected value, we reduce the distribution question to the equality of the expected values as in the proof of (2) from Lemma~\ref{l:inf-prob-undec}.
\end{proof}

\section{Conclusions}
In this work we study the probabilistic questions related to NWA and
automata with monitor counters.
We establish the relationship between NWA and automata with monitor counters,
and present a complete picture of decidability for all the
probabilistic questions we consider.
Our results establish a sharp contrast of the decidability and complexity
of the classical questions (of emptiness and universality) and
the probabilistic questions for deterministic automata
(see Tables~\ref{tab1},~\ref{tab:compInf} and Theorems~\ref{th:compLimInf},~\ref{th:compLimAvg}).
In addition, there is also a sharp contrast for deterministic and
non-deterministic automata.
For example, for $(\flimsup,\fsum)$-automata, the classical questions are
undecidable for deterministic and non-deterministic automata, while the
probabilistic questions are decidable for deterministic automata, but remain
undecidable for non-deterministic automata (see Table~\ref{tab2}).
We have some complexity gap (e.g., $\EXPTIME$ vs $\PSPACE$) which is due to the fact
that the computational questions we consider for Markov chains are in $\PTIME$
(as compared to $\NLOGSPACE$ for graphs), and we need to evaluate exponential-size
Markov chains. Closing the complexity gap is an interesting open question.

\section{Acknowledgments}
{
This research was funded in part by the European Research Council (ERC) under grant
agreement 267989 (QUAREM), by the Austrian Science Fund (FWF) projects S11402-N23 (RiSE) and Z211-N23 (Wittgenstein Award),
FWF Grant No P23499- N23, FWF NFN Grant No S11407-N23 (RiSE/SHiNE),
ERC Start grant (279307: Graph Games), Vienna Science and Technology Fund (WWTF) through project ICT15--003 and
by the National Science Centre (NCN), Poland under grant 2014/15/D/ST6/04543.}

\bibliographystyle{alpha}
\bibliography{papers}

\newcommand{\etalchar}[1]{$^{#1}$}
\begin{thebibliography}{BKKW14}

\bibitem[ABK14]{AlmagorBK14}
Shaull Almagor, Udi Boker, and Orna Kupferman.
\newblock Discounting in {LTL}.
\newblock In {\em {TACAS}, 2014}, pages 424--439, 2014.

\bibitem[ADD{\etalchar{+}}13]{DBLP:conf/lics/AlurDDRY13}
Rajeev Alur, Loris D'Antoni, Jyotirmoy~V. Deshmukh, Mukund Raghothaman, and
  Yifei Yuan.
\newblock Regular functions and cost register automata.
\newblock In {\em LICS 2013}, pages 13--22, 2013.

\bibitem[BBC{\etalchar{+}}11]{BBCFK11}
Tom{\'{a}}s Br{\'{a}}zdil, V{\'{a}}clav Brozek, Krishnendu Chatterjee, Vojtech
  Forejt, and Anton{\'{\i}}n Kucera.
\newblock Two views on multiple mean-payoff objectives in {M}arkov decision
  processes.
\newblock In {\em {LICS} 2011}, pages 33--42, 2011.

\bibitem[BCFK15]{Forejt}
Tom{\'{a}}s Br{\'{a}}zdil, Krishnendu Chatterjee, Vojtech Forejt, and
  Anton{\'{\i}}n Kucera.
\newblock Multigain: {A} controller synthesis tool for {MDPs} with multiple
  mean-payoff objectives.
\newblock In {\em {TACAS} 2015}, pages 181--187, 2015.

\bibitem[BCHK14]{BokerCHK14}
Udi Boker, Krishnendu Chatterjee, Thomas~A. Henzinger, and Orna Kupferman.
\newblock Temporal specifications with accumulative values.
\newblock {\em {ACM} {TOCL}}, 15(4):27:1--27:25, 2014.

\bibitem[BDK14]{Baier-CSL-LICS-1}
Christel Baier, Clemens Dubslaff, and Sascha Kl{\"{u}}ppelholz.
\newblock Trade-off analysis meets probabilistic model checking.
\newblock In {\em {CSL}-{LICS} 2014}, pages 1:1--1:10, 2014.

\bibitem[BGMZ10]{bollig2010pebble}
Benedikt Bollig, Paul Gastin, Benjamin Monmege, and Marc Zeitoun.
\newblock Pebble weighted automata and transitive closure logics.
\newblock In {\em {ICALP} 2010, Part {II}}, pages 587--598. Springer, 2010.

\bibitem[BK08]{BaierBook}
Christel Baier and Joost{-}Pieter Katoen.
\newblock {\em Principles of model checking}.
\newblock {MIT} Press, 2008.

\bibitem[BKKW14]{Baier-CSL-LICS-2}
Christel Baier, Joachim Klein, Sascha Kl{\"{u}}ppelholz, and Sascha Wunderlich.
\newblock Weight monitoring with linear temporal logic: complexity and
  decidability.
\newblock In {\em {CSL}-{LICS} 2014}, pages 11:1--11:10, 2014.

\bibitem[BMM14]{BouyerMM14}
Patricia Bouyer, Nicolas Markey, and Raj~Mohan Matteplackel.
\newblock Averaging in {LTL}.
\newblock In {\em {CONCUR} 2014}, pages 266--280, 2014.

\bibitem[BMR{\etalchar{+}}18]{DBLP:journals/acta/BouyerMRLL18}
Patricia Bouyer, Nicolas Markey, Mickael Randour, Kim~G. Larsen, and Simon
  Laursen.
\newblock Average-energy games.
\newblock {\em Acta Inf.}, 55(2):91--127, 2018.

\bibitem[CD11]{CD11}
Krishnendu Chatterjee and Laurent Doyen.
\newblock Energy and mean-payoff parity {M}arkov {D}ecision {P}rocesses.
\newblock In {\em {MFCS} 2011}, pages 206--218, 2011.

\bibitem[CDH09a]{Chatterjee:2009:AWA:1789494.1789497}
Krishnendu Chatterjee, Laurent Doyen, and Thomas~A. Henzinger.
\newblock Alternating weighted automata.
\newblock In {\em FCT'09}, pages 3--13. Springer, 2009.

\bibitem[CDH09b]{ChatterjeeDH09LimInf}
Krishnendu Chatterjee, Laurent Doyen, and Thomas~A. Henzinger.
\newblock A survey of stochastic games with limsup and liminf objectives.
\newblock In {\em {ICALP} 2009, Part {II}}, pages 1--15, 2009.

\bibitem[CDH10a]{DBLP:journals/corr/abs-1007-4018}
Krishnendu Chatterjee, Laurent Doyen, and Thomas~A. Henzinger.
\newblock Expressiveness and closure properties for quantitative languages.
\newblock {\em LMCS}, 6(3), 2010.

\bibitem[CDH10b]{Chatterjee08quantitativelanguages}
Krishnendu Chatterjee, Laurent Doyen, and Thomas~A. Henzinger.
\newblock Quantitative languages.
\newblock {\em ACM TOCL}, 11(4):23, 2010.

\bibitem[CFW13]{CFW13}
Krishnendu Chatterjee, Vojtech Forejt, and Dominik Wojtczak.
\newblock Multi-objective discounted reward verification in graphs and {MDPs}.
\newblock In {\em {LPAR}}, pages 228--242, 2013.

\bibitem[Cha07]{Cha07}
Krishnendu Chatterjee.
\newblock Markov decision processes with multiple long-run average objectives.
\newblock In {\em FSTTCS}, pages 473--484, 2007.

\bibitem[CHJS15]{probabilisticMeasuriung}
Krishnendu Chatterjee, Thomas~A. Henzinger, Barbara Jobstmann, and Rohit Singh.
\newblock Measuring and synthesizing systems in probabilistic environments.
\newblock {\em J. {ACM}}, 62(1):9:1--9:34, 2015.

\bibitem[CHO16a]{nwa-mfcs}
Krishnendu Chatterjee, Thomas~A. Henzinger, and Jan Otop.
\newblock Nested weighted limit-average automata of bounded width.
\newblock In {\em {MFCS} 2016}, pages 24:1--24:14, 2016.

\bibitem[CHO16b]{conferenceVersion}
Krishnendu Chatterjee, Thomas~A. Henzinger, and Jan Otop.
\newblock Quantitative automata under probabilistic semantics.
\newblock In {\em {LICS} 2016}, pages 76--85, 2016.

\bibitem[CHO16c]{DBLP:conf/sas/ChatterjeeHO16}
Krishnendu Chatterjee, Thomas~A. Henzinger, and Jan Otop.
\newblock Quantitative monitor automata.
\newblock In {\em {SAS} 2016}, pages 23--38, 2016.

\bibitem[CHO17]{nested}
Krishnendu Chatterjee, Thomas~A. Henzinger, and Jan Otop.
\newblock Nested weighted automata.
\newblock {\em {ACM} Trans. Comput. Log.}, 18(4):31:1--31:44, 2017.

\bibitem[CKK15]{CKK15}
Krishnendu Chatterjee, Zuzana Kom{\'{a}}rkov{\'{a}}, and Jan
  Kret{\'{\i}}nsk{\'{y}}.
\newblock Unifying two views on multiple mean-payoff objectives in {M}arkov
  {D}ecision {P}rocesses.
\newblock In {\em {LICS} 2015}, pages 244--256, 2015.

\bibitem[CMH06]{CMH06}
Krishnendu Chatterjee, Rupak Majumdar, and Thomas~A. Henzinger.
\newblock {M}arkov {D}ecision {P}rocesses with multiple objectives.
\newblock In {\em {STACS} 2006}, pages 325--336, 2006.

\bibitem[DKV09]{Droste:2009:HWA:1667106}
Manfred Droste, Werner Kuich, and Heiko Vogler.
\newblock {\em Handbook of Weighted Automata}.
\newblock Springer, 1st edition, 2009.

\bibitem[DR06]{DrosteR06}
Manfred Droste and George Rahonis.
\newblock Weighted automata and weighted logics on infinite words.
\newblock In {\em {DLT} 2006}, pages 49--58, 2006.

\bibitem[Fel71]{feller}
W.~Feller.
\newblock {\em An introduction to probability theory and its applications}.
\newblock Wiley, 1971.

\bibitem[FKN{\etalchar{+}}11]{FKN11}
Vojtech Forejt, Marta~Z. Kwiatkowska, Gethin Norman, David Parker, and Hongyang
  Qu.
\newblock Quantitative multi-objective verification for probabilistic systems.
\newblock In {\em TACAS}, pages 112--127, 2011.

\bibitem[FV96]{filar}
Jerzy Filar and Koos Vrieze.
\newblock {\em Competitive Markov decision processes}.
\newblock Springer, 1996.

\bibitem[HKNP06]{PRISM}
Andrew Hinton, Marta~Z. Kwiatkowska, Gethin Norman, and David Parker.
\newblock {PRISM:} {A} tool for automatic verification of probabilistic
  systems.
\newblock In {\em {TACAS} 2006}, pages 441--444, 2006.

\bibitem[Kec12]{kechris}
Alexander Kechris.
\newblock {\em Classical descriptive set theory}, volume 156.
\newblock Springer Science \& Business Media, 2012.

\bibitem[Min61]{minsky1961recursive}
Marvin Minsky.
\newblock Recursive unsolvability of {P}ost's problem of "tag" and other topics
  in theory of {T}uring machines.
\newblock {\em Annals of Mathematics}, pages 437--455, 1961.

\bibitem[MO18]{concur18}
Jakub Michaliszyn and Jan Otop.
\newblock Non-deterministic weighted automata on random words.
\newblock In {\em {CONCUR} 2018}, pages 10:1--10:16, 2018.

\bibitem[Moh02]{DBLP:journals/jalc/Mohri02}
Mehryar Mohri.
\newblock Semiring frameworks and algorithms for shortest-distance problems.
\newblock {\em J. Aut. Lang. \& Comb.}, 7(3):321--350, 2002.

\bibitem[Pap03]{papadimitriou2003computational}
Christos~H Papadimitriou.
\newblock {\em Computational complexity}.
\newblock Wiley, 2003.

\bibitem[Put94]{Puterman}
Martin~L. Puterman.
\newblock {\em {M}arkov {D}ecision {P}rocesses: Discrete Stochastic Dynamic
  Programming}.
\newblock Wiley, 1st edition, 1994.

\bibitem[Val79]{valiant1979complexity}
Leslie~G. Valiant.
\newblock The complexity of computing the permanent.
\newblock {\em Theoretical computer science}, 8(2):189--201, 1979.

\end{thebibliography}

\end{document}